%% file: main.tex
\documentclass[11pt]{article} 

\pdfoutput=1

\usepackage[margin=1in]{geometry}
\usepackage{subfigure}
\usepackage{amssymb}
\usepackage{amsmath}
\usepackage{amsthm}
\usepackage{mathtools}
\usepackage{setspace}
\usepackage{sectsty}
\usepackage{accents}
\usepackage{comment}
\usepackage{enumitem}
\usepackage{pdflscape}

\usepackage[utf8]{inputenc}
\usepackage{uniinput}
\usepackage{changepage,cite}
\usepackage[plain]{fancyref}

\usepackage{makecell} %
\usepackage{longtable} %
\usepackage[normalem]{ulem} %
\usepackage{transparent} %
\usepackage{bm} %

\usepackage[usenames]{xcolor}
\usepackage{graphicx}
\usepackage{setspace}
\usepackage{float}
\usepackage{booktabs}
\usepackage{caption}
\usepackage[vcentermath]{youngtab}

\usepackage{dynkin-diagrams}

\usepackage[debug,pageanchor=false]{hyperref}
\definecolor{link}{rgb}{.8,.15,.1}
\definecolor{pigment}{rgb}{0.36, 0.54, 0.66}
\definecolor{pigment2}{rgb}{0.19, 0.55, 0.91}
\definecolor{pigment3}{rgb}{0.2, 0.2, 0.6}
\definecolor{light-gray}{gray}{0.75}
\hypersetup{colorlinks=true,linkcolor=link,citecolor=link,urlcolor=link,linktocpage}
\usepackage{cleveref}

\usepackage{array,multirow,booktabs,longtable}
\usepackage{mathrsfs}
\usepackage{tikz-cd} 
\usetikzlibrary{backgrounds, arrows,calc,shapes,decorations.pathreplacing, decorations.markings, automata,positioning,matrix}

\allowdisplaybreaks
\renewcommand{\hat}{\widehat}
\newtagform{normalsize}[\normalsize]{\normalsize(}{\normalsize)} %

\newcommand\vertarrowbox[3][6ex]{%
  \begin{array}[t]{@{} c c @{}} #2 \\
  \left\downarrow\vcenter{\hrule height #1}\right.\kern-\nulldelimiterspace & #3
  \end{array}%
}

\tikzset{
        cvertex/.style={circle,draw=black,inner sep=1pt,outer sep=3pt},
        vertex/.style={circle,fill=black,inner sep=1pt,outer sep=3pt},
        star/.style={circle,fill=yellow,inner sep=0.75pt,outer sep=0.75pt},
        tvertex/.style={inner sep=1pt,font=\scriptsize},
        gap/.style={inner sep=0.5pt,fill=white}}

\tikzstyle{mybox} = [draw=black, fill=blue!10, very thick,
    rectangle, rounded corners, inner sep=10pt, inner ysep=20pt]
\tikzstyle{boxtitle} =[fill=blue!50, text=white,rectangle,rounded corners]

 \newtheorem{theorem}{Theorem}[section]
 \newtheorem{lemma}[theorem]{Lemma}
\newtheorem{corollary}[theorem]{Corollary}

\theoremstyle{definition}
 \newtheorem{example}[theorem]{Example}

\theoremstyle{remark}

\def\node#1#2{\overset{#1}{\underset{#2}{\circ}}}
\def\ver#1#2{\overset{{\llap{$\scriptstyle#1$}\displaystyle\circ{\rlap{$\scriptstyle#2$}}}}{\scriptstyle\vert}}

\usetikzlibrary{arrows}
\tikzstyle{every picture}+=[remember picture]
\tikzstyle{na} = [baseline=-.5ex]
\tikzstyle{mine}= [arrows={angle 90}-{angle 90},thick]

\def\Llleftarrow{%
\lower2pt\hbox{\begingroup
\tikz
\draw[shorten >=0pt,shorten <=0pt] (0,3pt) -- ++(-1em,0) (0,1pt) -- ++(-1em-1pt,0) (0,-1pt) -- ++(-1em-1pt,0) (0,-3pt) -- ++(-1em,0) (-1em+1pt,5pt) to[out=-105,in=45] (-1em-2pt,0) to[out=-45,in=105] (-1em+1pt,-5pt);
\endgroup}
}

\newcommand{\rr}{\mathbb{R}}
\newcommand{\cc}{\mathbb{C}}
\newcommand{\zz}{\mathbb{Z}}

\renewcommand{\Re}{\mathrm{Re}}
\renewcommand{\Im}{\mathrm{Im}}

\DeclareMathOperator{\SU}{SU}

\DeclareMathOperator{\U}{U}
\DeclareMathOperator{\SO}{SO}
\DeclareMathOperator{\GL}{GL}

\DeclareMathOperator{\Hom}{Hom}

\DeclareMathOperator{\depth}{depth}

\DeclareMathOperator{\Sym}{Sym}
\DeclareMathOperator{\Gr}{Gr}
\DeclareMathOperator{\PU}{PU}

\newcommand{\todo}[1]{}
\renewcommand{\todo}[1]{{\color{red} TODO: {#1}}}
\newcommand{\red}[1]{}
\renewcommand{\red}[1]{{\color{red} {#1}}}
\definecolor{navyblue}{rgb}{0, 0,0.5}
\newcommand{\navyblue}[1]{}
\renewcommand{\navyblue}[1]{{\color{navyblue} {#1}}}
\newcommand{\blue}[1]{}
\renewcommand{\blue}[1]{{\color{blue} {#1}}}
	\definecolor{nnyellow}{rgb}{1,0.8, 0}
\newcommand{\yellow}[1]{}
\renewcommand{\yellow}[1]{{\color{yellow} {#1}}}
\newcommand{\nnyellow}[1]{}
\renewcommand{\nnyellow}[1]{{\color{nnyellow} {#1}}}
\newcommand{\magenta}[1]{}
\renewcommand{\magenta}[1]{{\color{magenta} {#1}}}
\definecolor{ao}{rgb}{0,0.5,0}
\newcommand{\ao}[1]{}
\renewcommand{\ao}[1]{{\color{ao} {#1}}}
\definecolor{aquamarine}{rgb}{0.5, 1.0, 0.83}
\newcommand{\aquamarine}[1]{}
\renewcommand{\aquamarine}[1]{{\color{aquamarine} {#1}}}
\definecolor{azure}{rgb}{0.0, 0.5, 1.0}
\newcommand{\azure}[1]{}
\renewcommand{\azure}[1]{{\color{azure} {#1}}}
\definecolor{arsenic}{rgb}{0.23, 0.27, 0.29}
\newcommand{\arsenic}[1]{}
\renewcommand{\arsenic}[1]{{\color{arsenic} {#1}}}

\definecolor{orange}{rgb}{1.0, 0.5, 0.0}
\newcommand{\orange}[1]{}
\renewcommand{\orange}[1]{{\color{orange} {#1}}}

\newcommand{\su}[1]{}
\renewcommand{\su}[1]{{\mathfrak{su}({#1})}}
\newcommand{\uu}[1]{}
\renewcommand{\uu}[1]{{\mathfrak{u}({#1})}}
\newcommand{\so}[1]{}
\renewcommand{\so}[1]{{\mathfrak{so}({#1})}}
\newcommand{\usp}[1]{}
\renewcommand{\usp}[1]{{\mathfrak{usp}({#1})}}
\def\f{\mathfrak{f}}

\makeatletter
\renewcommand\xleftrightarrow[2][]{%
  \ext@arrow 9999{\longleftrightarrowfill@}{#1}{#2}}
\newcommand\longleftrightarrowfill@{%
  \arrowfill@\leftarrow\relbar\rightarrow}
\makeatother

\makeatletter
\@addtoreset{equation}{section}
\makeatother

\usepackage[en-US,showdow,showzone]{datetime2}
\DTMlangsetup[en-US]{abbr=true,showyear=false,zone=atlantic}

\graphicspath{{figures/}{figures/figs_dynkin}}

\makeatletter
\newcommand{\otherlabel}[2]{\protected@edef\@currentlabel{#2}\label{#1}}
\makeatother

\newcommand{\nocontentsline}[3]{}
\newcommand{\tocless}[2]{\bgroup\let\addcontentsline=\nocontentsline#1{#2}\egroup}

\begin{document}
\usetagform{normalsize} %

\begin{titlepage}

	\begin{center}

		\vskip .3in \noindent

		{\LARGE \textbf{Proving the 6d $\bm{a}$-theorem \\ \smallskip with the double affine Grassmannian}}
		\bigskip

		{\large \textsc{part iii}}

		\bigskip

		Marco Fazzi,$^{a,b,c}$ Suvendu Giri,$^{a,d,e}$ and Paul Levy$^f$

		\bigskip

		\bigskip
		{\small

			$^a$ Department of Physics and Astronomy,  Uppsala University,  SE-75120 Uppsala, Sweden \\
			\smallskip
			$^b$ NORDITA, Hannes Alfv\'ens v\"ag 12, SE-10691 Stockholm, Sweden \\
			\smallskip
			$^c$ School of Mathematical and Physical Sciences, The University of Sheffield,  Sheffield, S3 7RH, UK \\
			\smallskip
			$^d$ Department of Physics, Princeton University, Princeton, New Jersey 08544, USA\\
			\smallskip
			$^e$ Princeton Gravity Initiative, Princeton University, Princeton, New Jersey 08544, USA\\
			\smallskip
			$^f$ Department of Mathematics and Statistics, Fylde College,  Lancaster University,  Lancaster, LA1 4YF, UK\\
		}

		\smallskip
		{\small \tt \href{mailto:m.fazzi@sheffield.ac.uk}{m.fazzi@sheffield.ac.uk}  \hspace{.5cm} \href{mailto:suvendu.giri@princeton.edu}{suvendu.giri@princeton.edu}  \hspace{.5cm} \href{mailto:p.d.levy@lancaster.ac.uk}{p.d.levy@lancaster.ac.uk}}

		\bigskip
		\bigskip
		{\bf Abstract }
		\vskip .1in
	\end{center}

	\noindent This paper contains two results of independent interest, the first being more mathematical in nature whereas the second more physical.  We first show that the hierarchy of Higgs branch RG flows between the 6d $(1,0)$ SCFTs known as A-type orbi-instantons is given by the Hasse diagram of certain strata and transverse slices in the double affine Grassmannian of $E_8$.
	Secondly, we leverage the partial order naturally defined on this Hasse diagram to prove the $a$-theorem for orbi-instanton Higgs branch RG flows,  thereby exhausting the list of $c$-theorems in the even-dimensional supersymmetric setting.

	\vfill

	\begin{center}
		\noindent \textit{MF would like to dedicate this paper to the late Luciano Girardello, \\ who was the first to teach him about anomalies and RG flows.}
	\end{center}
	\vfill

	\begin{flushright}
		UUITP-39/23\\
		NORDITA 2023-122
	\end{flushright}
	\eject

\end{titlepage}

\tableofcontents

\section{Introduction and summary of results}
\label{sec:intro}

Recently, there have been many fascinating developments in quantum field theory (QFT).
Yet, many simple but central dynamical questions remain open,  especially pertaining to the structure of renormalization group (RG) flows along which QFTs sit, and which are ``bounded'' by fixed points, i.e.  scale-invariant field theories which (in the relativistic context) are typically assumed to also be conformally invariant, that is to be full-fledged conformal field theories (CFTs).\footnote{In two dimensions, scale and conformal invariance are famously equivalent for unitary theories \cite{Zamolodchikov:1986gt,Polchinski:1987dy}.  In four dimensions, the strong nonperturbative arguments of \cite{Luty:2012ww} were essentially proven by \cite{Yonekura:2014tha}.  Subtleties can arise when nonlocal operators are included in the theory however.  See also the lecture notes \cite{Nakayama:2013is}. } 

In this paper we will answer one of these important questions, namely we will prove the so-called $a$-theorem for an infinite class of six-dimensional superconformal field theories (6d SCFTs) with minimal $(1,0)$ supersymmetry. It is worth noting that this is the last case standing of \emph{supersymmetric} $c$-theorems in any even dimension,  and moreover that there are no known \emph{non-}supersymmetric interacting CFTs in six dimensions (or higher),  so its proof represents an important step in the study of (S)CFTs in general. The proof reduces to an analysis of some properties of a mathematical object called the double affine Grassmannian (of $E_8$).

The aim of the paper is thus to combine mathematics and physics to deliver a powerful statement in QFT.
Moreover, and perhaps most interestingly to the string theory cognoscenti, it provides another glimpse into the physics of one of the most mysterious extended objects in the theory:  the Ho\v{r}ava--Witten M9-wall.

\subsection[Six-dimensional $a$-theorem]{Six-dimensional \texorpdfstring{$\bm{a}$}{a}-theorem}
\label{sub:athm}

Concretely, by 6d supersymmetric $a$-theorem we mean:
\begin{itemize}
	\item[\emph{i)}]
	      that there is a function defined along the RG flow which at the fixed points equals the so-called $a$ conformal anomaly (or $a$ central charge) of the SCFT,  i.e. the coefficient of the 6d Euler density $E_6$ in the trace of its stress-energy tensor on a curved background,
	      \begin{equation}
		      \label{eq:T}
		      \left\langle T_\mu^\mu \right\rangle= a E_6 + \sum_{i=1}^3 c_i I_i\ .
	      \end{equation}
	      Namely, the Ward identity $\left\langle T_\mu^\mu \right\rangle =0$ of the CFT is violated by a c-number on a curved background, and for this reason equation \eqref{eq:T} is known as the trace, or Weyl,  anomaly \cite{Capper:1974ic,Capper:1975ig,Duff:1977ay,Christensen:1978gi,Duff:1993wm}.\footnote{The $I_i$ are Weyl invariants (or ``cocycles'') of weight 6, and $c_i$ their coefficients (which satisfy the relation $2c_1=c_2+c_3$ \cite{Cordova:2019wns}); see \cite{Bonora:1985cq,Deser:1993yx,henningson-skenderis,Bonora:2023yza}. Equation \eqref{eq:T} is sometimes written with a total derivative term $\nabla_i J^i$ added to the RHS, which can be safely neglected for our considerations (see e.g.  \cite{Deser:1993yx}).}
	\item[\emph{ii)}]
	      That this function decreases monotonically along any unitary RG flow connecting  ultraviolet (UV) and infrared (IR) SCFTs,  implying that
	      \begin{equation}
		      \Delta a := a_\text{UV}-a_\text{IR}>0\ .
	      \end{equation}
	      It is known that $a>0$ for all unitary SCFTs \cite{Cordova:2020tij}, and this statement is logically independent of $\Delta a >0$.\footnote{The statement $a\geq 0$ is also known to be true for 4d CFTs \cite{Hofman:2008ar}, with the bound saturated only by a theory with no local degrees of freedom.  For a recent discussion about the role of $\Delta c$ in 4d see \cite{Karateev:2023mrb,Hartman:2024xkw} and references therein.}
\end{itemize}
It is a quantitative translation of the intuition that the number of degrees of freedom should decrease along an RG flow as we integrate out high-energy modes,  and it also implies that the RG flow is irreversible. In other words, we can only flow to the IR.

As stated, this 6d $a$-theorem should be thought of as another instance of the $c$-theorem proven by Zamolodchikov in 2d \cite{Zamolodchikov:1986gt}; and conjectured \cite{CARDY1988749}, proven perturbatively \cite{Osborn:1989td,Jack:1990eb,Osborn:1991gm}, and finally non-perturbatively \cite{Komargodski:2011vj} also in 4d.\footnote{This hinges upon the findings of \cite{Schwimmer:2010za}.  See also \cite{Luty:2012ww,Komargodski:2011xv} for further discussions and \cite{Anselmi:1999xk} for an earlier investigation with a perturbative proof for some theories.} In both 2d and 4d the function that decreases along the flow equals the $a$ conformal anomaly of the CFT at fixed points (historically called $c$ in 2d, since it is equal to the central charge of the Virasoro algebra -- see e.g. \cite{Ginsparg:1988ui}).
In 5d or higher, we do not at the time of writing have \emph{conclusive} evidence in favor of the existence of non-supersymmetric interacting CFTs.\footnote{See however \cite{Apruzzi:2016rny, BenettiGenolini:2019zth,Akhond:2023vlb,Dierigl:2022reg} for proposals in 5d,  6d,  and even 8d from string constructions, even though their ultimate fate remains uncertain (see e.g. \cite{Morris:2004mg,Suh:2020rma,Bertolini:2021cew,Bertolini:2022osy,Mignosa:2023bbt,Apruzzi:2019ecr,Apruzzi:2021nle} for a large-$N$ analysis in 5d and 6d from both a field theoretic and holographic perspective).} Leaving aside the odd-dimensional case (for which there is \emph{no} trace anomaly anyway \cite[Eq. (14)]{Deser:1993yx}),  the strategy of \cite{Komargodski:2011vj} (i.e. the use of an effective action for the ``dilaton'',   the Nambu--Goldstone boson of broken conformal invariance, appearing as the conformal mode in a so-called local Riegert action \cite{Riegert:1984kt,Fradkin:1983tg,Gabadadze:2023quw}) does not generalize straightforwardly to 6d,  and for this reason it does not give rise to a general proof of the (non-supersymmetric) $a$-theorem \cite{Elvang:2012st,Kundu:2019zsl}.\footnote{See also \cite{Grinstein:2014xba,Gracey:2015fia,Osborn:2015rna,Stergiou:2016uqq}, or \cite{Casini:2023kyj} (and references therein), for an alternative version of this statement (that uses the CFT entanglement entropy  in the latter case). See \cite{Anselmi:1999xk,Anselmi:1999uk,Anselmi:1999ut,Anselmi:2002as} for early investigations on the trace anomaly in arbitrary even dimensions (in particular 6d) in the context of the $a$-theorem, and \cite{Casini:2004bw,Casini:2006es,Casini:2012ei,Casini:2017roe,Casini:2017vbe} for the use of the entanglement entropy to prove various $c$-theorems (e.g. in 3d \cite{Casini:2012ei}, where there is no trace anomaly).  See \cite{Hartman:2023qdn} for a proof that uses the averaged null energy condition. Finally,  see \cite{Myers:2010xs,Myers:2010tj} for a holographic proof of a $c$-theorem which holds in any dimension.} (See \cite{Heckman:2021nwg} and references therein for a recent reanalysis of this problem.) Therefore we have to settle for supersymmetric theories.

Luckily though, six is the largest dimension in which SCFTs can be defined \cite{Nahm:1977tg,Minwalla:1997ka} (see in particular \cite[Sec. 5.1.4]{Cordova:2016emh}),  and a massive body of literature has moreover shown that 6d SCFTs can be thought of as an ``organizing principle'' for most lower-dimensional SCFTs, geometrizing their construction, dualities, and interdependencies across dimensions. Therefore studying the $a$-theorem for 6d SCFTs is a meaningful endeavor.  With $(1,0)$ (or $(2,0)$) supersymmetry there are no relevant or marginal supersymmetry-preserving deformations we can turn on \cite{Louis:2015mka,Cordova:2016xhm}, so all RG flows are flows onto the moduli space of supersymmetric vacua of the UV SCFT \cite{Cordova:2015fha}, obtained by giving vacuum expectation values (VEVs) to some operators. This moduli space is composed of two main branches: a tensor branch,  parameterized by scalars in the tensor multiplets taking VEVs; and a Higgs branch, where scalars in the matter hypermultiplets take VEVs.
Along flows onto either branch, conformal invariance is spontaneously broken,\footnote{See e.g.  \cite[Sec. 2.2.]{Cordova:2015fha} for more details.} but in the latter case it is recovered in the deep IR, where a new SCFT sits, with a global symmetry generically different from the UV parent.\footnote{Likewise, the UV R-symmetry is broken by Higgs branch flows generically since the matter hypermultiplets are charged under it, and a new $\SU(2)$ R-symmetry emerges in the IR. There exist also mixed branch flows, arising as a combination of the two aforementioned types of flow; we will comment on them in the main text when appropriate.} On the other hand, for tensor branch flows the deep IR is a generalized quiver gauge theory of massless vectors plus tensors,\footnote{This is an example of a scale-invariant (but non-conformal) QFT, or SFT for short \cite{Luty:2012ww}. } and the $a$-theorem was proven in \cite{Cordova:2015fha} (and in \cite{Cordova:2015vwa} in the $(2,0)$ case).\footnote{See also \cite{Maxfield:2012aw} for earlier investigations in the special $(2,0)$ case. The $a$-theorem is also valid for flows from SCFT to supersymmetric SFT, see in particular \cite[Sec. 6.2]{Cordova:2015fha}.}

\subsection{A-type orbi-instantons and hierarchies of RG flows}\label{intro_RG}

Exploration of the Higgs branch RG flows, on the other hand, has remained elusive for longer. A proof of the $a$-theorem for a special but infinite class of theories known as ``T-brane theories'' (descending from the ``conformal matter'' of \cite{DelZotto:2014hpa}) was finally provided in \cite{Mekareeya:2016yal} (with prior evidence both in field theory and holography given in \cite{Gaiotto:2014lca,Heckman:2015axa,Cremonesi:2015bld,Apruzzi:2017nck}). However there exists another infinite class of 6d SCFTs \cite[Sec. 6]{DelZotto:2014hpa} which, together with a subclass of this conformal matter, generates \emph{all} other known 6d SCFTs via ``fission and fusion'' \cite{Heckman:2018pqx}. (This procedure involves subsequent ungaugings and gaugings of symmetries, respectively.) The theories in this second class were dubbed ``ADE-type orbi-instantons'' in \cite{Heckman:2018pqx},  and are the data of: the number $N$ of M5-branes simultaneously probing an M9-wall (i.e. acting as pointlike instantons in the four codimensions) and the orbifold point of $\cc^2/\Gamma_\text{ADE}$ wrapped by the M9 (with $\Gamma_\text{ADE} \subset \SU(2)$ finite); the order of said orbifold; a boundary condition at the spatial infinity $S^3/\Gamma_\text{ADE}$ of the orbifold, which is a representation $\rho_\infty : \Gamma_\text{ADE} \to E_8$ (the $E_8$ gauge bundle supported on the M9-wall acting as a flavor symmetry from the perspective of the 6d  worldvolume of the M5's). The integer $N$ and the homomorphism $\rho_\infty$ represent respectively the number of instantons and the holonomy at infinity (of the flat connection of the gauge bundle) needed to fully specify the instanton configuration on the (deformation/resolution of the) orbifold (see e.g. \cite[Sec. 2.1]{Mekareeya:2017jgc}).

When $\Gamma$ is of type A we can fix the order $k$ of the $\zz_k$ orbifold, allowing us to use the triple $(N,k,\rho_\infty)$ to indicate an orbi-instanton SCFT of this type.
We will also say that $N$ is the number of ``full instantons'' for reasons that will become clear later.\footnote{\label{foot:numberfull}Because the four-dimensional space $\cc^2/\Gamma_\text{A}$ contains an orbifold singularity, the instanton number $\int F\wedge F$ on the resolution/deformation of the singularity is \emph{fractional} and is given by $N -\left\langle \overline{\lambda}, \overline{\lambda}\right\rangle/(2k)$ ($\overline{\lambda}$ is defined in (\ref{eq:lambdaover}) and the inner product in (\ref{eq:lambdalambda})), if $\int F\wedge F = 1$ on $\cc^2$ for the smallest possible instanton number. See also \cite[p. 19]{Nakajima:2015txa} for the topological data of the instanton on the (fully unresolved) singularity.}
Each flat connection $\rho_\infty$ is conveniently specified  by a choice of ``Kac label'' \cite{Mekareeya:2017jgc}, that is an integer partition $[k_i]$ of $k$ which uses only the Coxeter labels $1,2,\ldots,6,4',3',2'$ of the \emph{affine} $E_8$ Dynkin diagram.
In the mathematics literature, and in the following sections, such a partition is called a ``Kac diagram $\lambda_\text{Kac}$ at level $k$''. It is given by a weighted affine $E_8$ Dynkin diagram of the form
\begin{equation}\label{eq:E8ni}
	\lambda_{\text{Kac}}: \quad n_1 -  n_2 - n_3 - n_4 - n_5 - \overset{\overset{\displaystyle n_{3'}}{\vert}}{n_6} - n_{4'} - n_{2'}\ ,
\end{equation}
with
\begin{equation}\label{eq:partitionpre}
	k =\left(\sum_{i=1}^6 i n_i\right) + 4 n_{4'} +3n_{3'}+2n_{2'}\ .
\end{equation}
(We briefly recap the theory of Kac diagrams in section \ref{sub:diag}.)
The UV SCFT whose Higgs branch we are exploring can be identified with the trivial choice of boundary condition which preserves the full $E_8$ from the M9,  that is Kac label $\rho_\infty: k=[1^k]$ (which obviously exists for any $k$), or $\lambda_\text{Kac}= \begin{smallmatrix} & & & & &0 &  & & \\ k & 0 & 0 & 0 & 0 & 0 & 0 & 0 \end{smallmatrix}\!\!$. The IR SCFTs that can be reached by (subsequent) Higgs branch RG flows are instead identified with other possible choices of coefficients $n_i,n_{i'}$ (some of which may be zero),  namely
\begin{equation}\label{eq:rhopart}
\rho_\infty: k=[1^{n_1},2^{n_2},3^{n_3},4^{n_4},5^{n_5},6^{n_6},4^{n_{4'}},2^{n_{2'}},3^{n_{3'}}]\ .%
\end{equation}
(How to extract the flavor symmetry of the UV and IR SCFTs from the associated Kac diagrams will be explained in section \ref{sub:eng}.)
In what follows, we will use the notations \eqref{eq:E8ni} and \eqref{eq:rhopart} interchangeably.
There is another mathematical structure naturally parameterized by pairs $(\lambda_{\rm Kac},n)$ of a Kac diagram and an integer $n$: the set of dominant coweights for an affine Kac--Moody Lie algebra.
A major part of our work will be to make precise the relationship between triples $(N,k,\rho_\infty)$ and coweights $(\lambda_{\rm Kac},n)$ for affine $E_8$.

In \cite{Fazzi:2022hal}, exploiting the technology of ``3d magnetic quivers'' \cite{Cabrera:2019izd} and ``quiver subtraction'' \cite{Cabrera:2018ann}, the first two authors have constructed very intricate hierarchies of allowed RG flows between UV and IR A-type orbi-instantons defined by different Kac diagrams, and  verified, by computing the $a$ anomaly of UV and IR SCFTs connected by an RG flow, that each flow thus constructed satisfies the $a$-theorem.
(In \cite{Fazzi:2022yca}, the same authors with Giacomelli have extended this analysis to another infinite class of 6d SCFTs known as massive E-string theories, which can be engineered from the orbi-instantons via the above-mentioned fission procedure.
They will make their appearance in Appendix \ref{app:fullhasse}.) The three present authors have moreover conjectured that, if we fix $k$ but allow $N$ to change, by performing small $E_8$-instanton transitions (i.e. by dissolving some M5's into flux on the M9), these intricate hierarchies can be understood as (semi-infinite periodic) Hasse diagrams of the ``double affine Grassmannian'' of $E_8$, introduced in the mathematics literature \cite{Braverman:2007dvq,braverman2011pursuing,braverman2012pursuing} to generalize the better known affine Grassmannian (see e.g. \cite[Sec. 6.5]{Bullimore:2015lsa} and \cite[Sec. 2]{Bourget:2021siw}).  Under the proposed identification,  (the Higgs branch of) a SCFT in the hierarchy is a symplectic leaf (or \emph{stratum}) of the Grassmannian, and the Higgs branch RG flows that connect SCFTs are the transverse slices between ``neighboring'' leaves.
Note that the direction of flow is opposite to the closure ordering on leaves, i.e. the dimension of the leaves increases as one flows to the IR.
The larger the leaf (i.e. Higgs branch), the smaller the $a$ anomaly of the SCFT. Let us see how.

In finite type, the strata in the affine Grassmannian are parameterized by the combinatorial data of {\it dominant coweights}.
(This is recapped below in section \ref{affinegrass}.)
The closure ordering on strata corresponds to the dominance ordering on coweights; the ``minimal degenerations'' (adjacent pairs) are classified by an algorithm due to Stembridge \cite{STEMBRIDGE1998340}.
Slices in the affine Grassmannian are known \cite{Braverman:2016pwk} to be isomorphic to Coulomb branches $\mathcal{M}_\text{C}$ of 3d $\mathcal{N}=4$ theories (in the sense of \cite{Nakajima:2015txa,Nakajima:2015gxa}).
The latter can also be realized in string theory via Hanany--Witten brane setups (and the brane moves one can do within them).  See e.g. \cite{Bullimore:2015lsa,Bullimore:2016nji,Bullimore:2016hdc,Bourget:2021siw} for many examples in type ABCD.
It has been recently established that these Coulomb branches have symplectic singularities \cite{Weekes,Bellamy}.\footnote{Symplectic singularities have been introduced in \cite{Beauville_2000} as an analog to rational Gorenstein singularities in Calabi--Yau varieties in the eight-supercharge setting, i.e. what physicists would call hyperkähler cones \cite{Bagger:1983tt}.}

The affine Grassmannian of a finite-dimensional simple Lie algebra ${\mathfrak g}$ plays a crucial role in its representation theory, via the geometric Satake correspondence \cite{mirkovicvilonen}.
The analogous object for the affine counterpart ${\mathfrak g}_{\rm aff}$ is the ``double affine'' Grassmannian, conjecturally defined by Braverman and Finkelberg in \cite{Braverman:2007dvq,braverman2011pursuing,braverman2012pursuing}.
It is conjectured that Braverman--Finkelberg's double affine Grassmannian also gives rise to Coulomb branches of 3d $\mathcal{N}=4$ theories   -- see e.g.  the proceedings \cite{Finkelberg:2017nbc} or \cite[Sec. 3(viii) (b)]{Braverman:2016pwk}.\footnote{In affine type A, it is proven that the Coulomb branch is given by a quiver variety of affine type \cite{Nakajima:2016guo} known as Cherkis bow variety \cite{Cherkis:2009hpw,Cherkis:2009jm,Cherkis:2010bn,Nakajima:2018ohd}.
	The bow in the name comes from a Hanany--Witten D3-D5-NS5 brane configuration on a circle.  See e.g. the video recordings of Cherkis' 2018 mini course on \emph{Instantons and monopoles} at ICTS, Bangalore \cite{cherkislec}.}
In fact, because of the semi-infinite periodic structure of its Hasse diagram (which will become apparent in later sections),  practically speaking, it only makes sense to construct transverse slices between strata.
These transverse slices are finite-dimensional, hence consist of finitely many strata ordered by the closure relation.
It is known that some, but not all, of the strata are parameterized by dominant affine coweights $(\lambda_\text{Kac},n)$, see below; the minimal degenerations for affine coweights have been classified by Roy's generalization \cite{roy} of Stembridge's result in finite types.
The basic objects of study are slices ${\mathcal M}_\text{C}(\mu,\lambda)$ between strata defined by dominant coweights $\lambda<\mu$.
Note that this requires $\lambda$ and $\mu$ to have the same level $k$ (i.e.  same \emph{fixed} order $k$ of the $\zz_k$ orbifold in M-theory -- we will comment in footnote \ref{foot:MthTbr} and in appendix \ref{app:fullhasse} on situations where $k$ is allowed to change).

The symplectic leaves in ${\mathcal M}_\text{C}(\mu,\lambda)$ have been classified in affine type A \cite[Sec. 7.7]{Nakajima:2016guo}.
In the following discussion we will assume (as is expected) that an analogous classification holds in any affine type.
(We also assume $k>1$; the case $k=1$ is similar, but with fewer strata, see section \ref{sub:k=1}.)
Just as in the finite case, each coweight $\nu$ satisfying $\lambda\leq\nu\leq\mu$ leads to a symplectic leaf in ${\mathcal M}_\text{C}(\mu,\lambda)$.
However, in contrast with the finite case, there are additional strata: if $\lambda\leq\nu\leq \mu-Mc$, where $c$ is the canonical central element (i.e. the minimal positive imaginary coroot),\footnote{\label{foot:root}Sometimes the notation $\delta$, i.e. the minimal positive imaginary {\it root}, cf. (\ref{deltadefinition}), is used in place of $c$; in simply-laced types this is permitted by the isomorphism between roots and coroots.} then between $\nu$ and $\nu+Mc$ there is also a sequence of strata corresponding to $\Sym^M({\mathbb C}^2/{\mathbb Z}_k\setminus \{ 0\})$ (and $\Sym^{M-1}({\mathbb C}^2/{\mathbb Z}_k\setminus \{ 0\})$, $\Sym^{M-2}({\mathbb C}^2/{\mathbb Z}_k\setminus \{ 0\})$, and so on).
The strata in $\Sym^M({\mathbb C}^2/{\mathbb Z}_k\setminus \{ 0\})$ are in one-to-one correspondence with the integer partitions $[m_i]$ of $M$, and are (closure) ordered via refinement of partitions (in reverse, joining parts).
We can express this in terms of 3d Coulomb branches as follows:
\begin{subequations}\label{eq:stratMC}
\begin{align}
	 & \text{strata:}\quad \mathcal{M}_\text{C}(\mu,\lambda) = \bigsqcup_{\nu,\,[m_i]} \mathcal{M}_\text{C}^\text{smooth}(\nu,\lambda) \times \Sym^M_{[m_i]}(\cc^2/\zz_k \! \setminus\! \{0\})\ ,  \label{eq:strata} \\
	 & \text{slices:}\quad \mathcal{M}_\text{C}({\mu -M c,\nu}) \times \prod_{i=1}^m  \prescript{\text{c}}{}{\mathcal{U}_{[m_i]}}\ ,\label{eq:slices}
\end{align}
\end{subequations}
where $\lambda\leq \nu\leq \mu-Mc$, $[m_i]$ is a partition of $M$, and $\bigsqcup$ stands for disjoint union.
Here (\ref{eq:slices}) denotes the transverse slice from a point of the stratum labeled by $(\nu,[m_i])$, and $\prescript{\text{c}}{}{\mathcal{U}_{[m_i]}}$ is the (Uhlenbeck partial compactification of the) centered moduli space of $m_i$ $E_8$-instantons on ${\mathbb C}^2$.
(See appendix \ref{app:magquivs} for more details.)
Note that the smooth locus $\mathcal{M}_\text{C}^\text{smooth}(\nu,\lambda)$ of ${\mathcal M}_\text{C}(\nu,\lambda)$ is precisely the open symplectic leaf (corresponding to $\nu$).\footnote{For stratification results in other contexts, e.g.  4d $\mathcal{N}=2$ Coulomb branches, see\cite{Martone:2020nsy,Argyres:2020wmq}.}

Before explaining how these 3d Coulomb branches come about in 6d, we remark the following. A similar behavior (i.e. stratification) has been observed  in  \cite{Bourget:2022tmw} for conformal matter theories of type $(A,A)$ (i.e.  ``just'' bifundamentals of $\SU(k)\times \SU(k)$ engineered by $N$ M5's probing $\cc^2/\zz_k$), in the following sense.  It is well known that if one flows onto the Higgs branch of the $(A,A)$ UV theory
\begin{equation}
	[\SU(k)]\, \underbrace{\overset{\su{k}}{2}\cdots \overset{\su{k}}{2}}_{N-1}\,[\SU(k)]
\end{equation}
by giving VEVs to matter hypermultiplets charged under the (nonabelian part of the) flavor symmetry, i.e. $\SU(k)\times \SU(k)$, one can reach new IR fixed points with a different flavor symmetry which is specified by the choice of a (two) nilpotent orbit(s) in $\SU(k)$ ($\SU(k)\times \SU(k)$); that is, we reach a T-brane theory \cite{DelZotto:2014hpa,Heckman:2016ssk}.\footnote{The abelian part of the flavor symmetry is studied in detail in \cite{Apruzzi:2020eqi} (see also \cite{Lee:2018ihr} for an earlier partial study). The AdS$_7$ holographic duals to T-brane theories were studied in \cite{Apruzzi:2013yva,Gaiotto:2014lca,Cremonesi:2015bld,Apruzzi:2017nck,Bergman:2020bvi}.} We will call these flows ``flavor Higgsings''.

However, one can also explore other phases of the UV theory (still on its Higgs branch) obtained by moving some (or all) of the M5's off of the singularity.  This is done in \cite{Bourget:2022tmw}, which constructs a Hasse diagram (via quiver subtraction of 3d magnetic quivers,  making use of an improved algorithm based on ``decorations'' \cite{Bourget:2022ehw}) for the Higgs branch of the UV theory while {excluding} the flavor Higgsings (which would lead to T-brane theories in the IR).  These phases correspond to a decoupled product of a lower-rank $(A,A)$ conformal matter engineered by $N-M$ M5's at the singularity, for some $M<N$, and a bunch of A-type $(2,0)$ theories given by stacks of M5's away from the singularity, such that the total number of these ``separated'' M5's is $M$.  Some of the leaves in this Hasse diagram are (the strata of) symmetric products and correspond to the Higgs branch of stacks of M5's away from the singularity, while the slices between adjacent symplectic leaves are either $A_{k-1}$ Kleinian singularities or certain non-normal singularities called $m$ and introduced in \cite[Sec. 1.8.4]{FJLS} (or unions of multiple copies thereof).

In this paper we will mostly be interested in flavor Higgsings, taking us from UV SCFT to IR SCFT with (generically) different flavor symmetry (with an exception for $k=1$, i.e. for E-string theories -- see section \ref{sub:k=1}).

\subsection{Higgs branch RG flows and slices in the Grassmannian}

We would now like to understand from the 6d QFT perspective what the different strata of the double affine Grassmannian of $E_8$, and the slices between them,  correspond to.  Under the identification of the hierarchy of RG flows with the Hasse diagram of the strata already proposed in \cite{Fazzi:2022hal}, the latter should correspond to Higgs branches of various SCFTs, and the slices to Higgs branch RG flows.

Generally speaking, there are two types of Higgsings (i.e. flows) that we can realize among orbi-instantons.  The first -- flavor Higgsings, as we called them above -- correspond to giving a VEV to (scalar) operators in the matter hypermultiplets charged under a flavor symmetry, e.g. the left flavor symmetry factor $\mathfrak{f} \subseteq E_8$ of the orbi-instanton. (This will also be denoted $[F]$ in the generic quiver on the tensor branch of the SCFT -- we refer the reader to \cite{Fazzi:2022hal} and \eqref{eq:genelequiv} below for the notation.) We then flow to a new IR SCFT on the Higgs branch of the UV one, which is defined by a \emph{different} Kac diagram at level $k$ (holonomy at infinity $\rho_\infty: \zz_k \to E_8$).  This type of flow is fairly easy to understand in QFT: the operators taking a VEV satisfy a chiral ring relation which defines either a ``minimal singularity'' $\mathfrak{a}_{i \leq 8},\mathfrak{d}_{i \leq 8},\mathfrak{e}_{6},\mathfrak{e}_{7},\mathfrak{e}_{8}$ (i.e.  a singular variety given by the closure of the minimal nilpotent orbit $\overline{\text{min}_\mathfrak{g}}$ of that algebra $\mathfrak{g}$), or the Kleinian singularity $A_{i-1}$ (i.e. $\cc^2/\zz_{i}$). This was described for any $k$ in \cite{Fazzi:2022hal} (with examples for some values of $k$ given before in \cite{Giacomelli:2022drw}), where moreover a connection with the so-called minimal degeneration singularities of the $E_8[\![z]\!]$-orbits of the (singly) affine Grassmannian of (non-affine) $E_8$ (classified by \cite{malin-ostrik-vybornov} using work by Stembridge \cite{STEMBRIDGE1998340}) was put forth.  Notice that this type of flow lacks a ``good'' description in terms of M-theory branes (or brane moves).  It is the analog of the ``T-brane VEVs'' for T-brane theories studied in \cite{DelZotto:2014hpa,Heckman:2016ssk}.\footnote{\label{foot:MthTbr}It would be interesting to study whether an analog to Nahm's equations, probably involving the $G_4$ flux, can be defined for M-theory flavor Higgsings. We thank Alessandro Tomasiello for this suggestion.} Likewise, one could turn on a T-brane VEV for matter charged under the right flavor symmetry factor, which for A-type orbi-instantons (and generic values of $N$) is just $[\SU(k)]$.\footnote{\label{foot:kchange}Although it is not obvious from the present discussion, it can be seen that Higgsings of the right  $[\SU(k)]$ flavor symmetry factor correspond to ``$k$-changing transitions'' in the Higgs branch Hasse diagram, and land us onto 6d quivers describing the gauge theory phase of either orbi-instantons at a \emph{lower} $k$ or massive E-string theories \cite{Fazzi:2022yca}.  The reason becomes clear if one considers the 3d magnetic quiver (\ref{eq:firstmagquiv}) of the 6d theory, as the $[\SU(k)]$ symmetry is realized by the $1-2-\cdots-k-$ tail there, and Higgsing the former means modifying the latter (generically reducing the highest rank $k$ of the unitary gauge groups that appear therein).  These transitions can be seen e.g.  along the bottom diagonal of  \cite[Fig. 17]{Bourget:2024mgn} and in figure \ref{fig:outputk=2}. (We would like to thank the authors of that paper for discussion on this and related points.) Presumably then, from the perspective of the original M-theory setup they correspond to partially deforming the orbifold, $\cc^2/\zz_k \to \cc^2/\zz_{k-1}$.
Such an effect is known to arise in the F-theory description \cite{colrafT,colrafMF} of T-brane bound states of Type IIB seven-branes \cite{cecotti-cordova-heckman-vafa}, which lends more evidence to support the claim of footnote \ref{foot:MthTbr}.}  Turning on such a VEV for either flavor symmetry factor nontrivially modifies the tails of the 6d quiver on the tensor branch, and in ``short-enough'' cases the VEVs for the two symmetries can interact, substantially modifying the ``core'' of the quiver, an effect already discussed in \cite[Sec. 5]{Heckman:2016ssk}.  For sufficiently long quivers however (i.e. for generic values of $N$) these effects do not arise; moreover, as already stated, the $a$-theorem for right flavor Higgsings has already been proven in \cite{Mekareeya:2016yal}. For these two reasons we will not discuss the latter further,  save for a brief mention in the conclusions and for appendix \ref{app:fullhasse}, where we will showcase the full Higgs branch of the UV orbi-instanton $(N=2,k=2,\rho_\infty:2=[1^2])$.

The second type of flow has the opposite behavior, i.e. it has an easy description in terms of branes but lacks an equally easy description in QFT, at least for orbi-instantons.\footnote{\label{foot:rank2}For a rank-2 E-string (which can be thought of as a limiting case of an orbi-instanton with $(N=2,k=1,\rho_\infty:1=[1^1])$, i.e. no orbifold to begin with) a QFT description exists: the Higgsing is induced by giving a VEV to an $\SU(2)$ moment map, where this $\SU(2)$ is a factor of the flavor symmetry \cite{Beem:2019snk,Giacomelli:2020jel}. We would like to thank Simone Giacomelli for discussion on this point.} It corresponds to separating some of the M5's from the stack at the singularity, while still keeping them on the M9, and organizing them into substacks, potentially made up of a single brane.\footnote{It is the analog of the Higgs branch flows studied in \cite{Bourget:2022tmw} for the $(A,A)$ theory of $N$ M5's probing $\cc^2/\zz_k$.  In \cite{Heckman:2016ssk} it is suggested that the ``separation'' brane move should correspond to a diagonalizable VEV, i.e. a semisimple element of the algebra $\mathfrak{f}=\su{k}$ of the flavor factor $[\SU(k)]$; on the contrary, our flavor Higgsings correspond in that context to nilpotent elements, i.e. T-brane VEV, as reiterated above.  We would like to thank Alessandro Tomasiello for discussion on this point.} At the orbifold point we are left with $N-M$ M5's, i.e. a ``lower-rank'' orbi-instanton. See figure \ref{fig:decoupling} for clarity.
\begin{figure}[ht!]
	\centering
	\includegraphics[width=0.7\textwidth]{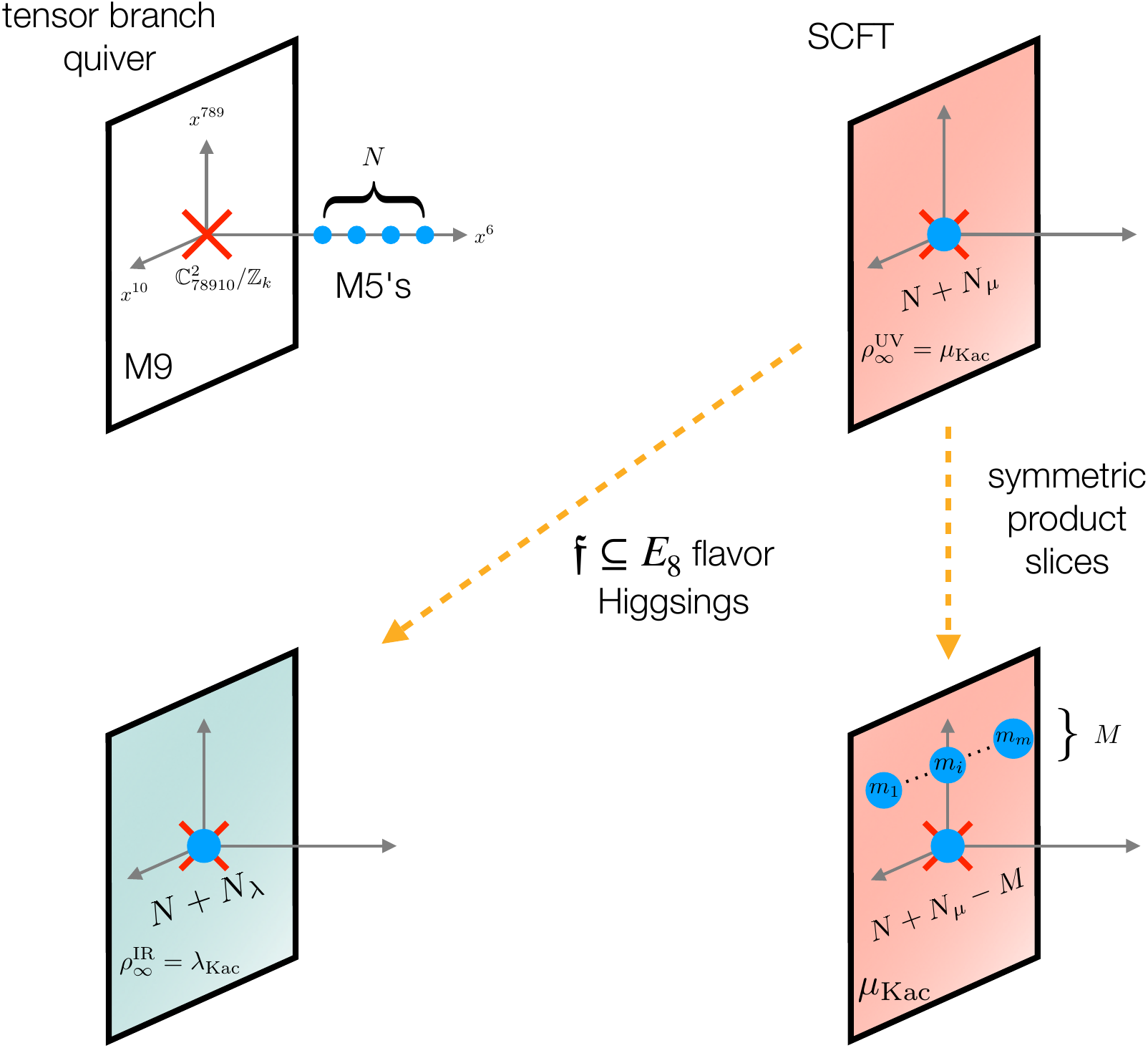}
	\caption{\textbf{Top left:} the generic point on the tensor branch of the UV orbi-instanton.  \textbf{Top right:} the origin of the tensor branch where the orbi-instanton SCFT defined by the triple $(N,k,\rho_\infty^\text{UV})$ lives (the generic $\rho_\infty^\text{UV}$ is defined in (\ref{eq:rhopart})); all $N$ M5's (the full instantons) are brought on top of each other and on top of the $\cc^2/\zz_k$ orbifold point, where the M9 sits (along $x^6$), fractionating into $N_\rho$ pieces (themselves acting as fractional instantons),  with $N_\rho$ defined in (\ref{eq:Nmu}). \textbf{Bottom left:} an IR orbi-instanton defined by a different Kac diagram $(N,k,\rho_\infty^\text{IR})$ connected via Higgs branch RG flow to the UV theory. (The coweight corresponding to $(N,k,\rho_\infty^\text{IR})$ lies \emph{above} the coweight corresponding to $(N,k,\rho_\infty^\text{UV})$ in the partial order.)
		To reach it, we activate a VEV for matter operators (in flavor hypermultiplets), and generically break the left  flavor symmetry factor of the UV theory (a maximal subalgebra of $E_8$, including the full $E_8$ by extension).   \textbf{Bottom right:} a product of a ``lower-rank'' orbi-instanton $(N-M,k,\rho_\infty^\text{UV})$ times $m$ substacks of $m_i$ M5's each, defining the partition $[m_i]$ (i.e. $\sum_{i=1}^m m_i = M$).}
	\label{fig:decoupling}
\end{figure}
Each of the ``separated'' M5's sees a copy of the singularity in its transverse space (within the M9), so $M$ (indistinguishable) M5's will see the symmetric product $\Sym^M(\cc^2/\zz_k\setminus \{0\})$ {(of quaternionic dimension $\dim_\mathbb{H}=M$)}.\footnote{The symmetric product is known to be the Coulomb branch (as a quiver variety) of the so-called ``Jordan quiver'' \cite{Nakajima:2015txa}.  The Higgs branch of the Jordan quiver is instead the (Uhlenbeck partial compactification of the) moduli space of $k$ instantons of $\SU(M)$ on $\cc^2$ by ADHM \cite{ATIYAH1978185}.  Thanks to mirror symmetry \cite{Intriligator:1996ex,deBoer:1996mp} (boiling down to a ``symplectic duality'' \cite{braden2022quantizations,Braden:2014iea} in this case), this Higgs branch is the same as the Coulomb branch of the mirror (framed) necklace quiver, whereas its symmetric product Coulomb branch is the same as the Higgs branch of the necklace.}
More precisely, we need to specify an integer partition $[m_i]$ of $M$, say with $m$ parts, to describe $m$ substacks each containing $m_i$ branes.
The symmetric product $\Sym^M({\mathbb C}^2/{\mathbb Z}_k\setminus \{ 0\})$ has a nontrivial Hasse diagram: we can join together substacks, i.e. we can glue together parts of the partition $[m_i]$.  (As stated above, the slices in this Hasse diagram are either the quotient ${\mathbb C}^2/{\mathbb Z}_k$ or a union of finitely many copies of ${\mathbb C}^2/{\mathbb Z}_2$ or the non-normal singularity $m$ studied in \cite{FJLS}.) Furthermore, the full symmetric product $\Sym^M({\mathbb C}^2/{\mathbb Z}_k)$ is the disjoint union (\emph{stratification}):
\begin{equation}
	\Sym^M({\mathbb C}^2/{\mathbb Z}_k  \! \setminus\! \{ 0\})\sqcup \Sym^{M-1}({\mathbb C}^2/{\mathbb Z}_k  \! \setminus\! \{ 0\}) \sqcup \Sym^{M-2}({\mathbb C}^2/{\mathbb Z}_k  \! \setminus\! \{ 0\}) \sqcup \cdots \sqcup \, \{ 0\}\ ,
\end{equation}
and hence the Hasse diagram extends to include \emph{all} symmetric products $\Sym^{M-i}$ for $i\leq M$; in this context one can also dissolve a stack of $i$ M5's into ($G_4$) flux (within the M9), i.e.  perform $i$ small $E_8$-instanton transitions transmuting $i$ tensors into $29i$ hypermultiplets \cite{Ganor:1996mu}, thus degenerating from $\Sym^M$ to $\Sym^{M-i}$.
Finally, we can remove any of the remaining M5's from the singularity (reducing the number $N-M$ of remaining M5's), in this case degenerating from $\Sym^M$ to $\Sym^{M+i}$ for some $i$.

Notice that, because of the presence of the M9, each substack of separated M5's provides a decoupled copy of a rank-$m_i$ E-string theory,
\begin{equation}
	[E_8]\,\underbrace{12\cdots 2}_{m_i}\ ,
\end{equation}
itself a 6d $(1,0)$ interacting SCFT.

It is natural to identify (the Higgs branch of) orbi-instantons defined by triples $(N,k,\rho_\infty)$ (and without any decoupled E-string) with strata of the double affine Grassmannian defined by a dominant coweight $\lambda=(\lambda_\text{Kac},n)$ of affine $E_8$,  since specifying $\rho_\infty$ is the same as specifying a $\lambda_\text{Kac}$ at level $k$.
(At the end of section \ref{subsub:mindeg} we will elucidate the precise relationship between $N$ and $n$).
On the other hand,  the $M$ separated M5's organized in $m$ substacks are identified with the strata $\Sym^M_{[m_i]}(\cc^2/\zz_k\! \setminus\! \{0\})$.
Slices between UV and IR orbi-instantons $(N,k,\rho_\infty^\text{UV})$ and $(N,k,\rho_\infty^\text{IR})$ are identified with degenerations $\mu > \lambda$; slices between symmetric products with the reduced moduli spaces of $m_i$ $E_8$-instantons on $\cc^2$. The quaternionic dimension of the latter is $\dim_\mathbb{H}  \prescript{\text{c}}{}{\mathcal{U}_{[m_i]}} = 30 m_i -1$.
For $m_i=1$ (i.e. dissolving a single M5 from the separated stacks into flux) this dimension is obviously 29, i.e. the dimension of $\overline{\min_{E_8}}$ or the Higgs branch of a (rank-1) E-string.  This is precisely the number of hypermultiplets produced by a single small $E_8$-instanton transition. Then the slice $\prod_{i=1}^m \prescript{\text{c}}{}{\mathcal{U}_{[m_i]}}$ means we are performing a total of $M=\sum_{i=1}^m m_i$ small $E_8$-instanton transitions, reducing $N$ to $N-M$.\footnote{For a rank-2 E-string, which can be thought of as the orbi-instanton $(N=2,k=1,\rho_\infty:1=[1^1])$, according to footnote \ref{foot:rank2} the full Hasse diagram of Higgs branch RG flows can be constructed \cite[Fig. 45]{Martone:2021ixp}.  In that figure, $\mathfrak{a}_1$ is the operator VEV charged under $\SU(2)$ of \cite{Beem:2019snk,Giacomelli:2020jel}; $\mathbb{H}=\cc^2=\Sym^1(\cc^2)$; $\mathfrak{g}=\mathfrak{e}_8$ indicates the $\mathfrak{e}_8$ minimal degeneration, i.e. $\overline{\min_{E_8}}$, and corresponds either to a slice of the type just described (going from a stack of 2 M5's to 2 stacks of 1 M5 each, i.e. $m=2$ and $m_i=1$) or a small $E_8$-instanton transition, when a rank-1 E-string (i.e. one tensor multiplet coupled to  $E_8$ matter) transmutes into $29=\dim_{\mathbb{H}} \overline{\min_{E_8}}$ hypermultiplets. More generally,  \cite[Sec.  3.2]{Lawrie:2023uiu} constructs (albeit in a schematic form) the Higgs branch Hasse diagram of the generic rank-$P$ E-string. Another explicit example for $P=4$ can be see on the left of figure \ref{fig:outputk=2}.}

\subsection{Three dimensions from six dimensions, and vice versa}

The result $\dim_\mathbb{H} \prescript{\text{c}}{}{\mathcal{U}_{[m_i]}}= 30 m_i -1$ was recovered using the 3d magnetic quiver of the rank-$m_i$ E-string (which will appear in section \ref{sub:k=1}) in \cite{Hanany:2018uhm},  and this is no coincidence, as we now explain.

The Higgs branch of a supersymmetric theory is expected to be invariant under torus compactification: compactifying on $T^3$ the F-theory ``electric quiver'' \eqref{eq:genelquivpreblow} of a 6d orbi-instanton (i.e. its weak-coupling limit on the tensor branch, obtained via the algorithm of \cite{Mekareeya:2017jgc} from a choice of $(N,k,\rho_\infty)$),\footnote{Both 6d $(1,0)$ tensor and vector multiplets reduce to 3d $\mathcal{N}=4$ vector multiplets.} and applying mirror symmetry to the 3d quiver gauge theory thus obtained,  we land on a new ``magnetic quiver'' \cite{Cabrera:2019izd}:\footnote{By the results of \cite{Benini:2010uu},  it is star-shaped with three legs, given that it is mirror dual to the compactification on $S^1$ of the class-S theory obtained by compactifying the orbi-instanton on $T^2$, itself a fixture with three punctures (defined in \cite[Eq. (4.1)]{Mekareeya:2017jgc}).}
\begin{equation}
	\label{eq:firstmagquiv}
	{\scriptstyle
		1 - 2 - \cdots - k - (r_1+\widetilde P) -(r_2+2\widetilde P) -(r_3+3\widetilde P) - (r_4+4\widetilde P) -(r_5+5\widetilde P)-\overset{\overset{\scriptstyle r_{3'}+3\widetilde P}{\vert}}{(r_6+6\widetilde P)}-(r_{4'}+4\widetilde P)-(r_{2'}+2\widetilde P)}\ .
\end{equation}
The latter is a 3d $\mathcal{N}=4$ unitary quiver gauge theory flowing to an SCFT in the IR; the SCFT sits at the intersection between 3d Higgs and Coulomb branch. (We are selecting a vacuum of maximal breaking of the gauge symmetry, i.e. the product gauge group $\Pi_i \U(s_i)$ of rank $r_\text{V}$ is Higgsed to a maximal torus $T^{r_\text{V}}$.
The rank $r_\text{V}$ is given by the sum of all gauge ranks $s_i$ in \eqref{eq:firstmagquiv} minus one, since there are no flavor nodes in the quiver.)  The power of this construction lies in the fact that the 3d Coulomb branch captures the Higgs branch of the starting 6d theory at weak \emph{and} strong coupling \cite{Cabrera:2019izd} (i.e.  at the generic point on the tensor branch and at its origin, where the CFT sits).

The details of the above quiver (which can be found in appendix \ref{app:magquivs}, together with the relevant notation) are not important for the present discussion. What is important is that this quiver is neither of finite nor of affine type, so strictly speaking the stratification result outlined in section \ref{intro_RG} (a conjectural extension of the results of \cite{Nakajima:2016guo} to affine type E) does \emph{not} apply straightforwardly. (Incidentally, it would be interesting to develop the technology needed to construct its Coulomb branch as a quiver variety \cite{naka-priv}. We will briefly come back to this point in the conclusions.)

As a first step toward understanding the Coulomb branch $\hat{\mathcal{M}}_\text{C}$ of \eqref{eq:firstmagquiv}, it was proposed in \cite[Sec. 4.3]{Mekareeya:2017jgc} that the hyperkähler quotient of the latter times the hyperkähler space $\mathcal{O}_{\bm{\xi}}$ by $\SU(k)$, i.e.  $(\hat{\mathcal{M}}_\text{C}\times \mathcal{O}_{\bm{\xi}})/\!\!/\!\!/_{\bm{\xi}}\SU(k)$, yields the centered moduli space $\mathcal{M}_\text{C}^{\text{inst},\bm{\xi}}$ of $E_8$-instantons on the \emph{deformation or resolution} of $\cc^2/\zz_k$.  (This is proved in appendix \ref{app:moduli}.) The $E_8$-instanton moduli space on the singular space $\cc^2/\zz_k$ (i.e. for $\bm{\xi}=0$) is given by the Coulomb branch of a (framed) quiver of affine type E,  namely
\begin{equation}
	\label{eq:secondmagquivinf}
	{\scriptstyle \boxed{\scriptstyle k} - (r_1+\widetilde P) -(r_2+2\widetilde P) -(r_3+3\widetilde P) - (r_4+4\widetilde P) -(r_5+5\widetilde P)-\overset{\overset{\scriptstyle r_{3'}+3\widetilde P}{\vert}}{(r_6+6\widetilde P)}-(r_{4'}+4\widetilde P)-(r_{2'}+2\widetilde P)}\ ,
\end{equation}
for which the results of \cite{Braverman:2016pwk, Nakajima:2016guo} should apply. 

In fact, for ${\bm{\xi}}=(\xi_\cc,\xi_\rr)\in \su{k} \otimes (\cc \oplus \rr)$ suitably generic,\footnote{See \cite[Sec. 2]{Nakajima:1994} for more details.} $\mathcal{O}_{\bm{\xi}}$ can be understood as a $\SU(k)_\cc$ regular semisimple orbit, and has a hyperkähler metric specified by $\xi_\rr$.  The latter acts as a resolution modulus for $\cc^2/\zz_k$; on the other hand a nonzero $\xi_\cc$ deforms the singularity.  From a physics perspective,  $\mathcal{O}_{\bm{\xi}}$ is nothing but the Coulomb (or Higgs) branch of the $T(\SU(k))$ tail (which is self-mirror dual \cite[Sec. 3.3.1]{Gaiotto:2008ak}) on the left of \eqref{eq:firstmagquiv}, i.e.
\begin{equation}
	1-2-\cdots-(k-1)-\boxed{k}\ ,
\end{equation}
where now ${\bm{\xi}}$ can be understood as an $\SU(2)$ (R-symmetry) triplet of mass parameters of the $\SU(k)$ flavor symmetry represented by $\boxed{k}$.\footnote{More precisely, the flavor symmetry of $T(\SU(k))$ is given by $\SU(k)\times \SU(k)^\vee$ \cite{Gaiotto:2008ak}, where the Langlands dual is $\SU(k)^\vee = \PU(k) \coloneqq \U(k)/\U(1) \cong \SU(k)/\mathbb{Z}_k$.  Then $\boxed{k}$ represents the $\PU(k)$ factor.  This fact will be relevant in appendix \ref{app:moduli}. The group $\SU(k)$ acts on the Higgs branch of $T(\SU(k))$, which corresponds to the nilpotent cone $\mathcal{N}_{\SU(k)}$, while $\PU(k)$ acts on the Coulomb branch, which is given by the nilpotent cone $\mathcal{N}_{\PU(k)} = \mathcal{N}_{\SU(k)}$ (since $\SU(k)$ and $\SU(k)^\vee$ share the same Lie algebra, $\mathfrak{su}(k)$).}  When ${\bm{\xi}}=0$ this Coulomb branch is given by the nilpotent cone of $\su{k}$ (see e.g. \cite[Rem. 8.5 (3)]{Nakajima:1994}).

Because of this $T(\SU(k))$ tail, the stratification of the Coulomb branch $\hat{\mathcal{M}}_\text{C}$ will contain \emph{many more} strata than $\mathcal{M}_\text{C}^{\text{inst},\bm{\xi}}$: some strata will come from the base (which is parameterized by $\bm{\xi}$), while others from the fiber (i.e. the instanton moduli space for each possible value of $\bm{\xi}$).  However, we can stay in the vicinity of $\xi_\rr = 0$ (that is, we do not resolve the orbifold, but we may allow for complex deformations);\footnote{This is the analog of \cite[Prop.  6.2]{Nakajima:1994}.} then, the stratification result \eqref{eq:stratMC}, which is expected to hold for a quiver of affine type E \cite[Sec. 6]{Nakajima:1994}, i.e. for  \eqref{eq:secondmagquivinf}, should carry over to \eqref{eq:firstmagquiv} \cite{naka-priv}.  In practice, we are disregarding the strata associated with the $T(\SU(k))$ tail or, in physics language,  the Higgsings of the \emph{right} $[\SU(k)]$ flavor symmetry factor in 6d (corresponding to the $k$-changing transitions mentioned in footnote \ref{foot:kchange}, and to selecting a nontrivial $\xi_\mathbb{C}$),  for which the $a$-theorem was proven in \cite{Mekareeya:2016yal}.  On the other hand, all other strata (namely dominant coweights, i.e. Kac diagrams for the left flavor symmetry, and symmetric products) will be considered in our proof of the $a$-theorem.  Moreover, in appendix \ref{app:fullhasse} we will see an example where \emph{all} strata are included.

We are finally ready to interpret strata and slices.
We identify the Higgs branch of the UV theory with $\mathfrak{f}=E_8$ flavor symmetry with the Coulomb branch $\mathcal{M}_\text{C}(\mu,\lambda)$ in \eqref{eq:strata}, where $\lambda,\mu$ are dominant coweights of affine $E_8$ satisfying $\lambda\leq\mu$.
We write $\lambda=(\lambda_{\rm Kac},n)$ and similarly for $\mu$; then we must have $\mu = \lambda+ \sum_i v_i \alpha_i^\vee$ where $\alpha_i^\vee$ are the simple coroots.
From now on we will \emph{always} set $\mu_\text{Kac}$ to be the UV Kac diagram $[1^k]$:
\begin{equation}
\rho_\infty^\text{UV}: k=[1^k]\ \longleftrightarrow\  \mu_\text{Kac}= \begin{smallmatrix} & & & & &0 &  & & \\ k & 0 & 0 & 0 & 0 & 0 & 0 & 0 \end{smallmatrix}\ .
\end{equation}
Then $\mathcal{M}_\text{C}(\mu,\lambda)$ is obtained from the Coulomb branch $\hat{\mathcal{M}}_\text{C}$ of the UV orbi-instanton's magnetic quiver  by \emph{disregarding} the $T(\SU(k))$ tail (in the sense just explained).
We now look at the stratification of this Coulomb branch, which determines the possible RG flows (with the direction of flow opposite to the closure order).
As outlined in section \ref{intro_RG}, each coweight $\nu$ satisfying $\lambda\leq\nu\leq\mu$ gives rise to a stratum ${\mathcal M}_\text{C}^\text{smooth}(\nu,\lambda)$, which is the open symplectic leaf in ${\mathcal M}_\text{C}(\nu,\lambda)$.
Moreover, if $\nu\leq \mu-Mc$ for $M\geq 1$ then there are several additional strata which lie above ${\mathcal M}_\text{C}^\text{smooth}(\nu,\lambda)$ in the closure order (i.e. downstream along RG flows) but not above ${\mathcal M}_\text{C}^\text{smooth}(\xi,\lambda)$ for any $\xi$ with $\nu<\xi\leq\mu$.
These strata are indexed by partitions $[m_i]$ of $M$ and are ``seen'' by $M$ separated M5's away from the singularity, organized in substacks containing $m_i$ branes each.
We will label the additional stratum given by a partition $[m_i]$ of $M$ by ${\mathcal M}^{[m_i]}_\text{C}(\nu,\lambda)$.
Note that ${\mathcal M}_\text{C}(\nu+Mc,\lambda)$ contains the closure of ${\mathcal M}^{[m_i]}_\text{C}(\nu,\lambda)$, which strictly contains ${\mathcal M}_\text{C}(\nu,\lambda)$; full details of the possible transitions are given in section \ref{subsub:newdeg}.
The transverse slice to ${\mathcal M}^\text{smooth}_\text{C}(\nu,\lambda)$ is isomorphic to the Coulomb branch ${\mathcal M}_\text{C}(\mu,\nu)$ (and can be understood by quiver subtraction).
The transverse slice to ${\mathcal M}^{[m_i]}_\text{C}(\nu,\lambda)$ is isomorphic to the product of ${\mathcal M}_\text{C}(\mu,\nu+Mc)$ and the (Uhlenbeck partial compactifications of the) centered moduli spaces of $m_i$ $E_8$-instantons on $\cc^2$, i.e. the Higgs branch of a rank-$m_i$ E-string (given by $m_i$ M5's away from the singularity but lying on the M9).%

\subsection{Strategy of proof and results}

Having spelled out the details of the various possible Higgsings, we will now explain what we will do with them. 

First of all,  for flows from the UV orbi-instanton to a lower-rank orbi-instanton times a collection of decoupled E-strings in the IR, proving the $a$-theorem is straightforward.  The 6d anomaly polynomial of the decoupled system is given by the sum of the anomaly polynomials of the two ingredients, i.e. lower-rank orbi-instanton and E-strings. The same is true for the total $a$ anomaly, being the sum of the $a$ anomalies of the two ingredients. Therefore, establishing an $a$-theorem for these flows boils down to separately proving it for their ``constituent'' $a$ anomalies. It turns out that  the $a$ anomaly of an orbi-instanton defined by $(N,k,\rho_\infty)$ is always higher than that of the system $(N-M,k,\rho_\infty)$ plus a collection of $n_i$ rank-$m_i$ E-strings (where $[n_i,m_i]$ is an integer partition of $M$, i.e. $\sum_i n_i m_i = M$), e.g. a single rank-$M$ E-string, $M$ rank-$1$ E-strings, and more general combinations. We prove this in appendix \ref{app:anomaly_decoupled}. Therefore, these flows are \emph{included} within our proof of the $a$-theorem for flavor Higgsings of the left symmetry factor $\mathfrak{f}$.

Having identified orbi-instantons $(N,k,\rho_\infty)$ with dominant coweights $\lambda=(\lambda_\text{Kac}, n)$ of affine $E_8$, and UV-IR left flavor Higgsings with degenerations $\mu > \lambda$,  our proof consists of two main steps:
\begin{itemize}
	\item[\emph{i})]
	      As we shall explain in section \ref{subsub:corootscoweightssubsec}, affine coweights can be expressed as triples $(k,\overline\lambda,n)$, where $\overline\lambda$ is a weighting of the \emph{finite} $E_8$ Dynkin diagram and $n$ is an arbitrary integer.
	      On the basis of an explicit algorithm (due to Roy \cite{roy}) determining the Hasse diagram of affine coweights, we will interpret $n$ in terms of the physical number $N$ of M5-branes (in the process clarifying the relationship between the numbers $N_3$, $N_S$ and $N_6$ introduced in \cite{Mekareeya:2017jgc}).
	      This provides the sought-after partial order on the space of homomorphisms $\Hom(\zz_k,E_8)$ (a question asked e.g.  in \cite{Heckman:2018pqx,Heckman:2018jxk,Frey:2018vpw,Argyres:2022mnu}) and establishes, independently of 3d magnetic quivers and quiver subtractions, a hierarchy between Higgs branch RG flows.
	      As a result, the hierarchies of \cite{Fazzi:2022hal} can be embedded as connected subdiagrams in the aforementioned Hasse diagram. (In that paper, they were obtained by fixing the sum $N+N_\rho$ for ease of illustration, where the number $N_\rho$ depends solely on the choice of $\rho_\infty$ and is defined in \eqref{eq:Nmu} -- see section \ref{sec:check} for an expanded explanation on this.)

	\item[\emph{ii})]
	      Secondly, the $a$ anomaly of each SCFT can be written in terms of combinatorial data of the associated stratum.  (This is akin to what happens for T-brane theories, where the hierarchy of nilpotent flavor Higgsings mimics the Hasse diagram of nilpotent orbits of the flavor Lie algebra \cite{Heckman:2016ssk}, and writing $a$ in terms of the dimensions of these orbits ultimately allows one to prove that $\Delta a>0$ for Higgs branch flows \cite{Mekareeya:2016yal}.  Here we see that this happens also in the case of orbi-instantons, but with the dominant coweight strata of the double affine Grassmannian.)  The partial order on the dominant coweights induces a partial order on the SCFTs (exactly the one obtained in \cite{Fazzi:2022hal,Giacomelli:2022drw} via 3d magnetic quiver subtraction and for $k=4$ via a 6d anomaly polynomial analysis in \cite{Frey:2018vpw}), and writing $a$ in terms of data of the strata allows us to prove $\Delta a>0$ in full generality.
	      There are finitely many cases to check, and we have checked the required condition $\Delta a>0$ by computer in all cases.
	      However, we also sketch a conceptual proof which is valid for a large family of cases (and which could in theory be adapted to other cases).
	      This requires a certain amount of analysis of the formula for the $a$ anomaly,  especially terms $\sum_{\alpha>0}\overline\lambda(\alpha)^3$ and $\sum_{\alpha>0}\overline\lambda(\alpha)^5$ which appear in it.
\end{itemize}

\subsection{Organization}

The rest of the paper contains the technical proofs for the statements made so far, and is organized as follows.
In section \ref{sec:orbi} we define Kac diagrams as the relevant objects needed to classify $\zz_k$-gradings of Lie algebras ($E_8$ in our case), i.e. homomorphisms $\Hom(\zz_k,E_8)$, and we give a lightning review of orbi-instantons (of type A).
In section \ref{sec:grass} we introduce the affine and double affine Grassmannian of $E_8$, and give an explicit algorithm to construct strata of the latter defined by dominant coweights.  Using this algorithm,  in section \ref{sec:check} we construct Hasse diagrams of dominant coweights for a few values of $k$, and we verify that the hierarchies of RG flows (left flavor Higgsings) obtained in  \cite{Fazzi:2022hal,Giacomelli:2022drw,Frey:2018vpw} (via magnetic quiver subtraction or 't Hooft anomaly matching based on the 6d anomaly polynomial) can easily be recovered from these Hasse diagrams as connected subdiagrams.   In section \ref{sec:a} we prove the $a$-theorem, inducing the partial order on the $a$ anomalies of UV and IR orbi-instantons  from the corresponding two strata (dominant coweights) of the Grassmannian, such that $\Delta a>0$ for any pair of SCFTs connected by RG flow (flavor Higgsing). We close in section \ref{sec:conc} with a few observations and an outlook.  In appendix \ref{app:magquivs} we recall properties and notation of 3d magnetic quivers associated with 6d orbi-instantons, and we also prove that $\mathcal{M}_\text{C}^{\text{inst},\bm{\xi}}=(\hat{\mathcal{M}}_\text{C}\times \mathcal{O}_{\bm{\xi}})/\!\!/\!\!/_{\bm{\xi}} \SU(k)$ as quiver varieties. In appendix \ref{app:anomaly_decoupled} we prove the $a$-theorem for flows between UV orbi-instanton and a decoupled system of a lower-rank orbi-instanton times a bunch of E-strings.  In appendix \ref{app:fullhasse} we showcase for the benefit of the reader the full Higgs branch of the UV orbi-instanton $(N=2,k=2,\rho_\infty^\text{UV}:2=[1^2])$, including left \emph{and} right flavor Higgsings (i.e. also the strata associated with the $T(\SU(2))$ tail $1-\boxed{2}$).

\section{Kac diagrams and orbi-instantons}
\label{sec:orbi}

We begin this section with a brief account of the theory of Kac diagrams, and end it with a lightning review of orbi-instantons and how Kac diagrams enter the physics discussion.

\subsection[Kac diagrams and $\zz_k$-gradings of $E_8$]{Kac diagrams and \texorpdfstring{$\bm{\zz_k}$}{Zk}-gradings of \texorpdfstring{$\bm{E_8}$}{E8}}
\label{sub:diag}

Kac proved in \cite{kac1990infinite} that the periodic automorphisms of a finite-dimensional simple Lie algebra, i.e. the gradings by finite cyclic groups (denoted $\zz_k$ here), all arise via ${\mathbb Z}$-gradings of a corresponding affine Kac-Moody Lie algebra.
The combinatorial data of this correspondence is encapsulated in the set of Kac diagrams.
In this section we briefly recall this relationship; we refer the reader to \cite[Ch.. 6-8]{kac1990infinite} for further details.
For simplicity, we focus on the untwisted case, which is the only one relevant to us.\footnote{The reader interested in the twisted case may consult \cite{Tachikawa:2011ch} for a brief discussion on outer automorphisms.}

Let ${\mathfrak g}$ be a finite-dimensional simple Lie algebra and let ${\mathfrak g}_{\rm aff}$ be the corresponding (untwisted) affine Kac--Moody Lie algebra.
Fix once and for all a Cartan subalgebra ${\mathfrak t}$ of ${\mathfrak g}$ (contained in a Cartan subalgebra ${\mathfrak t}_{\rm aff}$ for ${\mathfrak g}_{\rm aff}$), and let $\Phi=\Phi({\mathfrak g},{\mathfrak t})$ denote the root system of ${\mathfrak g}$ relative to ${\mathfrak t}$.
Choose a positive system $\Phi^+\subset \Phi$ (or equivalently, choose a Borel subalgebra ${\mathfrak b}$ containing ${\mathfrak t}$).
Let $\{ \alpha_1,\ldots ,\alpha_r\}$ be a basis of simple roots in $\Phi^+$, which we complete to a basis $\{ \alpha_0,\ldots ,\alpha_r\}$ for the affine root system $\Phi_{\rm aff}$.
Recall that ${\mathfrak g}_{\rm aff}$ is not simple: it is a semidirect product $[{\mathfrak g}_{\rm aff},{\mathfrak g}_{\rm aff}]\oplus {\mathbb C}d$, where $d$ is any element of the Cartan subalgebra not contained in the span of the coroots.
(The standard choice -- see e.g. \cite[Ch. 6.2]{kac1990infinite} is $d$ such that $\alpha_0(d)=1$ and $\alpha_i(d)=0$ for $i>0$.)
The derived subalgebra ${\mathfrak g}'_{\rm aff}\coloneqq[{\mathfrak g}_{\rm aff},{\mathfrak g}_{\rm aff}]$ is called the affine Lie algebra.
Note that ${\mathfrak g}'_{\rm aff}$ is  not simple either: it has a one-dimensional center ${\mathbb C}c$, where $c$ is the canonical central element \cite{kac1990infinite}. (The value of $c$ on a given representation is called, rather evocatively, the ``central charge'' \cite{Ginsparg:1988ui}.\footnote{See e.g. \cite{haine} for a brief account on the so-called Sugawara construction.})
It is well known that the quotient ${\mathfrak g}'_{\rm aff}/{\mathbb C}c$ is isomorphic to the {\it loop algebra} ${\mathcal L}({\mathfrak g})\coloneqq{\mathfrak g}\otimes {\mathbb C}[t,t^{-1}]$.

There is a more general loop algebra construction.
Let $\theta$ be an automorphism of order $k$ of the simple Lie algebra ${\mathfrak g}$ and let $\zeta=e^{2\pi i/k}$.
Since its minimal polynomial divides $t^k-1$, $\theta$ acts diagonalizably on ${\mathfrak g}$, hence there is a decomposition as a direct sum of eigenspaces:
\begin{equation}
	\mathfrak{g}={\mathfrak g}(0)\oplus {\mathfrak g}(1)\oplus\ldots \oplus {\mathfrak g}(k-1)
\end{equation}
where ${\mathfrak g}(j) = \{ x\in{\mathfrak g} : \theta(x)=\zeta^j x\}$.
This is a grading over $\zz_k$: for any integers $j,l$ modulo $k$ we have $[{\mathfrak g}(j),{\mathfrak g}(l)]\subset {\mathfrak g}(j+l)$.
We define an infinite-dimensional, ${\mathbb Z}$-graded Lie algebra, called the $\theta$-twisted loop algebra, as follows:
\begin{equation}
	{\mathcal L}({\mathfrak g},\theta) = \bigoplus_{j\in{\mathbb Z}} {\mathfrak g}(j)\otimes t^j\subset {\mathfrak g}\otimes {\mathbb C}[t,t^{-1}]\ .
\end{equation}
In the case $\theta=1$, we recover the loop algebra ${\mathcal L}({\mathfrak g})$ with the standard grading (i.e. with ${\mathfrak g}\otimes t^j$ in degree $j$).

Now let $\theta$ be an arbitrary inner automorphism.
By an important theorem of Kac \cite[Thm. 8.5]{kac1990infinite}, ${\mathcal L}({\mathfrak g},\theta)$ is isomorphic as an ungraded Lie algebra to ${\mathcal L}({\mathfrak g})$.
This isomorphism induces a non-standard grading on ${\mathcal L}({\mathfrak g})$, and hence on ${\mathfrak g}_{\rm aff}$.
After conjugating if necessary, any such grading has ${\mathfrak t}_{\rm aff}\subset {\mathfrak g}_{\rm aff}(0)$ and ${\mathfrak g}_{{\rm aff},\alpha}\subset \sum_{i\geq 0}{\mathfrak g}_{\rm aff}(i)$ for all simple roots $\alpha_0,\ldots ,\alpha_r$.
The data of such a grading therefore comes down to a non-negative integer weight for each simple root, and can be represented by a Kac diagram, i.e. a copy of the affine Dynkin diagram with the weights $(n_i)_{0\leq i\leq r}$ attached (see again \eqref{eq:E8ni} for $\mathfrak{g}=E_8$).
We thus obtain a one-to-one correspondence:
\begin{center}
	\{periodic inner automorphisms of ${\mathfrak g}$\}/conjugacy $\xleftrightarrow{\text{1-to-1}}$ \{Kac diagrams\}.
\end{center}%
Two remarks are in order:
	\begin{itemize}
		\item[\emph{i)}]
		      we here gloss over some subtleties concerning symmetries of the Dynkin diagram.
		      This will not be a problem for us since there are no such symmetries in type $E_8$.
		\item[\emph{ii)}]
		      The periodic \emph{outer} automorphisms are classified via the same procedure, leading to Kac diagrams on the \emph{twisted} affine root systems.
		      We refer to \cite{kac1990infinite} for the details (which we do not need, since there are no outer automorphisms in type $E_8$).
	\end{itemize}
\noindent Each Kac diagram $\lambda_{\rm Kac}$ in \eqref{eq:rhopart} uniquely defines a homomorphism ${\mathbb Z}\Phi_{\rm aff}\to{\mathbb Z}$ which sends $\alpha_i$ to $n_i$.
We will use the notation $\lambda_{\rm Kac}(\sum_i m_i\alpha_i)$ for this weight.  Recall that a root $\alpha\in\Phi_{\rm aff}$ is \emph{real} if it is conjugate to a simple root $\alpha_i$, and is \emph{imaginary} otherwise.  In (untwisted) affine types, there is a unique minimal positive imaginary root
\begin{equation}\label{deltadefinition}
	\delta\coloneqq\alpha_0+\sum_{i=1}^r m_i \alpha_i\ ,
\end{equation}
where $\hat\alpha=\sum_{i=1}^r m_i\alpha_i$ is the (unique) highest root in $\Phi^+$, and all imaginary roots are of the form $n\delta$ for $n\in \zz \setminus \{ 0\}$.  A Kac diagram $\lambda_{\rm Kac}$ determines a $\zz$-grading of ${\mathfrak g}_{\rm aff}$ such that the $\delta$-root space lies in ${\mathfrak g}_{\rm aff}(k)$, where
\begin{equation}
	k\coloneqq\lambda_{\rm Kac}(\delta)=n_0+\sum_{i=1}^r m_i n_i
\end{equation}
is the \emph{level} of $\lambda_{\text{Kac}}$ (once again using terminology common for 2d CFTs).
It follows that this grading induces a $\zz_k$-grading of ${\mathfrak g}$, hence an inner automorphism of ${\mathfrak g}$ of order dividing $k$.
If the $n_i$ have no common factor then this automorphism has order $k$.

Note for instance that there are (up to conjugacy) three $\zz_2$-gradings of a simple Lie algebra of type $E_8$, i.e. when $k=2$ (see section \ref{subsub:corootsfinite} for our conventions on numbering):
\begin{itemize}
	\item[\emph{i)}]
	      the trivial grading with ${\mathfrak g}={\mathfrak g}(0)$, given by $n_1=2$ and all other $n_i=0$. This is Kac diagram $\lambda_{\text{Kac}}=[1^2]$;
	\item[\emph{ii)}]
	      the grading with $n_2=1$ and all other $n_i=0$. This is Kac diagram $\lambda_{\text{Kac}}=[2]$;
	\item[\emph{iii)}]
	      the grading with $n_{2'}=1$ and all other $n_i=0$. This is Kac diagram $\lambda_{\text{Kac}}=[2']$.
\end{itemize}
For a Kac diagram $\lambda_{\text{Kac}}$ with corresponding automorphism $\theta$, the fixed point subalgebra ${\mathfrak g}^\theta={\mathfrak g}(0)$ can be read off as the reductive subalgebra generated by ${\mathfrak t}$ and the simple root elements $e_{\pm\alpha}$ such that $\lambda_{\text{Kac}}(\alpha)=0$; in this context we identify the affine root elements $e_{\pm \alpha_0}$ with $e_{\mp \hat\alpha}$, where $\hat\alpha$ is the highest root element in $\Phi^+$.
Such subalgebras are called \emph{pseudo-Levi subalgebras} in the mathematical literature.
For the three Kac diagrams of order 2 listed above, we have respectively ${\mathfrak g}(0)=E_8$, $E_7\oplus\su{2}$ and $\mathfrak{so}(16)$. Notice that this pseudo-Levi subalgebra is precisely what we called $\mathfrak{f}$ (i.e. the preserved flavor symmetry of the orbi-instanton) in section \ref{sub:eng} and \cite{Fazzi:2022hal,Fazzi:2022yca}.

\subsection{A-type orbi-instantons from M-theory}
\label{sub:eng}

Orbi-instantons are 6d $(1,0)$ SCFTs which are the data of the number $N$ of M5-branes probing an M9-wall and simultaneously the orbifold $\cc^2/\Gamma_\text{ADE}$ (with $\Gamma_\text{ADE}\subset \SU(2)$ finite), the order of the orbifold, and (the holonomy of) a flat connection  at the spatial infinity $S^3/\Gamma_\text{ADE}$ for the $E_8$ gauge bundle supported on the M9-wall in eleven dimensions.  Such flat connections are elements of $\Hom(\pi_1(S^3/\Gamma_\text{ADE}),E_8) \cong \Hom(\Gamma_\text{ADE},E_8)$.  Hereafter we are going to focus exclusively on $\Gamma_\text{ADE}=\zz_k$.  Then $\rho_\infty : \zz_k \to E_8 \in  \Hom(\zz_k,E_8)$ is a grading of $E_8$ by $\zz_k$ of the type introduced in the previous section.

The tensor branch of each such orbi-instanton can be given an F-theory description as follows \cite{DelZotto:2014hpa} (to which we refer the reader for the relevant notation):
\begin{equation}\label{eq:genelquivpreblow}
	[E_8]\, \underbrace{\overset{\su{k}}{1} \, \overset{\su{k}}{2} \cdots \overset{\su{k}}{2}}_N \, [\SU(k)]\ .
\end{equation}
However, this quiver represents only a \emph{partial} tensor branch. Depending on the specific flat connection chosen (with this data being hidden in the notation of \eqref{eq:genelquivpreblow}),
we may have to introduce extra compact curves in the base of F-theory. This is because the intersection $[E_8]\, \overset{\su{k}}{1}$ is too singular, and e.g. when $k=[1^k]$ a further $k$ blowups in the base are required in the middle of the two curves (each instance involves introducing a new $1$ curve, decorated by $\su{k − i}$, $i = 1, . . . , k$, and blowing up the ``old'' $1$ into a $2$).  More generally, the full tensor branch of the orbi-instanton is given by the following generalized F-theory quiver,
\begin{equation}\label{eq:genelequiv}
	[F]\underbrace{\overset{\mathfrak{g}}{1} \, \overset{\mathfrak{su}(m_1)}{2}\, \overset{\mathfrak{su}(m_2)}{2} \cdots \overset{\mathfrak{su}(m_{N_\rho-1})}{2}}_{\max\left(N_\rho,1\right)}\, \underbrace{\overset{\mathfrak{su}(k)}{\underset{[N_\text{f}=k-m_{N_\rho-1}]}{2}}\overset{\mathfrak{su}(k)}{2}\cdots \overset{\mathfrak{su}(k)}{2} [\SU(k)]}_{N+1}\ ,
\end{equation}
where $[F]$ is (the nonabelian part of) a maximal subalgebra $\mathfrak{f}$ of $E_8$,  and
$\mathfrak{g} \in \{\emptyset,\mathfrak{usp}(m_0),\mathfrak{su}(m_0)\}$.  For $\mathfrak{g}=\mathfrak{su}(m_0)$ we also have one hypermultiplet in the two-index antisymmetric representation of $\mathfrak{su}(m_0)$ for all $m_0\neq 6$ (and a half hypermultiplet in the three-index antisymmetric representation of $\mathfrak{su}(m_0)$ for $m_0=6$). All ranks and matter representations (i.e. hypermultiplets) are determined by the chosen $\rho_\infty$ via a simple algorithm which can be found in \cite{Mekareeya:2017jgc}.  The subalgebra $\f$ of $E_8$ that is unbroken by the Kac diagram (i.e. the pseudo-Levi subalgebra of $E_8$) is the commutant of the image $\rho_\infty(\zz_k)\subset E_8$, and its Dynkin diagram is easily obtained by deleting the nodes of affine $E_8$ appearing in the partition \eqref{eq:partitionpre} (i.e. the nodes of \eqref{eq:E8ni} whose $n_i$ are nonzero), together with an abelian subalgebra making the total rank 8 \cite[Sec. 8.6]{kac1990infinite}, i.e. a summand of the form $\bigoplus_i \mathfrak{u}(1)_i$.\footnote{The pseudo-Levi subalgebras that do not contain $\mathfrak{u}(1)$ summands are the semisimple ones. These are preserved by Kac diagrams with a single part; all other pseudo-Levi's by diagrams with more than one part. See \cite[Sec. 3.3]{Fazzi:2022hal} for the $\mathfrak{u}(1)$ summands.} We will also say that the orbi-instanton has a plateau of length $N+1$, i.e. the region where the gauge algebras all become $\su{k}$, corresponding to having stacks of $k$ D6's in that region of the Type IIA reduction of M-theory.

What is $N_\rho$ in the above expression? The F-theory description, i.e. the Type IIB background with varying axiodilaton and nonperturbative seven-branes, is T-dual to a Type IIA setup with one O8$^-$-plane and 8 D8-branes.  The different ways in which the stack of D8's split into substacks (one of which may be on top of the O8$^-$) encode the different subalgrebras $\mathfrak{f}$.  There are also $k$ D6-branes, which can end with different patterns on the D8's on the left of the configuration. (Each pattern is specified by a different $\rho_\infty$.) Moreover the D6's are suspended between NS5-branes. (For details see e.g. \cite{Fazzi:2022hal}.) For each M5-brane in the original M-theory configuration we have one NS5-brane; however, because of the orbifold, the M9 \emph{fractionates} \cite{DelZotto:2014hpa}, generating extra NS5's (contributing dynamical tensor multiplets) once we reduce to Type IIA. Then $N_\rho$ precisely captures this number.  It depends on the chosen $\rho_\infty$, and we will say that it gives the number of ``fractional instantons''.
In F-theory, it signals that the intersection between $[E_8]$ and $\overset{\su{k}}{1}$ is too singular, and needs to be blown up $N_\rho$ times to bring the model in Kodaira--Tate form.

$N_\rho$ can be determined in the following simple way in terms of $\rho_\infty$ (i.e. \eqref{eq:E8ni}):
\begin{equation}\label{eq:Nmu}
	N_{\rho}\coloneqq\sum_{i=1}^6 n_i+p\ , \quad p\coloneqq\min\left(\left\lfloor \frac{n_{3'}+n_{4'}}{2} \right\rfloor, \left\lfloor \frac{n_{2'}+n_{3'}+2n_{4'}}{3} \right\rfloor\right)\ .
\end{equation}
When the Kac diagram does not contain any primes, $N_\rho$ is identical to the total number of unprimed parts.\footnote{Notice that our $p$ already appears in \cite[Eq. (3.109)]{Cabrera:2019izd} with the same name, and in \cite{Mekareeya:2017jgc} denotes the difference between $N_S$ and $N_6$.} It may sometimes happen that $N_\rho=0$ (e.g.  picking $[4']$ for $k=4$); in this case  the above electric quiver reads instead
\begin{equation}\label{eq:genelequivNmu0}
	[F]\, \overset{\mathfrak{g}}{1} \,  \underbrace{\overset{\mathfrak{su}(m_1)}{\underset{[N_{\text{f}_1}]}{2}} \, \overset{\mathfrak{su}(m_2)}{\underset{[N_{\text{f}_2}]}{2}}\cdots \overset{\mathfrak{su}(k)}{\underset{[N_{\text{f}_\text{plateau}}]}{2}}\, \overset{\mathfrak{su}(k)}{2} \cdots \overset{\mathfrak{su}(k)}{2} [\SU(k)]}_{N+1}\ ,
\end{equation}
where the ranks and matter representations depend nontrivially on the chosen $k$ and $\mathfrak{g}$. (For concrete examples see \cite{Fazzi:2022hal}.)

Sometimes it will be convenient to call $P\coloneqq N+N_\rho$ the total number of full plus fractional instantons (see e.g. appendix \ref{app:magquivs}).

\section{Affine and double affine Grassmannians}
\label{sec:grass}

Having introduced orbi-instantons and Kac diagrams, we move on to the next dramatis personae -- \emph{affine Grassmannians}. We begin the section by introducing the affine Grassmannian of a simple Lie algebra.
We then discuss double affine Grassmannians,  whose original (in some sense conjectural) definition is due to Braverman--Finkelberg \cite{Braverman:2007dvq,braverman2011pursuing,braverman2012pursuing}, and which have subsequently been studied, in the context of 3d Coulomb branches, by the same authors in collaboration with Nakajima \cite{Braverman:2016pwk,Nakajima:2016guo,Braverman:2017ofm,Braverman:2018gvt,Nakajima:2018ohd,Nakajima:2022sbi}.

\subsection{Affine Grassmannian}\label{affinegrass}

To define the affine Grassmannian of a simple algebraic group $G$ over ${\mathbb C}$, let $\cc[\![z]\!]$ denote the ring of formal complex power series $f(z)=\sum_{i=0}^\infty a_i z^i$, and let $\cc(\!(z)\!)$ be its field of fractions, consisting of all Laurent series $h(z) = \sum_{i=N}^\infty a_i z^i$ with $N\in \zz$.
Let $G$ be an algebraic subgroup of $ \text{GL}(n,\cc)$.
By definition, $G=\{ M \in \text{GL}(n,\cc)\ |\ P_j(M) =0\}$ for some polynomials $P_j$.
(Assume $\{ P_j\}$ is a complete set, i.e. generates the ideal $I$ of all polynomials vanishing on $G$.)
The exceptional groups, in particular $E_8$, can be considered as algebraic subgroups of some $\GL(n,{\mathbb C})$ in this way.
The collection of polynomials $P_j$ allows us to define points of $G$ in an arbitrary ring containing ${\mathbb C}$.
In particular, we have: $G[\![z]\!] =\{ M \in \text{GL}(n,\cc[\![z]\!])\ |\ P_j(M) =0\;\forall\, j\}$ and $G(\!(z)\!) =\{ M \in \text{GL}(n,\cc(\!(z)\!))\ |\ P_j(M) =0\;\forall\, j\}$.
We consider $G[\![z]\!]$ as a subset of $G(\!(z)\!)$ in the canonical way.
(Note that $G(\!(z)\!)$ is sometimes called the loop group; it can be considered as a completed version of the group analog $G[z,z^{-1}]$ of the loop algebra considered in section \ref{sub:diag}.)

The affine Grassmannian of $G$ is the space of left cosets:
\begin{equation}
	\Gr_G \coloneqq G(\!(z)\!)/G[\![z]\!]\ .
\end{equation}
Left multiplication defines an action of $G[\![z]\!]$ on $\Gr_G$; note that the set of orbits for this action is in one-to-one correspondence with the set of double cosets $G[\![z]\!]\!\setminus\!\! \,G(\!(z)\!)/G[\![z]\!]$.
There is a natural topology on $\Gr_G$, and the closure $\overline{\mathcal O}$ of a $G[\![z]\!]$-orbit in $\Gr_G$ has a structure of (finite-dimensional, almost always singular) projective variety containing ${\mathcal O}$ as a Zariski-open subset.
As we recount below, there are countably many orbits, and any orbit closure $\overline{\mathcal O}$ is a union of finitely many orbits, each a locally closed subset of $\overline{\mathcal O}$.
This gives $\Gr_G$ a structure of infinite-dimensional variety, or \emph{ind-variety}.
Given any $M\in G(\!(z)\!)$, denote by $[M]$ the coset $M\cdot G[\![z]\!]\in\Gr_G$, and by $G[\![z]\!]\cdot [M]$ the corresponding $G[\![z]\!]$-orbit.

Recall that an algebraic group is {\it simple} if it is nonabelian and has no closed connected normal subgroups; this includes $E_8$.
Let $G$ be simple and let ${\mathfrak g}$ be the Lie algebra of $G$ (which is a simple Lie algebra).
We retain the notation (${\mathfrak t}$, $\Phi$, $\Phi^+$, $\alpha_i$) introduced in section  \ref{sub:diag}.
Let $T$ be a maximal torus of $G$ with Lie algebra ${\mathfrak t}$, and identify $\Phi=\Phi({\mathfrak g},{\mathfrak t})$ with the roots $\Phi(G,T)$ relative to $T$.
Each $\alpha\in\Phi$ is therefore an algebraic character $T\rightarrow {\mathbb C}^\times$, and hence the root lattice ${\mathbb Z}\Phi$ embeds in the character group $X(T)$.
Similarly, for each $\alpha$ there is a corresponding coroot $\alpha^\vee$, which we think of both as an element of ${\mathfrak t}$ and as a cocharacter, i.e. a homomorphism ${\mathbb C}^\times\rightarrow T$.
We denote the group of all cocharacters by $Y(T)$.
The coroot lattice ${\mathbb Z}\Phi^\vee$ therefore embeds in $Y(T)$.
In general, both inclusions are proper: the center $Z(G)$ is isomorphic to $X(T)/{\mathbb Z}\Phi$, and the fundamental group $\pi_1(G)$ is isomorphic to $Y(T)/{\mathbb Z}\Phi^\vee$.
Thus, $G$ is {\it simply-connected} if $Y(T)={\mathbb Z}\Phi^\vee$.
Obversely, $G$ is of \emph{adjoint type} if $X(T)={\mathbb Z}\Phi$; for any simple $G$ the quotient $G/Z(G)$ is of adjoint type.

The lattices $X(T)$ and $Y(T)$ are dual, in the following sense.
For $\xi\in X(T)$ and $\lambda\in Y(T)$ we denote by $\lambda(\xi)$ the unique integer such that $\xi(\lambda(t))=t^{\lambda(\xi)}$ for all $t\in{\mathbb C}^\times$.
This defines a {perfect pairing} $X(T)\times Y(T)\to {\mathbb Z}$, i.e. it identifies $Y(T)$ with ${\rm Hom}(X(T),{\mathbb Z})$, and vice-versa.
We can therefore identify elements of $Y(T)$ with linear functions $X(T)\rightarrow{\mathbb Z}$.
Hence we have dual inclusions:
\begin{equation}
	{\mathbb Z}\Phi\subset X(T)\subset {\rm Hom}({\mathbb Z}\Phi^\vee,{\mathbb Z})\ ,\quad {\rm Hom}({\mathbb Z}\Phi,{\mathbb Z})\supset Y(T)\supset {\mathbb Z}\Phi^\vee\ ,
\end{equation}
where the first (resp. second) inclusion in each pair is an equality if $G$ is of adjoint type (resp. simply-connected).
The homomorphisms ${\mathbb Z}\Phi^\vee\rightarrow {\mathbb Z}$ (resp. ${\mathbb Z}\Phi\rightarrow {\mathbb Z}$) are called {\it weights} (resp. {\it coweights}).
Since any coweight is determined by its values on the simple roots $\alpha_i$, the lattice of coweights has a basis consisting of the {\it fundamental coweights} $\varpi_i^\vee$ satisfying $\varpi_j^\vee(\alpha_i) = 1$ if $i=j$ and $0$ otherwise.

A coweight (hence also any cocharacter) $\lambda$ is \emph{dominant} if $\lambda(\alpha)\geq 0$ for all positive roots $\alpha$.
Clearly, the dominant coweights are those of the form $\sum_i n_i \varpi_i^\vee$ with $n_i\in {\mathbb Z}_{\geq 0}$.
To each dominant cocharacter $\lambda\in Y(T)$ is naturally associated an element $\lambda_z\in G({\mathbb C}(\!(z)\!))$ (which we can think of as ``evaluating $\lambda$ at $z$''), and hence a $G[\![z]\!]$-orbit:
\begin{equation}
	[\Gr_{G}]^\lambda \coloneqq G[\![z]\!] \cdot [\lambda_z]\ ,
\end{equation}
sometimes known as an (affine) Schubert cell.
We have $\dim_\cc [\Gr_{G}]^\lambda = 2\lambda(\rho)$ where $\rho$ is the Weyl vector, i.e.  the half sum of the positive roots.
(The Weyl vector can also be defined by $\alpha_i^\vee(\rho)=1$ for all simple coroots $\alpha_i^\vee$, that is, $\rho$ is the sum of the fundamental weights $\varpi_i$.)
It turns out that all $G[\![z]\!]$-orbits are of this form, and $[\Gr_G]^\lambda\neq [\Gr_G]^\mu$ for distinct dominant coweights $\lambda,\mu$.
When $G$ is of adjoint type, we therefore obtain the statement:
\begin{center}
	for $G$ adjoint, the $G[\![z]\!]$-orbits in $\Gr_G$ are in correspondence with dominant coweights.
\end{center}
There are finitely many Schubert cells of each dimension, hence the closure of $[\Gr_G]^\lambda$ is a union of finitely many $G[\![z]\!]$-orbits.
We write $\lambda\leq\mu$ when $[\Gr_G]^\lambda\subset\overline{[\Gr_G]^\mu}$; this defines a partial order on the dominant cocharacters.
Below we will return to this partial order, which is very well understood.

It may be worth clarifying the relationship between coroots and coweights.
For a simple root $\alpha_j$ and a simple coroot $\alpha_i^\vee$, we have $\alpha_i^\vee(\alpha_j) = C_{ij}$ is the coefficient of the Cartan matrix for ${\mathfrak g}$.
Hence, $\alpha_i^\vee=\sum_j C_{ij}\varpi_j^\vee$.
(This shows that the coroot lattice is of index $(\det C)$ in $Y(T)$.)
We can therefore express the fundamental coweights as (in general rational) linear combinations of the simple coroots: $\varpi_i^\vee=\sum_j A^{ij}\alpha_j^\vee$, where $(A^{ij})$ is the inverse of the Cartan matrix.

\subsubsection[Coweights in type $E_8$]{Coweights in type \texorpdfstring{$\bm{E_8}$}{E8}}
\label{subsub:corootsfinite}

We specialize to $G=E_8$ from now on.
In this case $G$ is both simply-connected and of adjoint type.
(This follows from the fact that the Cartan matrix has determinant $-1$.
Note that $E_8$ is the only irreducible simply-laced root system with this property.)

We use Coxeter labels $i \in \{2, 3, 4, 5, 6, 4',3', 2'\} = \Delta$ for the Dynkin diagram of $E_8$:
\begin{equation}
	\label{eq:E8}
	E_8: \quad \node{2}{\alpha_2}-\node{3}{\alpha_3}-\node{4}{\alpha_4}-\node{5}{\alpha_5}-\node{6\ver{3'}{\alpha_{3'}}}{\alpha_6}-\node{4'}{\alpha_{4'}}-\node{2'}{\alpha_{2'}}\ .
\end{equation}
As above, denote by $\alpha_i$, $\alpha_i^\vee$ and $\varpi_i^\vee$ the simple root, simple coroot and fundamental coweight indexed by $i$.
It follows from the simply-laced property that for any root $\alpha=\sum_i m_i \alpha_i$, the corresponding coroot is $\alpha^\vee=\sum_i m_i \alpha_i^\vee$.
Denote by $\hat\alpha$ the highest root in $\Phi$.
Then we have $\varpi_2^\vee = \hat{\alpha}^\vee = \sum_{i=2}^6 i \alpha_i^\vee + \sum_{i=2}^4 i \alpha_{i'}^\vee$, which in particular is dominant.
This holds more generally: if $\hat\alpha_I$ is the highest root element of any irreducible simply-laced root system $\Phi_I$, then $\hat\alpha_I^\vee$ is a dominant coweight with respect to $\Phi_I$ (and is the only dominant coroot).
We will apply this in cases where $I$ is a connected subdiagram of the (affine) Dynkin diagram of type $E_8$.
For later reference, we now express $\hat\alpha_I^\vee$ in terms of fundamental coweights.
In all cases, a node $i\in \Delta\setminus I$ is joined via at most one edge to elements of $I$.
We let $\varpi_I^{-}=\sum_{\text{edge}\, I-j} \varpi_j^\vee$.
Then $\hat\alpha_I^\vee=\varpi_I^+-\varpi_I^-$, where $\varpi_I^+$ is the unique dominant coroot within the root subsystem $\Phi_I$.
We specify $\varpi_I^+$ in each case:
\begin{itemize}
	\item
	      if $I=\{ i\}$ is a singleton set then $\varpi_I^+=2\varpi_i^\vee$;
	\item
	      if $I=\{ i,\ldots ,j \}$ is of type $A_m$ ($m\geq 2$) then $\varpi_I^+=\varpi_i^\vee+\varpi_j^\vee$ (where $i,j$ are the end nodes);
	\item
	      if $I=\{ i,i+1,\ldots ,j-2,j-1,j\}$ is of type $D_m$ (with branch node $(j-2)$ and ``tail'' $i-\cdots - (j-2)$) then $\varpi_I^+=\varpi_{i+1}^\vee$;
	\item
	      if $I=\{ 2',3',4',6,5,4\}$ is of type $E_6$ then $\varpi_I^+=\varpi_{3'}^\vee$;
	\item
	      if $I=\{ 2',3',4',6,5,4,3\}$ is of type $E_7$ then $\varpi_I^+=\varpi_{2'}^\vee$.
\end{itemize}

\subsubsection{Orbits as weighted Dynkin diagrams}
\label{subsub:orbitssubsec}

As mentioned above, the $G[\![z]\!]$-orbits $[\Gr_{E_8}]^\mu$ are in one-to-one correspondence with the dominant coweights $\mu$, i.e. the non-negative integer linear combinations
\begin{equation}
	\mu\coloneqq\sum_{i\in \Delta} n_i \varpi_i^\vee\ .
\end{equation}
The integers $n_i$ can be thought of as ``multiplicities'' of the Coxeter labels in the $E_8$ Dynkin \eqref{eq:E8}.
Similarly to the Kac diagrams studied in section \ref{sub:diag}, this allows us to represent orbits via ``weighted Dynkin diagrams''.\footnote{This is non-standard terminology: in the mathematical literature the phrase ``weighted Dynkin diagram'' usually means only those coweights arising in the classification of nilpotent orbits via $\su2$-triples.}

The closure order on orbits induces a partial order on dominant coweights.
In fact we have $[\Gr_{E_8}]^\lambda \subset \overline{[\Gr_{E_8}]^\mu}$ if and only if $\mu − \lambda$ is a non-negative linear combination of simple coroots.
A pair $(\lambda, \mu)$ with $\lambda< \mu$ is called a \emph{degeneration}; it is a minimal degeneration if there is no $\nu$ with $\lambda < \nu< \mu$.
The partial order on coweights can be represented via a Hasse diagram, where an edge connects adjacent coweights (i.e. minimal degenerations).
According to a result of Stembridge \cite{STEMBRIDGE1998340}, there are (in the simply-laced case) exactly two types of minimal degeneration $\lambda < \mu$:
\begin{enumerate}[label=\emph{\roman*})]
	\item
	      pairs $(\lambda, \mu= \lambda+ \hat{\alpha}_I^\vee)$, where $I$ is a connected component of the subdiagram of zeros of $\mu$ (i.e. zero multiplicities in the weighted Dynkin);
	\item
	      pairs $(\lambda, \mu= \lambda+ \alpha_i^\vee)$, where $n_j >0$ for all $j$'s connected to $i$ via an edge.
\end{enumerate}
Note that if we express $\lambda$ as sum $\sum_i n_i \varpi_i^\vee$ of coweights and $\mu-\lambda=\sum_i l_i \alpha_i^\vee$ as a sum of coroots, then in case \emph{i}) we have $n_i l_i=0$ for all $i\in \{ 2',3',4',6,5,4,3,2\}$.

\subsubsection{Geometry and slices in the affine Grassmannian}

Each orbit $\Gr_\mu$ in the affine Grassmannian is smooth, with singular locus equal to the boundary
\begin{equation}
	\overline{\Gr_\mu}\setminus\Gr_\mu=\bigcup_{\lambda<\mu}\Gr_\lambda\ .
\end{equation}
Suppose $\lambda<\mu=\lambda+\sum l_i \alpha_i^\vee$ is a degeneration of dominant coweights.
The geometry of $\overline{\Gr_\mu}$ in a neighbourhood of the singular point $\lambda_z$ can be understood using a transverse slice ${\mathcal S}_{\lambda,\mu}$, for which there is a standard construction.
It is known that ${\mathcal S}_{\lambda,\mu}$ is a symplectic singularity.
By an earlier remark, $\dim_{\mathbb H}({\mathcal S}_{\lambda,\mu})=(\mu-\lambda)(\rho)=\sum l_i$, i.e. the \emph{height} of $\mu-\lambda$.
In particular, in case \emph{i)} above (with $\mu-\lambda=\hat\alpha_I^\vee$) this dimension is $h_I-1$, where $h_I$ is the Coxeter number of the root system spanned by $I$; in case \emph{ii}), $dim_{\mathbb H}({\mathcal S}_{\lambda,\mu})=1$.

The main result of \cite{malin-ostrik-vybornov} is a classification of the singularity of ${\mathcal S}_{\lambda,\mu}$ for any \emph{minimal} degeneration of dominant coweights (in which case the singularity is isolated at $\lambda_z$).
Specifically, in case \emph{i}) one obtains the closure of the minimal nilpotent orbit of a simple Lie algebra of type $I$ (a so-called \emph{minimal singularity}); in case \emph{ii}) it is an $A_{n_i+1}$ type surface singularity, i.e. ${\mathbb C}^2/{\mathbb Z}_{n_i+2}$.
These results are evidently consistent with the dimension statement above.
Note that the pairs in \emph{i}) and \emph{ii}), and hence the partial order on dominant coweights, are naturally expressed in purely combinatorial terms (i.e. without reference to the affine Grassmannian).

\subsection{Double affine Grassmannian}

The affine Grassmannian has an important role in representation theory because of the geometric Satake correspondence \cite{mirkovicvilonen}, a deep theorem relating irreducible representations of ${\mathfrak g}$ to the geometry of $\Gr_G$.
The double affine Grassmannian is a (somewhat conjectural) object playing a similar role for the affine Lie algebra ${\mathfrak g}_{\rm aff}$.
To outline what is known about the double affine Grassmannian, we need to clarify what are the coweights for the corresponding affine (i.e. Kac--Moody) Lie algebra.
We continue to restrict to type $E_8$.

\subsubsection{Coroots and coweights}
\label{subsub:corootscoweightssubsec}

Extend the Dynkin diagram for $E_8$ by adding an affine node, labeled 1; set $\Delta_\text{aff} = \Delta \cup \{1\}$:
\begin{equation}\label{eq:E81}
	E_8^{(1)}: \quad \node{1}{\alpha_1}-\node{2}{\alpha_2}-\node{3}{\alpha_3}-\node{4}{\alpha_4}-\node{5}{\alpha_5}-\node{6\ver{3'}{\alpha_{3'}}}{\alpha_6}-\node{4'}{\alpha_{4'}}-\node{2'}{\alpha_{2'}}\ .
\end{equation}
Correspondingly, we introduce an affine root $\alpha_1$ and coroot $\alpha_1^\vee$, generating (along with $\Phi$ and $\Phi^\vee$) the affine root lattice $\zz \Phi_\text{aff}$ and coroot lattice $\zz \Phi_\text{aff}^\vee$. These lattices are \emph{not} dual via the affine Cartan matrix, because of the imaginary roots in $\Phi_\text{aff}$ and the central elements in $\Phi_\text{aff}^\vee$. Thus we need to extend further (in the language of Kac \cite{kac1990infinite}, to a \emph{realization} of the generalized Cartan matrix, allowing for a {perfect} pairing between weights and coweights).
Using the same notation as \cite[Ch. 6]{kac1990infinite}, we therefore introduce an additional weight $\Lambda_0$ and an additional coweight $d$. Then the weight lattice (of the affine Kac--Moody group) is $\widehat{X}_\text{aff} \coloneqq \zz \Phi_\text{aff} \oplus \zz \Lambda_0$ and the coweight lattice is $\widehat{Y}_\text{aff} \coloneqq \zz \Phi_\text{aff}^\vee \oplus \zz d$.

The elements $d$ and $\Lambda_0$ are defined such that ${\rm Hom}(\widehat{X}_\text{aff},{\mathbb Z})=\widehat{Y}_{\text{aff}}$, and vice versa.
To this end we need to describe the pairing $\widehat{X}_\text{aff} \times \widehat{Y}_\text{aff} \to \zz$.

We use the generalized Cartan matrix to define the pairing between roots and coroots:
\begin{equation}
	\alpha_j^\vee(\alpha_i)=\begin{cases} 2 & \text{if $i=j$,} \\ -1 & \text{if $i$ and $j$ are connected by an edge,} \\ 0 & \text{otherwise.} \end{cases}
\end{equation}
(Note that this pairing extends the pairing on simple roots and coroots for the finite diagram.)
We can extend this by $\alpha_i^\vee(\Lambda_0) = 0$ for $i\neq  1$ and $\alpha_1^\vee(\Lambda_0)= 1$; similarly,  $d(\alpha_i)= 0$ for $i\neq  1$ and $ d(\alpha_1)= 1$; finally $d(\Lambda_0)= 0$.  This gives the perfect pairing: with respect to the bases $\{\alpha_1, \ldots, \alpha_6, \alpha_{4'},\alpha_{3'},\alpha_{2'},\Lambda_0\}$ and $\{\alpha_1^\vee, \ldots, \alpha_6^\vee, \alpha_{4'}^\vee,\alpha_{3'}^\vee,\alpha_{2'}^\vee,d\}$, it is given by the matrix
\begin{equation}
	\begin{pmatrix}
		2  & -1 & 0  & 0  & 0  & 0  & 0  & 0  & 0  & 1 \\
		-1 & 2  & -1 & 0  & 0  & 0  & 0  & 0  & 0  & 0 \\
		0  & -1 & 2  & -1 & 0  & 0  & 0  & 0  & 0  & 0 \\
		0  & 0  & -1 & 2  & -1 & 0  & 0  & 0  & 0  & 0 \\
		0  & 0  & 0  & -1 & 2  & -1 & 0  & 0  & 0  & 0 \\
		0  & 0  & 0  & 0  & -1 & 2  & -1 & -1 & 0  & 0 \\
		0  & 0  & 0  & 0  & 0  & -1 & 2  & 0  & -1 & 0 \\
		0  & 0  & 0  & 0  & 0  & -1 & 0  & 2  & 0  & 0 \\
		0  & 0  & 0  & 0  & 0  & 0  & -1 & 0  & 2  & 0 \\
		1  & 0  & 0  & 0  & 0  & 0  & 0  & 0  & 0  & 0
	\end{pmatrix}\ .
\end{equation}
Note in particular that all entries in the matrix are integers (so the pairing is well-defined) and the determinant equals $-1$ (so the pairing is perfect, i.e. it identifies $\widehat{Y}_{\rm aff}$ with the linear functions $\widehat{X}_{\rm aff}\rightarrow {\mathbb Z}$).
The analog here of the fundamental coweights in $\zz \Phi^\vee$ is the basis for $\widehat{Y}_\text{aff}$ which is dual to the basis $\{\alpha_1,\ldots,\Lambda_0\}$ for $\widehat{X}_\text{aff}$. Let
\begin{equation}
	\xi_2 =\varpi_2^\vee +2d\ ,\quad \xi_3 =\varpi_3^\vee +3d\ ,\ \ldots \  , \quad \xi_{2'} =\varpi_{2'}^\vee +2d\ .
\end{equation}
It is easy to check that the basis for $\widehat{Y}_\text{aff}$ which is dual to $\{\alpha_1,\ldots,\Lambda_0\}$ is $\{d,\xi_2,\ldots,\xi_{2'},c\}$ where $c$ is the canonical central element
\begin{equation}
	c=\left(\sum_{i=1}^6 i\alpha_i^\vee\right) +4\alpha_{4'}^\vee +3\alpha_{3'}^\vee +2\alpha_{2'}^\vee\ .
\end{equation}
Hence we can denote $d$ by $\xi_1$.
(Note that the fundamental coweights $\varpi_i^\vee$ in ${\mathbb Z}\Phi^\vee$ are \emph{not} fundamental coweights for the affine root system; this role is now played by the $\xi_i$.)

In the affine case, a \emph{dominant coweight} is an element of $\widehat{Y}_\text{aff}$ of the form:
\begin{equation}\label{eq:lambdaaffE8}
	\lambda \coloneqq \left(\sum_{i=1}^6 n_i\xi_i\right) +n_{4'}\xi_{4'} +n_{3'}\xi_{3'} +n_{2'}\xi_{2'}+nc
\end{equation}
where all the $n_i$ are non-negative integers. (To connect to well known physics, $n$ can be an arbitrary integer, and would correspond to the energy of a 2d model. Moreover there is a natural $\zz$-shift acting on $n$ which will have a 6d interpretation in terms of small $E_8$-instanton transitions shifting $N$.) We say that $\lambda$ has \emph{level}
\begin{equation}
	k \coloneqq \left(\sum_{i=1}^6  i n_i\right) + 4n_{4'} + 3n_{3'} + 2n_{2'}\ .
\end{equation}
This gives another way of understanding the Kac diagrams mentioned in section \ref{sub:diag}.
Moreover,
\begin{equation}
	k = \lambda(\delta)\, ,
\end{equation}
where $\delta$ is the smallest positive imaginary root.\footnote{See footnote \ref{foot:root}.} Hence the dominant coweights of level $k$ can be identified with the pairs $(\lambda_\text{Kac},n)$ of a Kac diagram of level $k$ \emph{and} an arbitrary integer $n$. In the notation of \cite{braverman2011pursuing},  one writes
\begin{equation}\label{eq:lambdaover}
	\lambda_\text{Kac}=(k, \overline\lambda)\ , \quad \overline\lambda = \sum_{i\neq 1}  n_i \varpi_i^\vee = \lambda_\text{Kac} − kd\ .
\end{equation}
Note that $\lambda(\alpha)=\lambda_{\rm Kac}(\alpha)$ for any $\alpha\in\Phi_{\rm aff}$ and $\lambda(\alpha)=\overline\lambda(\alpha)$ for any $\alpha\in\Phi$.
The point of the above discussion is to highlight the three different ways of writing dominant coweights for affine $E_8$:
\begin{enumerate}[label=\emph{\roman*})]
	\item
	      for \cite{braverman2011pursuing}, these are triples $(k, \overline\lambda, n)$ where $\overline\lambda$ is a dominant coweight for the \emph{finite} diagram, of level (i.e. $\overline\lambda(\hat\alpha)$) less than or equal to $k$. In our notation, this form for $λ$ arises from the direct sum $\widehat{Y}_\text{aff}=\zz d \oplus \zz \Phi^\vee \oplus \zz c$. In particular, $\overline\lambda$ is a $\zz$-linear combination of $\alpha_2^\vee,\ldots,\alpha_{2'}^\vee$, hence also of the fundamental coweights $\varpi_2^\vee,\ldots ,\varpi_{2'}^\vee$ introduced in the previous subsection.
	\item
	      Alternatively, we can write $\lambda$ in \eqref{eq:lambdaaffE8} as $(\lambda_\text{Kac},n)$ where $\lambda_\text{Kac}$ is a Kac diagram. Here $\lambda_\text{Kac}$ is a linear combination of $\xi_i$ for $i \in \Delta_\text{aff}$, and this decomposition arises from \emph{i}) by considering $\zz d \oplus \zz \Phi^\vee$ as the span of the $\xi_i$. (This form helps to understand the connection with Kac diagrams.)
	      In particular, a given Kac diagram gives rise to infinitely many dominant coweights in $\widehat{Y}_{\rm aff}$, parametrized by $n\in{\mathbb Z}$.
	\item
	      Finally, we can express $\lambda$ as a pair $(k,\lambda_\text{comb})$, where $\lambda_\text{comb}$ is a $\zz$-linear combination of the simple coroots $\alpha_i^\vee$ for $i \in \Delta_\text{aff}$, subject to the positivity conditions $\langle \alpha_i,\lambda_\text{comb}\rangle \geq 0$ for $i\neq  1$ and $\langle \alpha_1,\lambda_\text{comb}\rangle \geq -k$. This arises from the decomposition $\widehat{Y}_\text{aff} = \zz d \oplus \zz \Phi^\vee_\text{aff}$. (In the notation of \emph{i}), we have $\zz \Phi^\vee \oplus \zz c = \zz \Phi^\vee_\text{aff}$.)
\end{enumerate}

\subsubsection{Partial order on coweights}
\label{subsub:partial}

Recall from the previous subsection that the closure order on orbits in the affine Grassmannian induces a partial order on the dominant coweights in ${\mathfrak g}$, with respect to which the adjacent pairs $(\lambda,\mu)$ were classified by Stembridge \cite{STEMBRIDGE1998340}.
We now explore the analogous partial order for dominant coweights for the \emph{affine} Lie algebra.
Similarly to the affine Grassmannian, we write $\lambda\leq\mu$ if $ \mu-\lambda$ is a non-negative integer linear combination of simple roots $\alpha_i^{\vee}$, where $i \in \Delta_\text{aff}$.
We explore the statement $\mu\geq \lambda$ in each of the settings \emph{iii})-\emph{i}) above for the dominant coweights.

\begin{itemize}
	\item[\emph{iii})]
	      If $(k, \lambda_\text{comb}) \leq (l, \mu_\text{comb})$ then clearly $k = l$. Thus, there are infinitely many connected components of $\widehat{Y}_\text{aff}$ with respect to the partial order (at least one for each $k$); the partial order arises solely from its restriction to $\zz\Phi_\text{aff}^\vee$.
	      On the other hand, the meaning of $\lambda_{\text{comb}}\leq \mu_{\text{comb}}$ is not immediately clear.
	\item[\emph{ii})] \otherlabel{bullet:ii322}{\emph{ii)}}
	      If $(\lambda_\text{Kac},n) \leq (\mu_\text{Kac},n')$ then it follows from \emph{i}) that $\lambda_\text{Kac}$ and $\mu_\text{Kac}$ are Kac diagrams of the same order; further, we clearly have $n\leq n'$.
	      A (mathematically, but not physically) trivial case is $(\lambda_\text{Kac}, n) < (\lambda_\text{Kac}, n+1)$, by addition of $c$. (We will see in section \ref{sub:k=1} that this minimal degeneration corresponds to performing a small instanton transition from a rank-$\left(n+1\right)$ E-string to a rank-$n$ one, i.e. to dissolving one M5 into flux on the M9.)
	\item[\emph{i})]
	      If $\overline\lambda$, $\overline\mu$ are dominant coweights for ${\mathfrak g}$, both of level less than or equal to $k$, then $(k,\overline\lambda,n)\leq (k,\overline\mu,n)$ if and only if
	      $\overline\lambda \leq \overline\mu$ in the partial order on $\zz\Phi^\vee$.
	      Thus (using the fact that the coroot lattice equals the coweight lattice in type $E_8$), for any dominant coweight $(k, \overline\lambda, n)$ we have $(k, \overline\lambda, n) \geq (k, 0, n)$ (since $\overline\lambda$ is dominant, so is a non-negative linear combination of the $\varpi_i^\vee$ for $i\neq 1$, hence is a non-negative linear combination of the $\alpha_i^\vee$ too). By our remark in \emph{ii}), it follows that there is exactly one component of $\widehat{Y}_\text{aff}$ for each positive value of $k$. This component is periodic with respect to the $\lambda_\text{Kac}$ component.
\end{itemize}

\subsubsection{Minimal degenerations for affine coweights}
\label{subsub:mindeg}

Let $\lambda = (\lambda_\text{Kac},n)$ be a dominant coweight in $\widehat{Y}_\text{aff}$.
As per the finite case, if $I$ is a connected proper subdiagram of the affine Dynkin diagram then we define $\hat{\alpha}_I^\vee$ to be the highest coroot in the corresponding (finite) root subsystem.
As coweights, we have $\hat\alpha_I^\vee=\xi_I^+-\xi_I^-$.
We have:
$$\xi_I^- = \left\{ \begin{array}{cc} 2\xi_1                    & \mbox{if $I=E_8$}, \\
             \sum_{\text{edge}\, j-I} \xi_j & \mbox{otherwise},\end{array}\right.$$
and $\xi_I^+=\xi_2$ if $I=E_8$ and follows exactly the same pattern as $\varpi_I^+$ (see section \ref{subsub:corootsfinite}) otherwise.
We are interested in the adjacent pairs of dominant coweights for the \emph{affine} Lie algebra of type $E_8$.
This combinatorial problem has been solved for arbitrary affine Kac--Moody Lie algebras by Roy \cite{roy}.
The results for type $E_8$ (indeed for any simply-laced case) are:

\begin{theorem}[Roy]\label{Roytheorem}
	For ${\mathfrak g}$ of type $E_8$, let $\lambda=(\lambda_{{\rm Kac}},n)=(k,\overline\lambda,n)$  be a dominant coweight for $\mathfrak{g}_{\rm aff}$.
	If $k=1$ then $\lambda=\xi_1+nc$, and the only minimal degeneration $\lambda<\mu$ is $\mu=\lambda+c$.
	For $k>1$, the minimal degenerations are:
	\begin{enumerate}[label=\roman*\emph{)}]
		\item
		      $\lambda<\lambda+ \hat{\alpha}_I^\vee$ where $I$ is a connected component of the subdiagram of zeros of $\lambda_{\rm Kac}$;
		\item
		      $\lambda<\lambda+\alpha_i^\vee$ where $\lambda(\alpha_j)>0$ for all $j$ connected to $i$ by an edge.
	\end{enumerate}
\end{theorem}
\noindent Directly generalizing the cases in section \ref{subsub:orbitssubsec}, we will refer to these as type \emph{i}) and type \emph{ii}) minimal degenerations.
As an example, consider both types in the framework of the previous subsection.
\begin{enumerate}[label=\emph{\roman*})]
	\item
	      If $I$ does not contain the affine node, then $\mu = (k,\overline\lambda+ \hat{\alpha}_I^\vee, n)$, so this can be identified with a degeneration of coweights for the finite root system of type $E_8$.  If $I$ does contain the affine node then $c-\hat{\alpha}_I^\vee$ is a positive coroot and $\mu= (k, \overline\lambda + c − \hat{\alpha}_I^\vee, n + 1)$.
	\item
	      If $i$ is not the affine node then $\mu=(k,\overline\lambda+\alpha_i^\vee,n)$. If $i$ is the affine node then $\mu=(k,\overline\lambda−\hat{\alpha}^\vee,n+1)$.
\end{enumerate}

\subsubsection[Relationship between $N$ and $n$]{Relationship between \texorpdfstring{$\bm{N}$}{N} and \texorpdfstring{$\bm{n}$}{n}}
\label{subsub:relation}

The dominant coweights are the triples $(k,\overline\lambda,n)$, with $\overline\lambda$ a dominant coweight on the finite root system satisfying $\overline\lambda(\hat\alpha)\leq k$.
The subset of dominant coweights with fixed values of $k$ and $n$ has a unique minimal element ($\overline\lambda = 0$), hence is connected.
For instance, for $k=4$ this leads to a Hasse diagram such as \cite[Fig. 12]{Fazzi:2022hal} (which is also identical to the Higgs branch RG flow hierarchy of \cite[Fig. 1]{Frey:2018vpw} obtained via 't Hooft anomaly matching).

In \cite[Fig. 11 \& Figs. 13–25]{Fazzi:2022hal}, different hierarchies were considered: in that case the number of full and fractional instantons was kept constant.
(Note that $N_\text{\cite{Fazzi:2022hal}} = N_6 \neq N_\text{here}$, i.e. the number of full instantons, cf. footnote \ref{foot:N6} below.)
Given the full periodic poset of coweights of (\emph{fixed}) level $k$, the value of $n$ for a given coweight $\lambda$ is equal to the maximum $n$ such that $(k,0,n)\leq \lambda$.
This stratification (by value of $n$) is therefore very easy to understand in terms of the combinatorics of coweights.
On the other hand, we are interested in Coulomb branches ${\mathcal M}_\text{C}(\mu,\lambda)$ with $\overline\mu=0$, so, given any $\lambda$, we want to understand the {\it minimal} $n$ such that $(k,0,n)\geq \lambda$.
Note that the hierarchies in \cite{Fazzi:2022hal} each contained a unique maximal node, with Kac diagram $[1^k]$, i.e. with $\overline\mu=0$.
(These were however depicted with $\mu$ at the bottom of the Hasse diagram. With this convention, the direction of flow to the IR -- i.e. smaller $a$ -- is downwards; note that the IR strata, i.e. Higgs branches, are of \emph{larger} dimension, as already noted in the introduction.)

In all of the following, integers $i'$ are interpreted in formulas ``without primes'' when required; the notation $\lceil x\rceil$, resp. $\lfloor x\rfloor$ means the smallest integer greater than or equal to $x$, resp. the greatest integer less than or equal to $x$.

\begin{lemma}\label{sigmalem}
	Let $\lambda=(k,\overline\lambda,n)$ be a dominant coweight for ${\mathfrak g}_{\rm aff}$.
	Let $n_i$ be the coefficients of $\lambda_{\rm Kac}$, let $\overline\lambda=\sum_{i\neq 1} m_i\alpha_i^\vee$, and let
	\begin{equation}
		N_\rho = \sum_{i=1}^6 n_i+p \ , \quad p=\min\left(\left\lfloor\frac{n_{3'}+n_{4'}}{2}\right\rfloor,\left\lfloor\frac{n_{2'}+n_{3'}+2n_{4'}}{3}\right\rfloor\right)
	\end{equation}
	as in \eqref{eq:Nmu}.
	Then:
	\begin{enumerate}[label=\alph*\emph{)}]
		\item
		      The minimum value of $m$ such that $(k,0,m)\geq (k,\overline\lambda,n)$ is $n+\max_{i\neq 1} \lceil \frac{m_i}{i}\rceil$;
		\item
		      in terms of the Kac coefficients, this value is $n+k-N_\rho$.
	\end{enumerate}
\end{lemma}

\begin{proof}
	We may clearly assume $n=0$.
	By the discussion in section \ref{subsub:corootscoweightssubsec}, we have $\overline\lambda=\sum_{i\neq 1} n_i\varpi_i^\vee$.
	Hence we have to find the smallest value of $m$ such that $mc\geq \sum_{i\neq 1} n_i \varpi_i^\vee$.
	The equality $m=\max_{i\neq 1} \lceil\frac{m_i}{i}\rceil$ is clear, so we only have to prove (b).
	It is easily observed from the inverse Cartan matrix that $$ic-\varpi_i^\vee=\left\{ \begin{array}{ll} c+\alpha_{i-1}^\vee+2\alpha_{i-2}^\vee+\ldots+(i-1)\alpha_1^\vee            & \mbox{if $i=2,\ldots ,6$,} \\
             \alpha_{3'}^\vee+\alpha_{4'}^\vee+2\sum_{i=1}^6 \alpha_i^\vee                    & \mbox{if $i=2'$,}          \\
             \alpha_{2'}^\vee+\alpha_{3'}^\vee+2\alpha_{4'}^\vee+3\sum_{i=1}^6\alpha_i^\vee   & \mbox{if $i=3'$,}          \\
             \alpha_{2'}^\vee+2\alpha_{3'}^\vee+2\alpha_{4'}^\vee+4\sum_{i=1}^6 \alpha_i^\vee & \mbox{if $i=4'$.}
		\end{array}\right.$$
	In particular, this implies that $m\leq 2n_{2'}+3n_{3'}+4n_{4'}+\sum_{i=1}^6 (i-1)n_i = k-\sum_{i=1}^6 n_i$.
	Then:
	\begin{align}
		\left(k-\sum_{i=1}^ 6 n_i\right)c-\overline\lambda = & \  (n_{3'}+n_{4'})\alpha_{2'}^\vee + (n_{2'}+n_{3'}+2n_{4'})\alpha_{3'}^\vee + (n_{2'}+2n_{3'}+2n_{4'})\alpha_{4'}^\vee\ + \nonumber \\ &+ \sum_{i=1}^6 \left(2n_{2'}+3n_{3'}+4n_{4'}+n_{i+1}+2n_{i+2}+\ldots + (6-i)n_6\right)\alpha_i^\vee
	\end{align}
	Note that $\frac{n_{3'}+n_{4'}}{2}\leq \frac{n_{2'}+2n_{3'}+2n_{4'}}{4}$ and $\frac{n_{2'}+n_{3'}+2n_{4'}}{3}\leq \frac{2n_{2'}+3n_{3'}+4n_{4'}}{6}$.
	It follows that the smallest value of $m$ such that $mc\geq\overline\lambda$ is $k-\sum_{i=1}^6 n_i-\min (\lfloor \frac{n_{3'}+n_{4'}}{2}\rfloor, \lfloor \frac{n_{2'}+n_{3'}+2n_{4'}}{3}\rfloor)$, which equals $k-N_\rho$, as required.
\end{proof}

\begin{example}
	\begin{enumerate}[label=\emph{\alph*})]
		\item
		      Let $\lambda = (\lambda_\text{Kac}, 0)$ where $\lambda_\text{Kac}$ is the Kac diagram $[3',2'^2]$.
		      This can be thought of as $\varpi_{3'}^\vee +2\varpi_{2'}^\vee = \xi_{3'}+2\xi_{2'}-7\xi_1$.  Then $k = 7$ and $N_\rho=\sum_{i=1}^6 n_i=p=0$.
		      Concretely:
		      \begin{equation}
			      \overline\lambda =\varpi_3^\vee +2\varpi^\vee_2 =7\alpha^\vee_2 +14\alpha^\vee_3 +21\alpha^\vee_4 +28\alpha^\vee_5 +35\alpha^\vee_6 +24\alpha^\vee_{4'} +18\alpha^\vee_{3'} +13\alpha^\vee_{2'}\ ,
		      \end{equation}
		      and this is the unique expression in $\zz\Phi^\vee$ for $\overline\lambda$. The smallest multiple $mc$ satisfying $mc\geq \overline\lambda$ is $m=7$, obtained from $\lceil \frac{13}{2}\rceil = 7=k-N_\rho$.
		\item
		      Consider instead $\mu=([4',3'],0)$.
		      Again $k=7$ and $\sum_{i=1}^6 n_i=0$, but here $N_\rho=1$.
		      We have
		      \begin{equation}
			      \overline\mu = \varpi_{4'}^\vee+\varpi_{3'}^\vee = 7\alpha_2^\vee+14\alpha_3^\vee+21\alpha_4^\vee+28\alpha_5^\vee+35\alpha_6^\vee+24\alpha_{4'}^\vee+18\alpha_{3'}^\vee+12\alpha_{2'}^\vee\ ,
		      \end{equation}
		      with only the coefficient of $2'$ changing. (We see here that $\overline\mu<\overline\lambda=\overline\mu+\alpha_{2'}^\vee$ is a minimal degeneration.)
		      In this case we can read off that $m=6$.
		\item
		      Finally, consider $\nu_{\rm Kac}=[3',2',1^2]$.
		      Then $k=7$ and $p=0$ as in (a), but now $N_\rho=2$.
		      We have:
		      \begin{equation}
			      \overline\nu=\varpi_{3'}^\vee+\varpi_{2'}^\vee = 5\alpha_2^\vee+10\alpha_3^\vee+15\alpha_4^\vee+20\alpha_5^\vee+25\alpha_6^\vee+17\alpha_{4'}^\vee+13\alpha_{3'}^\vee+9\alpha_{2'}^\vee\ ,
		      \end{equation}
		      and the value of $m$ is $5$ (seen from the values of $m_6$, $m_{4'}$, $m_{3'}$ and $m_{2'}$).
	\end{enumerate}
\end{example}
\noindent Now we note that each time the number $n+k-N_\rho$ increases by $1$, the number $N_\text{\cite{Fazzi:2022hal}}$ \emph{decreases} by 1.
Recalling also that $N_\text{\cite{Fazzi:2022hal}} = N_\text{here}+N_\rho$, we therefore set $N_\text{here} =-n$.\footnote{In the notation of \cite{Mekareeya:2017jgc}, we therefore have $n=-N_3=-N_{\text{here}}$, but $N_\text{\cite{Fazzi:2022hal}}=N_6$.
We remark however that the $a$ anomaly formula in \cite{Fazzi:2022hal} was expressed in terms of $N_3$.\label{foot:N6}}

We close this section with an explanation of the difference between the hierarchies of Higgs branch RG flows that have already appeared in \cite{Fazzi:2022hal} for any $k$ and in \cite{Frey:2018vpw,Giacomelli:2022drw} for a few chosen values of $k$, since they use different notions of $N$.
\begin{itemize}
	\item
	      In \cite{Fazzi:2022hal} we have kept the total number of instantons (full and fractional), i.e.  $P=N_\rho+N$ in our present notation, fixed (for ease of presentation, and in analogy with what was done in \cite{Mekareeya:2017jgc}).  Then, since in those flows the number $N_\rho$ of fractional instantons  increases along the Higgs branch RG flow,  to compensate this increment the number $N$ of full instantons has to decrease. So, we are performing a number of small $E_8$-instanton transitions down the hierarchy.  As a result,  $n$ \emph{increases} (or, more precisely, it does not decrease).

	      In that paper we explicitly verified that $\Delta a>0$ for all such flows.  The allowed transitions (compatibly with the $a$-theorem) between orbi-instantons are determined via 3d magnetic quiver subtraction. Such hierarchies can be understood as connected subdiagrams of the Hasse diagrams at fixed level $k$ that we will present below. They are obtained at fixed values of $n+k-N_\rho$, so by slicing by minimal coweight with Kac diagram $\mu=[1^k] \geq \lambda$.

	\item
	      In \cite{Frey:2018vpw}, which only analyzed the $k=4$ case, the number $N$ of full instantons (and consequently $n$) is held fixed (so there are no small $E_8$-instanton transitions involved), and the flow is between orbi-instantons from higher to lower number of fractional instantons, i.e.  $N_\rho$ decreases (or rather, does not increase) as well as the total number $P=N+N_\rho$, as a consequence.

	      Again,  in \cite{Fazzi:2022hal} we have verified that $\Delta a>0$ for this second choice. The allowed transitions are determined via a 't Hooft anomaly matching analysis.\footnote{Essentially, one subtracts from the 6d anomaly polynomial of the UV SCFT, the one defined by Kac diagram $\mu=[1^k]$,  the anomaly polynomial of all other IR SCFTs (defined by the other allowed $\rho_\infty$'s), and checks which subtractions are perfect squares, according to a criterion of \cite{Intriligator:2014eaa,Heckman:2015ola}.  All such subtractions represent allowed flows from UV to IR SCFTs.}

	\item
	      Finally, \cite{Giacomelli:2022drw} constructed RG flows for a few values of $k$, and both $N$ and $N_\rho$ are non-increasing along the flow. Again, this is done via 3d magnetic quiver subtraction.
\end{itemize}
See figure \ref{fig:flows} for a pictorial representation when $k=4$.  See figure \ref{fig:frey-rudelius} for the actual hierarchy of RG flows in the prescription of the first, resp. second, bullet here above.%
\begin{figure}[t]
	\centering
	\includegraphics[width=0.9\textwidth]{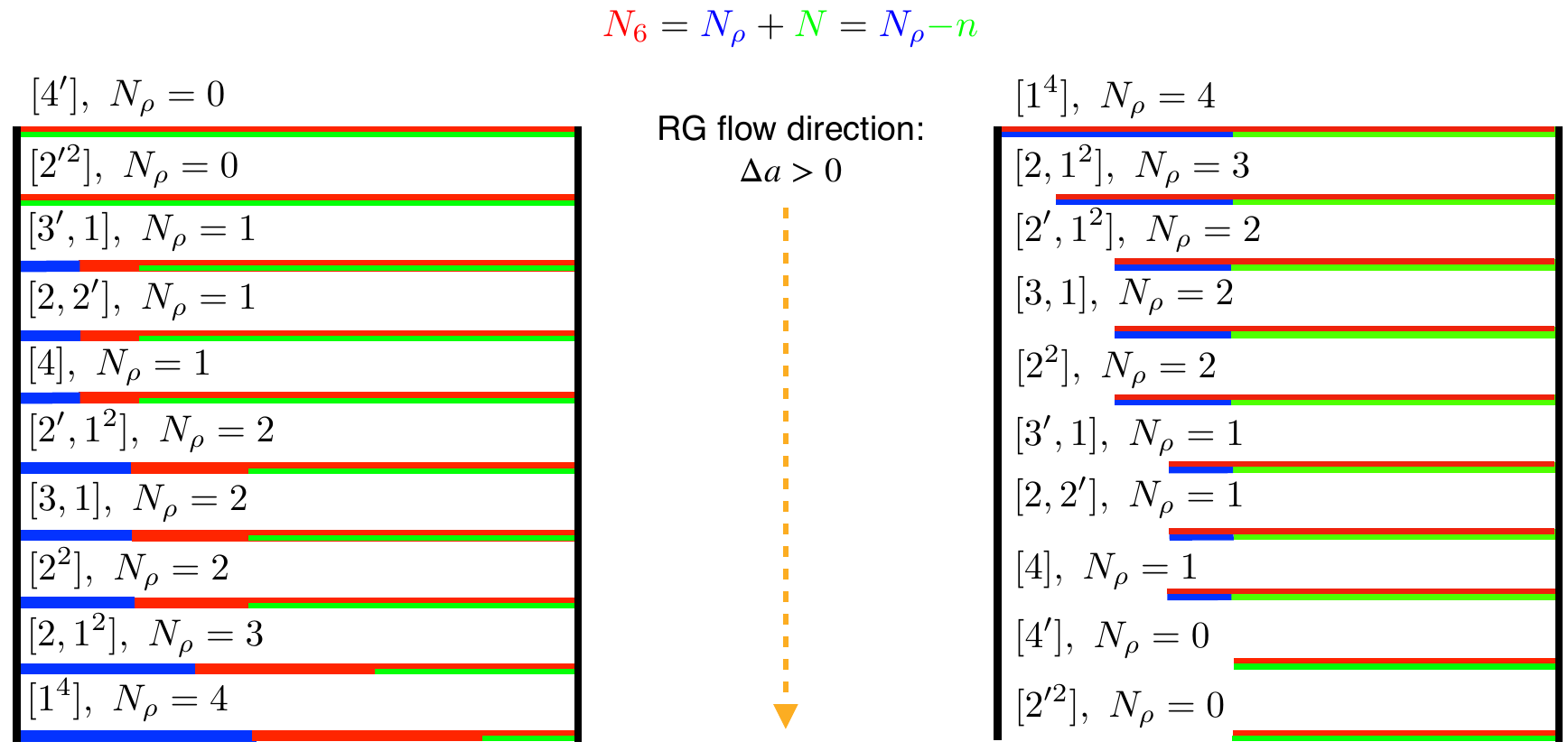}
	\caption{The hierarchy of Higgs branch RG flows for $k=4$ from \cite[Fig. 11(c)]{Fazzi:2022hal} (left) and \cite[Fig. 12]{Frey:2018vpw} (right), which are also shown in figure \ref{fig:frey-rudelius} (on the left and right respectively), and with red and blue nodes respectively in the right panel of \fref{fig:k34notvec}. (Notice that in the above diagrams we are disregarding branchings between flows.) In the first the value of $n$ increases (as $-n$ decreases), whereas it remains constant in the second.  The $a$ conformal anomaly decreases from top to bottom along the yellow arrow.}
	\label{fig:flows}
\end{figure}%
\begin{figure}[t]
	\centering
	\hskip -15pt
	\subfigure[$P=N+N_ρ$ fixed.]{
		\begin{tikzpicture}[node distance=45pt, every node/.style={scale=0.7}]
			\tikzstyle{arrow} = [thick,->,>=stealth]
			\node (n1) {$\overset{(N,N_\rho =0)}{[4']\leftrightarrow \f =\su{8}\oplus\su{2}}$};
			\node (n2) [below left=20pt and -10pt of n1] {$\overset{(N,N_\rho =0)}{[2'^2]\leftrightarrow \f =\so{16}}$};
			\node (n3) [below right=20pt and -20pt of n1] {$\overset{(N-1,N_\rho =1)}{[3',1]\leftrightarrow \f =\su{8}\oplus \uu{1}}$};
			\node (n4) [below=60pt of n1] {$\overset{(N-1,N_\rho =1)}{[2,2'] \leftrightarrow \f =\so{12}\oplus\su{2}\oplus\uu{1}}$};
			\node (n5) [below left=20pt and -20pt of n4] {$\overset{(N-1,N_\rho =1)}{[4]\leftrightarrow \f =\so{10}\oplus \su{4}}$};
			\node (n6) [below right=20pt and -20pt of n4] {$\overset{(N-2,N_\rho =2)}{[2^\prime, 1²]\leftrightarrow \f =\so{14}}$};
			\node (n7) [below= 60pt of n4] {$\overset{(N-2,N_\rho =2)}{[3,1]\leftrightarrow \f =E_6 \oplus \su{2} \oplus \uu{1}}$};
			\node (n8) [below of=n7] {$\overset{(N-2,N_\rho =2)}{[2²]\leftrightarrow \f =E_7 \oplus \su{2}}$};
			\node (n9) [below of=n8] {$\overset{(N-3,N_\rho =3)}{[2,1²]\leftrightarrow \f = E_7 \oplus \uu{1}}$};
			\node (n10) [below of=n9] {$\overset{(N-4,N_\rho =4)}{[1^4] \leftrightarrow \f = E_8}$};

			\draw [arrow] (n1)--(n2) node [midway, above left] {} node[midway, below right] {$\mathfrak{a}_1$};

			\draw [arrow] (n1)--(n3) node [midway, above right] {} node[midway, below left] {$\mathfrak{a}_7$};

			\draw [arrow] (n2)--(n4) node [midway, left] {} node[midway, above right] {$\mathfrak{d}_8$};

			\draw [arrow] (n3)--(n4) node [midway, right] {} node[midway, above left] {$\mathfrak{a}_7$};

			\draw [arrow] (n4)--(n5) node [midway, above left] {} node[midway, below right] {$\mathfrak{d}_6$};

			\draw [arrow] (n4)--(n6) node [midway, above right=0pt and 1pt] {} node[midway, below left] {$A_1$};

			\draw [arrow] (n5)--(n7) node [midway, left] {} node[midway, above right] {$\mathfrak{a}_3$};

			\draw [arrow] (n6)--(n7) node [midway, right] {} node[midway, above left] {$\mathfrak{d}_7$};

			\draw [arrow] (n7)--(n8) node [midway, left] {} node[midway, right] {$\mathfrak{a}_1$};

			\draw [arrow] (n8)--(n9) node [midway, left] {} node[midway, right] {$A_1$};

			\draw [arrow] (n9)--(n10) node [midway, left] {} node[midway, right] {$A_3$};
		\end{tikzpicture}
	}
	\subfigure[$N$ fixed.]{
		\begin{tikzpicture}[node distance=45pt,  every node/.style={scale=0.7}]
			\tikzstyle{arrow} = [thick,->,>=stealth]
			\node (n1) {$\overset{(N,N_\rho =4)}{[1^4] \leftrightarrow \f = E_8}$};
			\node (n2) [below of=n1] {$\overset{(N,N_\rho =3)}{[2,1²]\leftrightarrow \f = E_7 \oplus \uu{1}}$};
			\node (n3) [below of=n2] {$\overset{(N,N_\rho =2)}{[2^\prime, 1²]\leftrightarrow \f =\so{14}}$};
			\node (n4) [below of=n3] {$\overset{(N,N_\rho =2)}{[3,1]\leftrightarrow \f =E_6 \oplus \su{2} \oplus \uu{1}}$};
			\node (n5) [below left=20pt and -20pt of n4] {$\overset{(N,N_\rho =2)}{[2²]\leftrightarrow \f =E_7 \oplus \su{2}}$};
			\node (n6) [below right=20pt and -20pt of n4] {$\overset{(N,N_\rho =1)}{[3',1]\leftrightarrow \f =\su{8}\oplus \uu{1}}$};
			\node (n7) [below= 60pt of n4] {$\overset{(N,N_\rho =1)}{[2,2'] \leftrightarrow \f =\so{12}\oplus\su{2}\oplus\uu{1}}$};
			\node (n8) [below of=n7] {$\overset{(N,N_\rho =1)}{[4]\leftrightarrow \f =\so{10}\oplus \su{4}}$};
			\node (n9) [below of=n8] {$\overset{(N,N_\rho =0)}{[4']\leftrightarrow \f =\su{8}\oplus\su{2}}$};
			\node (n10) [below of=n9] {$\overset{(N,N_\rho =0)}{[2'^2]\leftrightarrow \f =\so{16}}$};

			\draw [arrow] (n1)--(n2) node [midway, left] {} node[midway, right] {$\mathfrak{e}_8$};

			\draw [arrow] (n2)--(n3) node [midway, left] {} node[midway, right] {$\mathfrak{e}_7$};

			\draw [arrow] (n3)--(n4) node [midway, left] {} node[midway, right] {$\mathfrak{d}_7$};

			\draw [arrow] (n4)--(n5) node [midway, above left] {} node[midway, below right] {$\mathfrak{a}_1$};

			\draw [arrow] (n4)--(n6) node [midway, above right] {} node[midway, below left] {$\mathfrak{e}_6$};

			\draw [arrow] (n5)--(n7) node [midway, left] {} node[midway, above right] {$\mathfrak{e}_7$};

			\draw [arrow] (n6)--(n7) node [midway, right] {} node[midway, above left] {$\mathfrak{a}_7$};

			\draw [arrow] (n7)--(n8) node [midway, left] {} node[midway, right] {$\mathfrak{d}_6$};

			\draw [arrow] (n8)--(n9) node [midway, left] {} node[midway, right] {$\mathfrak{d}_5$};

			\draw [arrow] (n9)--(n10) node [midway, left] {} node[midway, right] {$\mathfrak{a}_1$};
		\end{tikzpicture}
	}
	\caption{Hierarchies of left flavor Higgsings for $k=4$: on the left $P=N+N_\rho$ is held fixed (both $N_\rho$ and $n$ do not decrease), on the right only $N$ is fixed (both $N_\rho$ and $n$ do not increase).
	}
	\label{fig:frey-rudelius}
\end{figure}%

In sum, because of the periodic structure of the double affine Grassmannian's Hasse diagram, we can decrease $N_\rho$ and $N$ simultaneously: within a fixed-$N$ slicing of it (i.e. subdiagram) only $N_\rho$ is reduced, but we can also ``transition'' from an orbi-instanton at rank $N$ (i.e. with $N$ M5-branes) to another at $N-1$ (in general defined by a different Kac diagram) by performing a small $E_8$-instanton transition.
In the next two sections we will construct such Hasse diagrams and prove that $\Delta a>0$ for both types of transition.

\subsubsection{Geometry of the double affine Grassmannian}
\label{subsub:newdeg}

As remarked earlier, the symplectic leaves in a slice ${\mathcal M}_\text{C}(\mu,\lambda)$ in the double affine Grassmannian (when $k>1$) are expected to be indexed by two pieces of data,  see again \eqref{eq:strata}:
\begin{itemize}
	\item
	      coweights $\nu$ with $\lambda\leq \nu\leq \mu$. Each of them corresponds to a different holonomy at infinity $\rho_\infty$ defined by $\nu_\text{Kac}$ with $\nu=(\nu_\text{Kac},n)$, i.e.  identifies the orbi-instanton triple $(N,k,\rho_\infty)$: we have $N$ full M5's at the singularity within the M9 (fractionating into $N_\rho$ bits).
	\item
	      For each such $\nu$, and for each $M$ such that $\nu\leq\mu-Mc$, a partition $[m_i]$ of $M$. Namely, we remove a total of $M$ instantons (full M5's \emph{or} fractions of M9) from the singularity while still keeping them on the wall (see again figure \ref{fig:decoupling}). In other words, we have a decoupled product of a lower-rank orbi-instanton defined by the triple $(N-M,k,\rho_\infty)$ and a bunch of rank-$m_i$ E-strings such that $\sum_{i=1}^m m_i = M$.
\end{itemize}
Let us denote by ${\mathcal M}^{[m_i]}_\text{C}(\nu,\lambda)$ the symplectic leaf corresponding to $\nu$ and the partition $[m_i]$.
Note that we count the case $M=0$, i.e. the trivial partition $\emptyset$, and we have ${\mathcal M}_\text{C}^\emptyset(\nu,\lambda)={\mathcal M}_\text{C}^\text{smooth}(\nu,\lambda)$; in general we have ${\mathcal M}_\text{C}^{[m_i]}(\nu,\lambda)\cong {\mathcal M}_\text{C}^\text{smooth}(\nu,\lambda)\times {\rm Sym}_{[m_i]}^M({\mathbb C}^2/{\mathbb Z}_k \! \setminus\! \{ 0\})$, as noted in \eqref{eq:stratMC}.
The partial order on symplectic leaves restricts to the dominance order on coweights outlined in section \ref{subsub:partial}.
The dominance order also replicates some of the partial order for nontrivial partitions: if $\lambda\leq \xi\leq \nu\leq\mu-Mc$ then
\begin{equation}
	\overline{{\mathcal M}_\text{C}^{[m_i]}(\nu,\lambda)}\supset {\mathcal M}_\text{C}^{[m_i]}(\xi,\lambda)\ .
\end{equation}
There are three further transitions generating the full Hasse diagram of symplectic leaves in ${\mathcal M}_\text{C}(\mu,\lambda)$, as we now explain.

\begin{enumerate}[label=\emph{\roman*})]
	\item[a)]
	      Firstly, let us subdivide the partition $[m_i]$ into two subpartitions $[m_i^{(1)}]$, $[m_j^{(2)}]$ of $M_1$, $M_2$ respectively (with $M_1+M_2=M$).
	      Clearly, $\nu\leq\mu-Mc$ implies $(\nu+M_2 c)\leq\mu-M_1 c$, so there is a copy of ${\rm Sym}^{M_1}({\mathbb C}^2/{\mathbb Z}_k)$ lying between the (unadorned) strata with coweights $\nu+M_2 c$ and $\mu$.
	      In particular, this gives rise to the symplectic leaf ${\mathcal M}_\text{C}^{[m_i^{(1)}]}(\nu+M_2 c,\lambda)$, the closure of which should contain ${\mathcal M}_\text{C}^{[m_i]}(\nu,\lambda)$.
	      This aspect of the partial order was considered in depth in \cite{Bourget:2022ehw,Bourget:2022tmw} (see e.g. \cite[Fig. 5(a)]{Bourget:2022tmw}).
	\item[b)]
	      On the other hand, there is a copy of ${\rm Sym}^{M_2}({\mathbb C}^2/{\mathbb Z}_k)$ lying between the strata with coweights $\nu$ and $\mu-M_1 c$ and hence also between $\nu$ and $\mu$.
	      By the above, ${\rm Sym}^M({\mathbb C}^2/{\mathbb Z}_k)$ contains a copy of ${\rm Sym}^r({\mathbb C}^2/{\mathbb Z}_k)$ for all $r\leq M$, and the inclusions $\{ 0\}\subset {\rm Sym}^1({\mathbb C}^2/{\mathbb Z}_k)\subset {\rm Sym}^2({\mathbb C}^2/{\mathbb Z}_k)\subset\ldots$ are the obvious ones, where the extra component at each step is $0$.
	      The symplectic leaf ${\mathcal M}_\text{C}^{[m_i^{(1)}]}(\nu,\lambda)$ is contained in the closure of ${\mathcal M}_\text{C}^{[m_i]}(\nu,\lambda)$.
	\item[c)]
	      Finally, if we fix the coweight $\nu$ and the integer $M$ such that $\nu\leq\mu-Mc$, then the strata are in one-to-one correspondence with the partitions of $M$.
	      Joining any of the parts of the partition $[m_i]$ together (always possible unless $[m_i]=[M]$), one obtains a coarser partition $[m'_i]$ of $M$; then the closure of ${\mathcal M}_\text{C}^{[m_i]}(\nu,\lambda)$ contains ${\mathcal M}_\text{C}^{[m'_i]}(\nu,\lambda)$.
\end{enumerate}

\noindent Considering the last two types of transition, and taking all triples $(\nu,M,[m'_i])$ where $\nu\leq\mu-Mc$ and $[m'_i]$ is a partition of $M$, we obtain a subhierarchy as in \cite[Fig. 5(b)]{Bourget:2022tmw}.
In fact, a transverse slice from ${\mathcal M}_\text{C}^\text{smooth}(\nu,\lambda)$ to $\overline{{\mathcal M}_\text{C}^{[1^M]}(\nu,\lambda)}$ is expected to be isomorphic to $\Sym^M({\mathbb C}^2/{\mathbb Z}_k)$.
This symplectic singularity is well understood (for example, it admits a symplectic resolution).

The above considerations lead us to the classification of minimal degenerations of strata (i.e. edges in the Hasse diagram), according to four basic types.
We can describe the singularity (of a transverse slice) in each case.
Recalling that the slice from ${\mathcal M}_\text{C}^{[m_i]}(\nu,\lambda)$ is independent of $\lambda$, it follows that the singularity associated to each minimal degeneration is also independent of the choice of $\lambda$.  (In the following, $[m_i]$ is a partition of $M$, with the parts taken in any order.)
\begin{enumerate}[label=\emph{\roman*})]
	\item \otherlabel{bullet:i325}{\emph{i)}}
	      \azure{Changing boundary condition}: degenerations $\overline{{\mathcal M}_\text{C}^{[m_i]}(\nu,\lambda)}\supset {\mathcal M}_\text{C}^{[m_i]}(\xi,\lambda)$, where $\nu > \xi$ is a minimal degeneration of coweights and $\nu\leq\mu-Mc$.
	      The partition $[m_i]$ has no impact on the singularity: if $\nu=\xi+\hat\alpha_I^\vee$ where $I$ is a connected subdiagram of the affine Dynkin diagram, then we obtain a minimal singularity of type $I$; if $\nu=\xi+\alpha_i^\vee$ then we obtain the surface singularity ${\mathbb C}^2/{\mathbb Z}_{n_i+2}$.  Physically, we are simply changing the holonomy at infinity $\rho_\infty$ on the M9-wall from $\xi_\text{Kac}$ to $\nu_\text{Kac}$, without modifying the number of decoupled E-strings.
	\item \otherlabel{bullet:ii325}{\emph{ii)}}
\red{Performing small $E_8$-instanton transitions}: degenerations $\overline{{\mathcal M}_\text{C}^{[m_i,1^{l-1}]}(\nu,\lambda)}\supset {\mathcal M}_\text{C}^{[m_i,1^l]}(\nu-c,\lambda)$, where $[m_i,1^{l-1}]$ is a partition of $M$ (all $m_i>1$) and $\nu\leq\mu-Mc$.
	      The singularity is a union of $l$ copies ({\it branches}) of the minimal singularity ${\mathfrak e}_8$, where $l$ is the multiplicity of $1$ in the partition $[m_i,1^l]$ (see e.g. \cite[Fig. 5(a)]{Bourget:2022tmw}).
	      (Note the presence of branching, which does not occur in type \ref{bullet:i325} above.) Physically, we are performing a small $E_8$-instanton transition, i.e. we are dissolving one out of $l$ separated M5's into flux.
	\item \otherlabel{bullet:iii325}{\emph{iii)}}
	     	     \orange{Separating $m_i$ additional M5's}: degenerations $\overline{{\mathcal M}_\text{C}^{[m_i]}(\nu,\lambda)}\supset {\mathcal M}_\text{C}^{[\ldots ,m_{i-1},m_{i+1},\ldots]}(\nu,\lambda)$ (including as a special case $\overline{{\mathcal M}_\text{C}^{[M]}(\nu,\lambda)}\supset {\mathcal M}_\text{C}^\text{smooth}(\nu,\lambda)$).
	      The singularity in all such cases is ${\mathbb C}^2/{\mathbb Z}_k$ (where $k$ is the level of $\lambda$).  Physically,  we are removing a stack of $m_i$ M5's from the singularity while still keeping them on the wall, i.e. we are adding a rank-$m_i$ E-string.
	\item \otherlabel{bullet:iv325}{\emph{iv)}}
	     \nnyellow{Splitting $m_i+m_{i+1}$ separated M5's}: degenerations $\overline{{\mathcal M}_\text{C}^{[m_i]}(\nu,\lambda)}\supset{\mathcal M}_\text{C}^{[\ldots ,m_i+m_{i+1},\ldots]}(\nu,\lambda)$ (where there exist at least two parts of the partition $[m_i]$, and the parts can be in any order).
	      Let $l$ be the multiplicity of $m_i+m_{i+1}$ in $[\ldots ,m_i+m_{i+1},\ldots]$.
	      Then the singularity associated to this degeneration has $l$ isomorphic branches: these are given by $A_1={\mathbb C}^2/{\mathbb Z}_2$ if $m_i=m_{i+1}$, and by the non-normal singularity $m$ from \cite{FJLS} otherwise.  Physically,  we are splitting one of the $l$ substacks of $m_i+m_{i+1}$ separated M5's into two smaller substacks containing $m_i$ and $m_{i+1}$ branes each.
\end{enumerate}
In figure \ref{fig:degs} we sketch the M-theory brane realization of the four degenerations.
Note that a mathematical proof of these results depends on extending \cite[Thm. 7.26]{Nakajima:2016guo} to affine $E_8$. 
However, the classification of the singularities is not needed for our proof of the $a$-theorem (although the classification (\ref{eq:stratMC}) of strata is).
\begin{figure}
	\centering
	\subfigure{\includegraphics[width=0.7\textwidth]{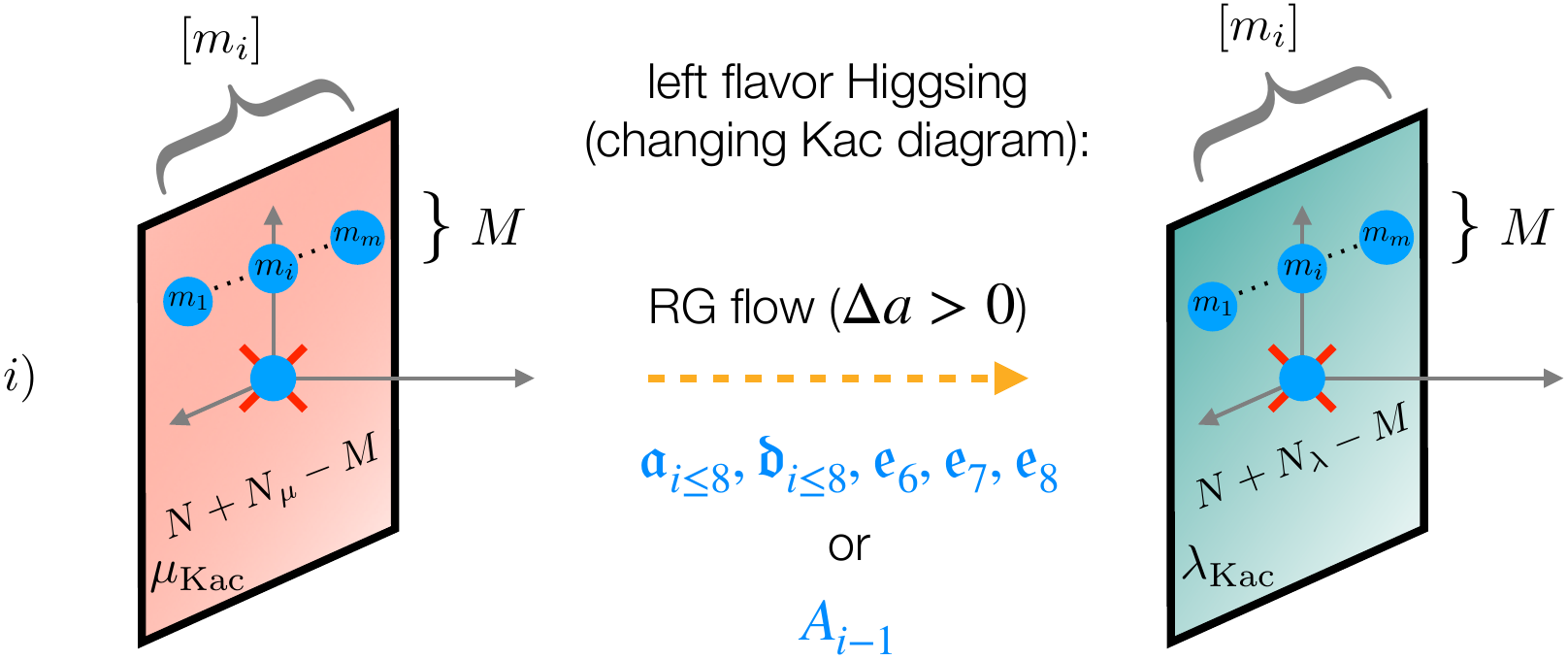}}
	\hfill
	\subfigure{\includegraphics[width=0.7\textwidth]{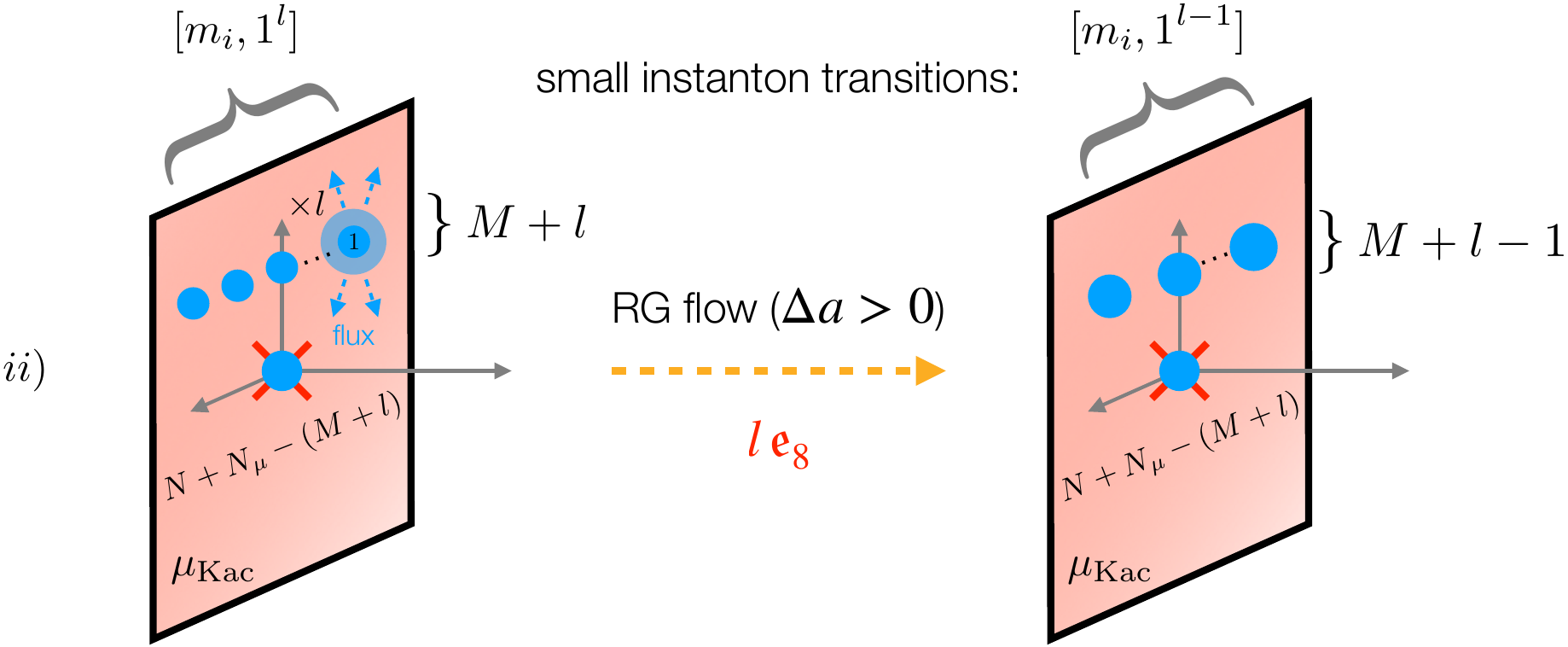}}
	\subfigure{\includegraphics[width=0.7\textwidth]{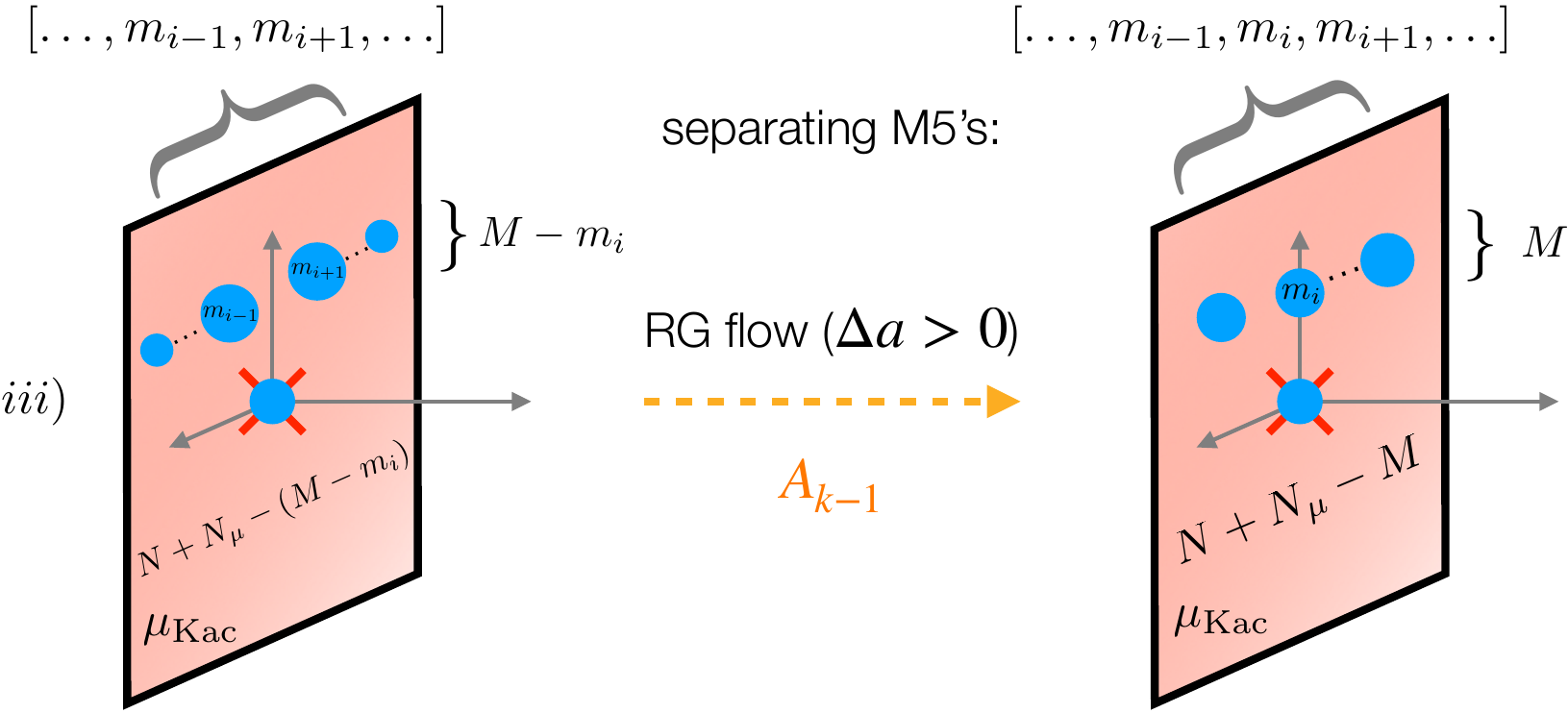}}
		\hfill
	\subfigure{\includegraphics[width=0.7\textwidth]{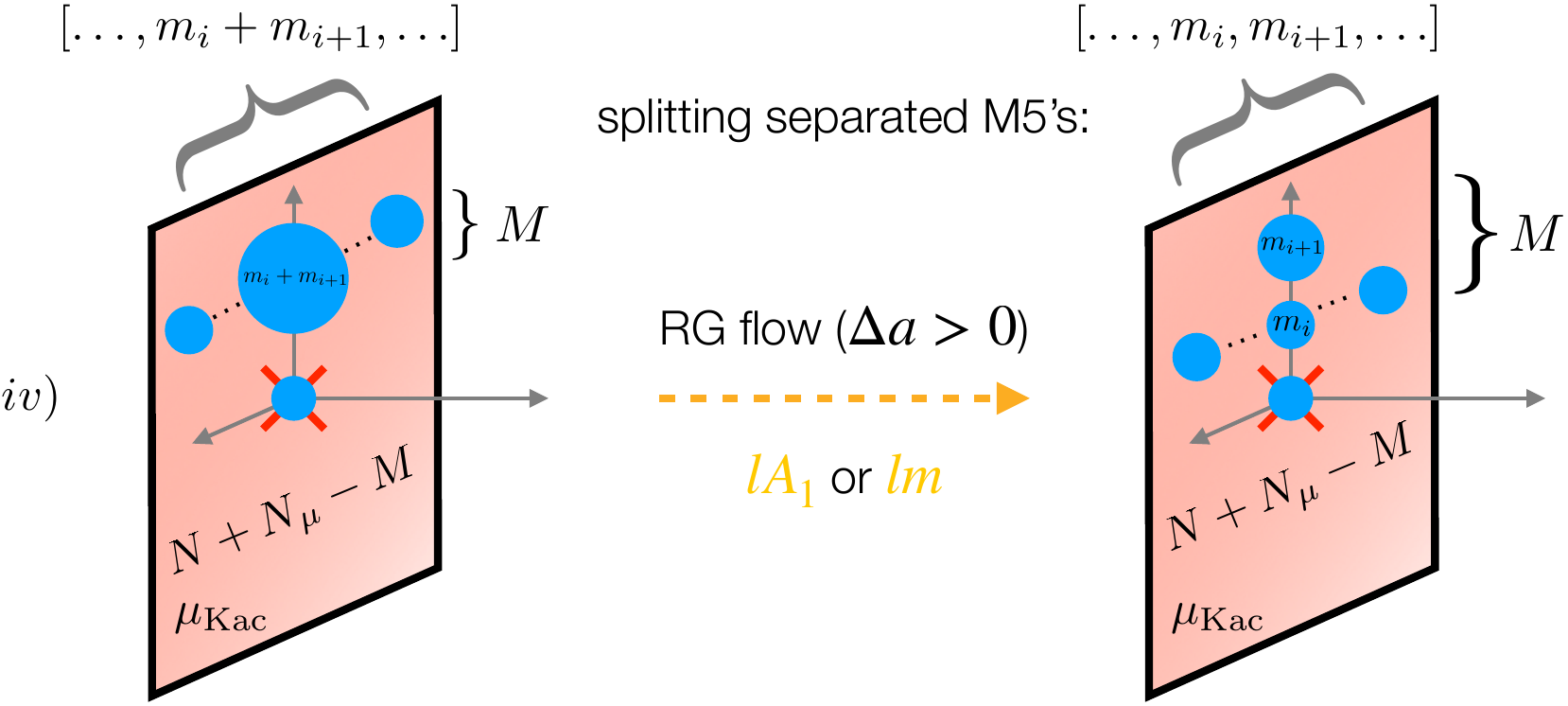}}
	\caption{An M-theory realization of the degenerations \emph{\ref{bullet:i325}},  \emph{\ref{bullet:ii325}},  \emph{\ref{bullet:iii325}}, and \emph{\ref{bullet:iv325}} as brane moves.}
	\label{fig:degs}
\end{figure}

\section{Hierarchy of RG flows as Hasse diagram of dominant coweights}
\label{sec:check}

We can now bring the abstract lessons learned in the previous section to fruition, and construct the hierarchy of left flavor Higgsings among A-type orbi-instantons, as the Hasse diagram of strata in the double affine Grassmannian of $E_8$ corresponding to dominant coweights $\lambda,\mu$ connected by minimal degenerations $\lambda <\mu$ (transverse slices). We will do so explicitly for $k=1,\ldots,7$ and generic value of $N$, even though it is clear the algorithm presented above can be applied at any arbitrary level $k$.

The power of this setup, using the double affine Grassmannian,  is to bring to light the iterative features of the Hasse of strata/hierarchy of (Higgs branches of) SCFTs and slices/RG flows,  which are present, but nowhere near as obvious, in other approaches such as the one in \cite{Bourget:2023dkj,Bourget:2024mgn}.

\subsection[$k=1$]{\texorpdfstring{$\bm{k=1}$}{k=1}}
\label{sub:k=1}

The $k=1$ case is somewhat degenerate and deserves a separate discussion.
Since there is no orbifold to begin with, the theories in this class are just E-strings, with a simple electric quiver on the tensor branch given by
\begin{equation}\label{eq:Estrele}
	[E_8]\underbrace{12\cdots 2}_N\ .
	\end{equation}
The only allowed Kac diagram is $\lambda_\text{Kac}=[1]$, so $\lambda=(\lambda_\text{Kac},n)=([1],n)$ is the only dominant coweight at level $k=1$ and fixed value of $n=-N$.
Thus $([1],n)<([1],n+1)=([1],n)+c$ is the only possible minimal degeneration between coweights (see Thm. \ref{Roytheorem}).
A decomposition along the lines of (\ref{eq:strata}) of ${\mathcal M}_\text{C}(\mu,\lambda)$ still holds, however the strata are in general unions of several of these subsets.
This is illustrated by the case $\mu=\lambda+c$: then (\ref{eq:strata}) gives three subsets
$${\mathcal M}_\text{C}(\lambda,\lambda) < {\mathcal M}_\text{C}(\lambda,\lambda)\times {\rm Sym}^1_{[1]}(\cc^2\setminus\{ 0\}) < {\mathcal M}^\text{smooth}_\text{C}(\mu,\lambda)$$
where the notation $<$ indicates the closure order.
Clearly, the union of the first two of these subsets is smooth (and equal to $\cc^2$), hence forms a stratum.
In general, given an interval $\lambda<\mu=\lambda+Nc$, the strata of ${\mathcal M}_\text{C}(\mu,\lambda)$ are in correspondence with the pairs $(\nu,[m_i])$ of a coweight $\nu=\mu-Mc$ and a partition $[m_i]$ of $M$.
(Note the distinction with the case $k>1$, where $\nu$ is also associated to strata for partitions of $M-1$, $M-2$ etc.)
Given this clarification, the analysis of section \ref{subsub:newdeg} holds, but with the caveat that there are no degenerations of type \ref{bullet:i325} or \ref{bullet:iii325}.

In the M-theory setup, there are two ways to reduce the number of coincident M5's probing the M9: tensor branch flows,  corresponding to separating ``horizontally'' (i.e. along $x^{6}$ --  see again figure \ref{fig:decoupling}) one M5 at a time;\footnote{From the geometric perspective, this was already analyzed in \cite{Heckman:2015ola}; from the 6d anomaly polynomial perspective, in \cite{Intriligator:2014eaa}.  Each tensor branch flow contracts the $-1$ curve in (\ref{eq:Estrele}) to zero size, and F-theory requires \cite{Heckman:2013pva} that the adjacent $-2$ in turn blows down to $-1$, leaving behind a rank-$(N-1)$ E-string. In M-theory it corresponds to separating one M5 from the M9 horizontally, i.e. along a direction transverse to the M9 worldvolume.  This process can be iterated all the way down to $[E_8]1$, the rank-1 E-string corresponding to $([1],1)$, or even to the empty theory (denoted $\emptyset$ at the top of figure \ref{fig:outputk=2}) corresponding to $([1],0)$ (i.e. with 0 M5's).} Higgs branch flows,  of interest to this paper, corresponding to either separating ``vertically'' (i.e. along the $x^{789}$-$x^{10}$ plane) some of the M5's (which is the minimal degeneration type \ref{bullet:iv325} from section \ref{subsub:newdeg}) or to dissolving some (or all) back into flux (which is degeneration type \ref{bullet:ii325}). For a Higgs branch example which includes the full Hasse diagram of a rank-4 E-string see the left of figure \ref{fig:outputk=2}.\footnote{The schematic form of the full Higgs branch of the rank-$N$ E-string can be seen in \cite[Fig. 3.6]{Lawrie:2023uiu} and was given already in \cite[Fig. 45]{Martone:2021ixp} for $N=2$. It can also be obtained via the new algorithm proposed in \cite{Bourget:2023dkj,Bourget:2024mgn}, which supersedes the decorated quiver subtraction algorithm of \cite{Bourget:2022ehw,Bourget:2022tmw}. This algorithm is used here to produce figure \ref{fig:outputk=2} in appendix \ref{app:fullhasse}.} 

What about the 3d perspective? If we include the center-of-mass hypermultiplet degrees of freedom,  the Higgs branch of the rank-$N$ E-string has $\dim_\mathbb{H} \text{HB}= 30N$, which for $N=1$ gives $30=29+1=\dim_\mathbb{H} \overline{\min_{E_8}} + \dim_\mathbb{H} \mathbb{H}$.  This agrees with the dimension of the Coulomb branch of the magnetic quiver for the ``$k=1$ $N=1$ orbi-instanton'' (i.e. the rank-1 E-string) in \cite[Eq. (2.20)]{Cabrera:2019izd},\footnote{This quiver is the ``over-extended'' affine $E_8$ Dynkin quiver introduced in \cite{Cremonesi:2014vla,Cremonesi:2014xha}.}
\begin{equation}
	1-1-2-3-4-5-\overset{\underset{|}{\displaystyle 3}}{6}-4-2\ ,
\end{equation}
given by the total sum of the gauge ranks minus 1. (See appendix \ref{app:magquivs} for general definitions and results on magnetic quivers, in particular the dimension formula \eqref{eq:dimmodspE8}, when $k>1$.) This is not a coincidence, as the Coulomb branch of the above quiver is by construction the same as the Higgs branch of the E-string SCFT (via compactification on $T^3$ followed by mirror symmetry).  Decoupling the center-of-mass hypermultiplet,\footnote{A free hypermultiplet, corresponding to the center of mass of the $N$ M5's decouples from the dynamics. It corresponds to a translational mode of the codimension-4 instantons on $\cc^2_{789,10}$ \cite{Intriligator:2014eaa}.} the Higgs branch of the \emph{interacting} part of the E-string theory is given by the reduced (i.e. centered) moduli space of $N$ $E_8$-instantons on $\cc^2$ \cite{Ohmori:2014pca}, i.e.  $\dim_\mathbb{H} \text{HB}_\text{E-str}= 30N-1$. For $N=1$, this has quaternionic dimension 29, as does $ \overline{\min_{E_8}}$ or the Coulomb branch of
\begin{equation}
	1-2-3-4-5-\overset{\underset{|}{\displaystyle 3}}{6}-4-2\ ,
\end{equation}
which is the magnetic quiver originally proposed in \cite[Eq. (2.27)]{Hanany:2018uhm}.\footnote{Here $k_\text{\cite{Hanany:2018uhm}}=4$.}

Finally,  for $k=1$ we can already give a direct computational proof of the $a$-theorem, given the simplicity of the transverse slices between the strata considered.  (Alternatively, the proof descends from Lemma \ref{lemma:B2} in appendix \ref{app:anomaly_decoupled}.) This is done as follows (we repeat the argument of \cite{Cordova:2015fha} for the convenience of the reader).  The eight-form 6d anomaly polynomial of the rank-$N$ E-string plus free hypermultiplet $I = I_\text{E-str} + I_\text{free hyper}$ was computed in \cite{Ohmori:2014pca} (to which we refer the reader for the notation), and reads:
\begin{align}\label{eq:Estrpoly}
	I                   & = \frac{4 N^3+6 N^2+3N}{24} c_2(R)^2  -\frac{6 N^2+5N}{48}  p_1(T)c_2(R) + \frac{7 N}{192} p_1(T)^2 -\frac{N}{48} p_2(T)\ , \\
	I_\text{E-str}      & =\frac{4 N^3+6 N^2+3N}{24} c_2(R)^2  -\frac{6 N^2+5N}{48}  p_1(T)c_2(R) \nonumber                                           \\
	                    & \quad+ \frac{7 (30N-1)}{5760} p_1(T)^2 -\frac{120 N-4}{5760} p_2(T)\ , \label{eq:Estrpolynohyper}                           \\
	I_\text{free hyper} & =  \frac{7 p_1(T)^2 -4 p_2(T)}{{5760}}\ .
\end{align}
In particular, from the relation $I \supset \dim_\mathbb{H} \text{HB} \, \frac{7 p_1(T)^2-4p_2(T)}{5760}$ one can read off the dimension of the Higgs branch of the theory, since $\frac{7 p_1(T)^2-4p_2(T)}{5760}$ is the contribution to the 6d anomaly polynomial of a single hypermultiplet. Then,  \eqref{eq:Estrpoly} gives $\dim_\mathbb{H} \text{HB} = 30N$, while \eqref{eq:Estrpolynohyper} gives $\dim_\mathbb{H} \text{HB} = 30N-1$, as expected.  The $a$ anomaly can now be computed applying the well-known relations \cite[Eqs. (1.6) \& (1.7)]{Cordova:2015fha} to \eqref{eq:Estrpoly},  yielding \cite[Eq. (5.2)]{Cordova:2015fha}, namely:
\begin{align}\label{eq:aEstrplushyper}
	a_\text{E-str $+$ free hyper}(N) & = \frac{64}{7}N^3+\frac{144}{7}N^2+\frac{99}{7}N\ ;                                        \\
	a_\text{E-str}(N)                & = \frac{64}{7}N^3+\frac{144}{7}N^2+\frac{99}{7}N - \frac{11}{7\cdot 30}\ .\label{eq:aEstr}
\end{align}
It is easy to verify the $a$-theorem for any Higgs branch flow, e.g.  the one triggered by dissolving $M$ separated M5's (out of $N$) into the M9:
\begin{align}\label{eq:DeltaaHiggs}
	\Delta a_\text{Higgs} & = a_\text{E-str $+$ free hyper}(N)-30M a_\text{free hyper}= a_\text{E-str}(N)-29M a_\text{free hyper} \\
	                      & = \frac{64}{7}N^3+\frac{144}{7}N^2+\frac{99}{7}N - 30\, \frac{11}{7\cdot 30} M\ ,
\end{align}
which is positive for
\begin{equation}
	{64}N^3 + {144} N^2 + 99 N > 11M\ ,
\end{equation}
i.e. for any $0\leq M\leq N$ when $N \geq 1$.\footnote{Analogously,  for a tensor branch flow and for any $N\geq 0$:
\begin{align}
	\Delta a_\text{tensor} & = a_\text{E-str $+$ free hyper}(N+1)-a_\text{E-str $+$ free hyper}(N) = a_\text{E-str}(N+1)-a_\text{E-str}(N) = \tfrac{192}{7} N^2 +\tfrac{480}{7}N+\tfrac{307}{7} > 0\, . \nonumber
\end{align}}
\subsection[$k=2$]{\texorpdfstring{$\bm{k=2}$}{k=2}}
\label{sub:k=2}

The $k=2$ case is the first for which actual orbi-instantons exist (i.e. we have a $\cc^2/\zz_2$ orbifold in M-theory).

In figure \ref{fig:k=2hasse} we can see the repeating pattern of the semi-infinite Hasse diagram: in red we have highlighted the hierarchy of RG flows at fixed number $N$ of full instantons,  whereas in blue the hierarchy at fixed total number of instantons (full plus fractional), which coincides with \cite[Fig.  11(a)]{Fazzi:2022hal}.
\begin{figure}[ht!]
	\hskip 0.25\textwidth
	\def\svgwidth{.65\textwidth}
	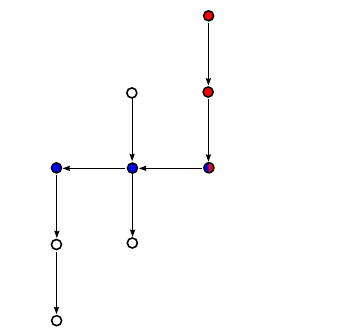
	\caption{A cutout from the semi-infinite periodic Hasse diagram of dominant coweights at level $k=2$, that is the hierarchy of left flavor Higgsings between orbi-instantons for $k=2$ and non-increasing numbers of full ($N$) and fractional ($N_\rho$)  instantons. The labels on the transitions are the minimal degeneration types from section \ref{sec:grass} (and would also correspond to the result of quiver subtraction between the 3d magnetic quivers associated with the UV and IR orbi-instantons).  E.g.  the $\mathfrak{a}_1=A_1=\cc^2/\zz_2$ flavor Higgsing is triggered by a VEV for a hypermultiplet charged under the $\su{2}$ subalgebra of $\mathfrak{f}=E_7\oplus \su{2}$ associated with Kac diagram $[2]$. The $a$ conformal anomaly decreases along all allowed oriented paths (RG flows).}
	\label{fig:k=2hasse}
\end{figure}

\subsection[$k=3,\ldots,7$]{\texorpdfstring{$\bm{k=3,\ldots,7}$}{k=3...7}}
\label{sub:k=37}

In this section we simply showcase the Hasse diagrams of dominant coweights of affine $E_8$ at level $k=3,\ldots,7$. Once again, the $a$ anomaly decreases along any allowed path, and by slicing at fixed $N+N_\rho$ we recover the hierarchies of \cite[App. A]{Fazzi:2022hal}.
\begin{figure}
	\centering
	\subfigure[$k=3$]{\includegraphics[height=0.9\textheight]{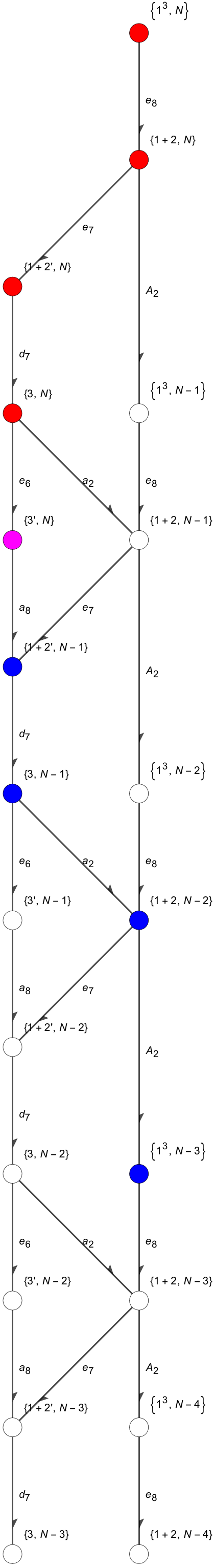}}
	\hskip 50pt
	\subfigure[$k=4$]{\includegraphics[height=0.9\textheight]{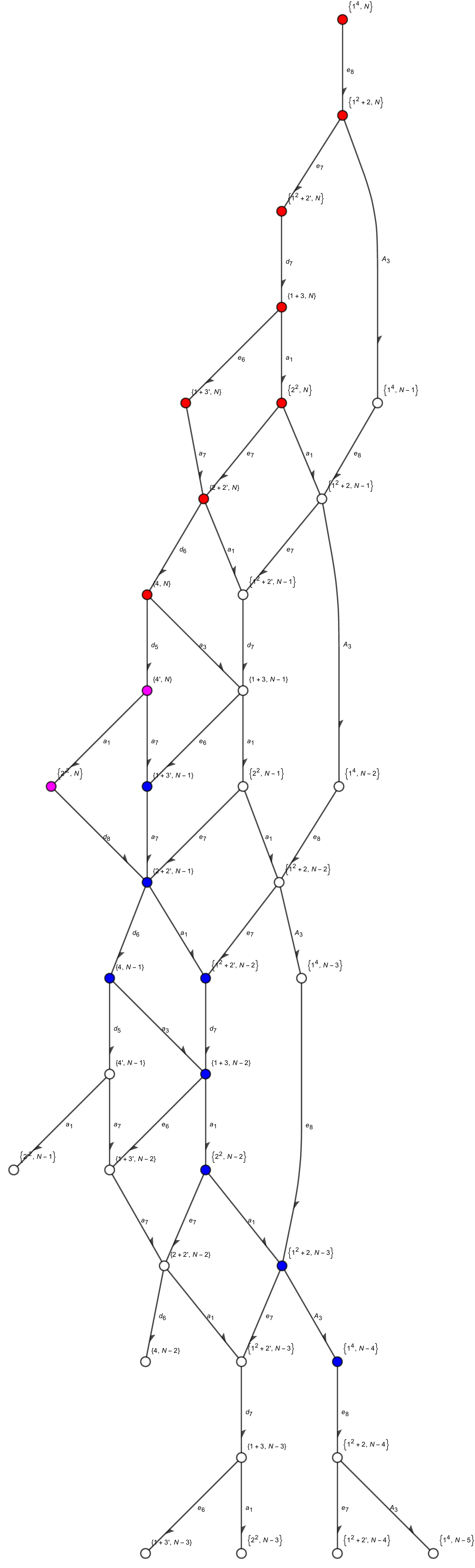}}
	\caption{A cutout from the semi-infinite periodic Hasse diagram at level $k=3,4$. The magenta node is shared between blue and red slicing (connected subdiagram). The blue slicing coincides with \cite[Fig.  11(b,c)]{Fazzi:2022hal}.\label{fig:k34notvec}}
\end{figure}
\begin{figure}
	\centering
	\subfigure[$k=5$]{\includegraphics[height=0.9\textheight]{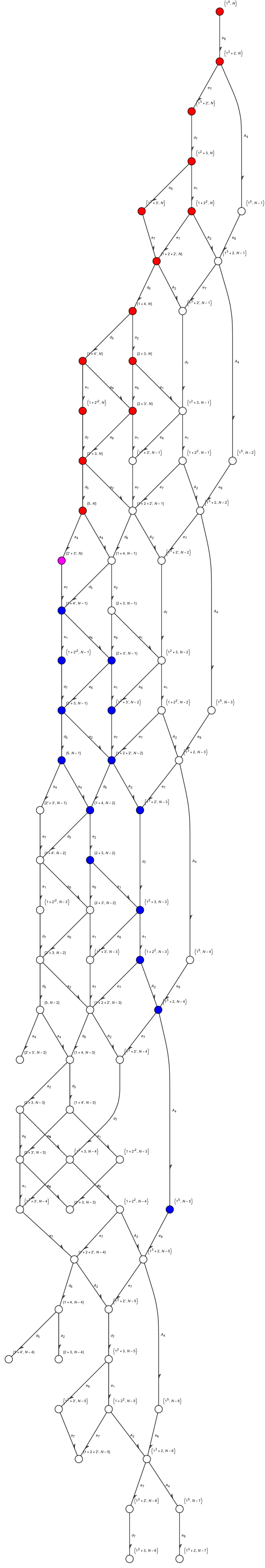}}
	\hskip 50pt
	\subfigure[$k=6$]{\includegraphics[height=0.9\textheight]{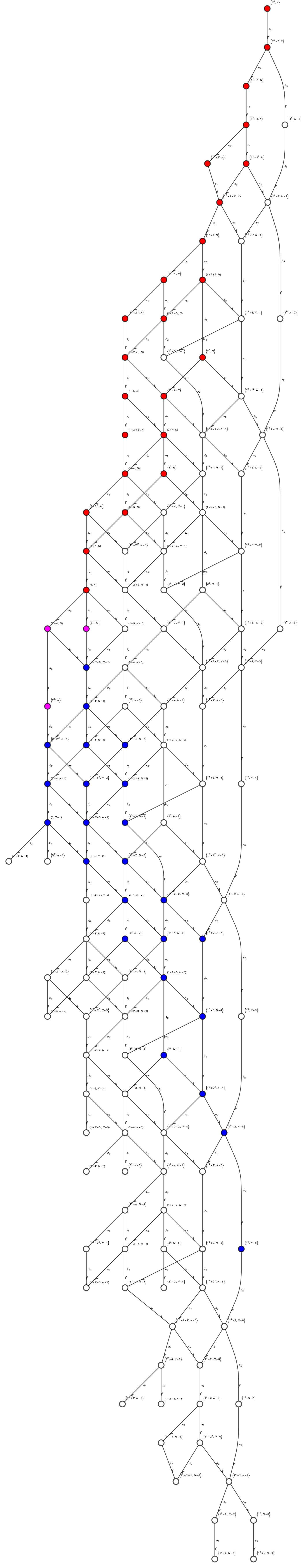}}
	\caption{A cutout from the semi-infinite periodic  Hasse diagram at level $k=5,6$. The magenta node is shared between blue and red slicing (connected subdiagram). The blue slicing coincides with \cite[Fig.  11(d), Fig. 3]{Fazzi:2022hal}.\label{fig:k56notvec}}
\end{figure}

\begin{figure}[t]
	\centering
	\subfigure[]{\includegraphics[width=0.3\textwidth]{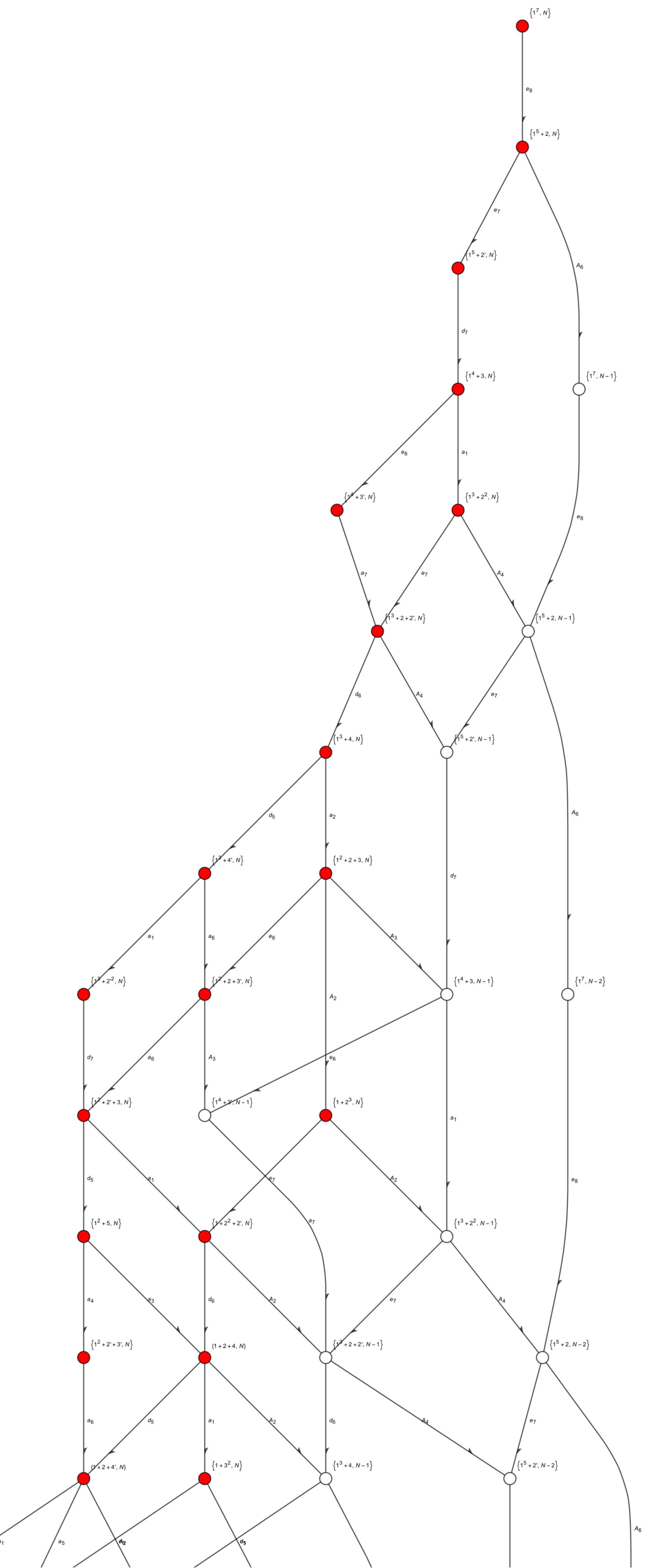}}
	\hfill
	\subfigure[]{\includegraphics[width=0.3\textwidth]{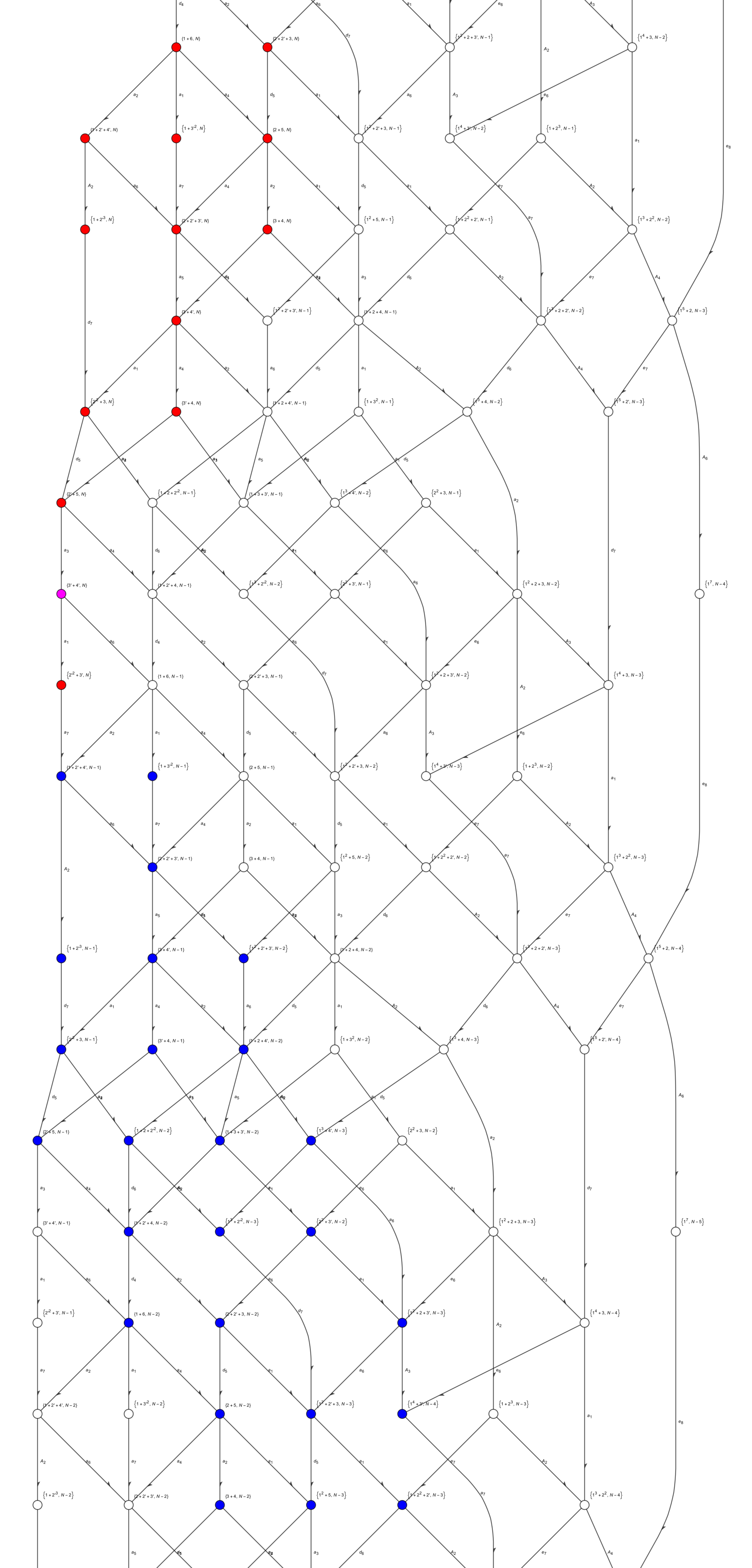}}
	\hfill
	\subfigure[]{\includegraphics[width=0.3\textwidth]{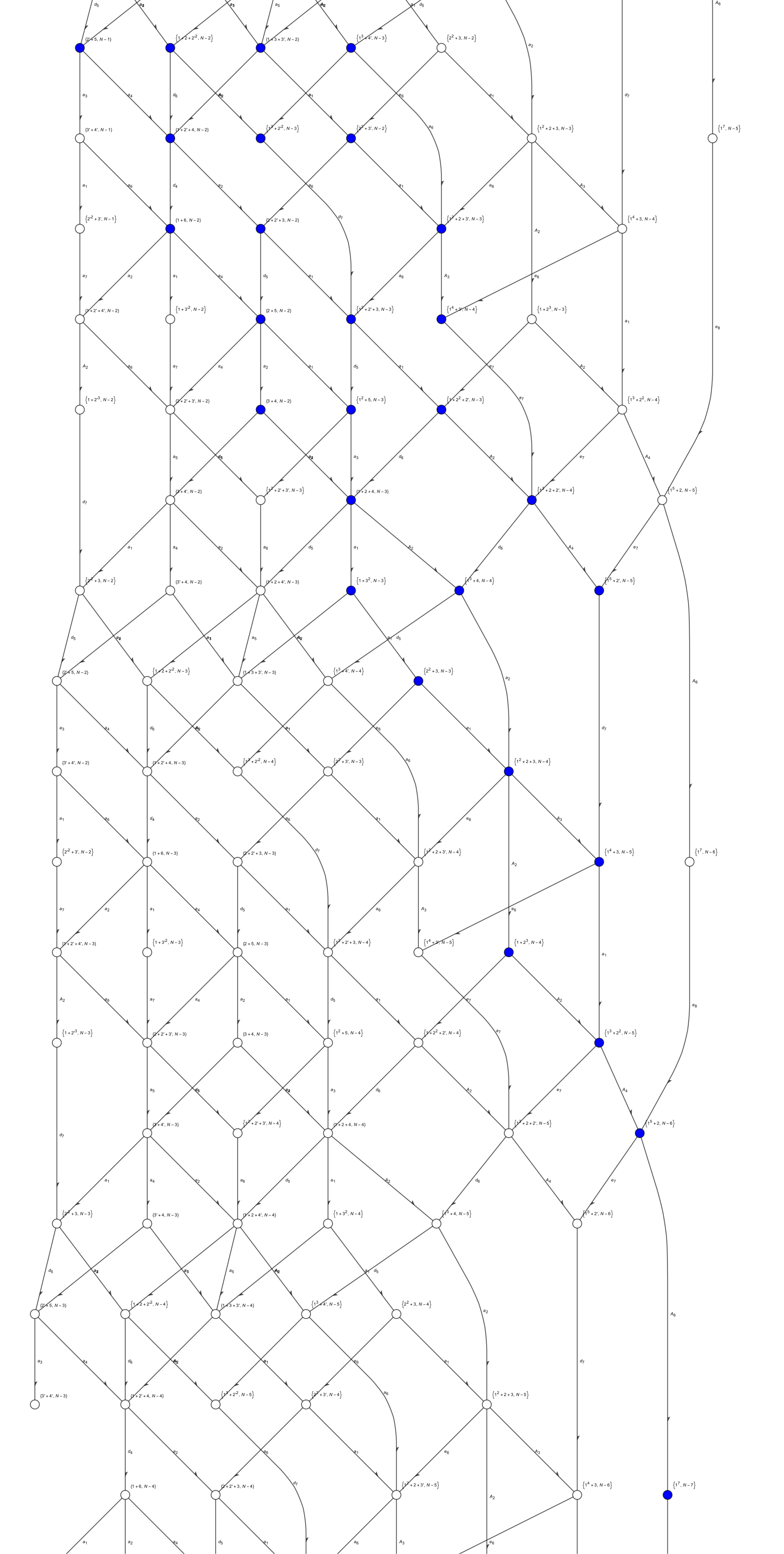}}
	\caption{A cutout from the semi-infinite periodic Hasse diagram at level $k=7$. The full diagram is too long to present here, so it has been cut into three pieces corresponding to (a), (b), (c), with each beginning where the previous ends. Some overlapping regions have been presented for clarity. The magenta node is shared between blue and red slicing (connected subdiagram). The blue slicing coincides with \cite[Fig.  13]{Fazzi:2022hal}.\label{fig:k7notvec}}
\end{figure}

\section[Proving the \textbf{\textit{a}}-theorem]{Proving the \texorpdfstring{$\bm{a}$}{\textbf{\textit{a}}}-theorem}
\label{sec:a}

We are finally in a position to prove the $a$-theorem for Higgs branch RG flows of the flavor Higgsing type between 6d A-type orbi-instantons.  (The $a$-theorem for decoupling flows descends from it, as proven in appendix \ref{app:anomaly_decoupled}.) To prove that $\Delta a>0$ for all allowed flows we use the partial order on the Hasse diagram of dominant coweights of affine $E_8$, and express $a$ in terms of combinatorial data of the latter.

\subsection[The $a$ anomaly of A-type orbi-instantons]{The \texorpdfstring{$\bm{a}$}{a} anomaly of A-type orbi-instantons}
\label{sub:afun}

We retain all the notation $\lambda=(\lambda_{\text{Kac}},n)=(k,\overline\lambda,n)$, $n_i$, $N_\rho$ (defined in Lemma \ref{sigmalem}) and the relation $N=-n$ determined in section \ref{subsub:relation}.  In particular, the condition (which holds for generic orbi-instantons $(N,k,\rho_\infty)$, i.e. for generic $N$) that the total number of full and fractional instantons is at least equal to $k$ is given in our notation by $n\leq N_\rho-k$. Recall that there is a bilinear form $\langle .\, , . \rangle$ on $Y(T)$ which is invariant with respect to the Weyl group and which satisfies $\langle \overline\lambda,\alpha^\vee\rangle=\overline\lambda(\alpha)$ for any root $\alpha$.
For distinct simple coroots $\alpha_i^\vee,\alpha_j^\vee$, we therefore have $\langle\alpha_i^\vee,\alpha_i^\vee\rangle=2$ and $\langle \alpha_i^\vee,\alpha_j^\vee\rangle$ is $-1$ if $i-j$ is an edge of the Dynkin diagram, and is zero otherwise.
Expressing $\overline\lambda$ as $\sum_{i\in\Delta} b_i \alpha_i^\vee$, we therefore have:
\begin{equation}\label{eq:lambdalambda}
	\langle\overline\lambda,\overline\lambda\rangle = 2(b_2^2-b_2b_3+b_3^2-\ldots +b_6^2-b_6b_{3'}+b_{3'}^2-b_6b_{4'}+b_{4'}^2-b_{4'}b_{2'}+b_{2'}^2)\ .
\end{equation}
The quadratic form $\langle\overline\lambda,\overline\lambda\rangle$ is the {\it length function}; it is simply the bilinear form obtained from the Cartan matrix for $E_8$.
When $\overline\lambda=\sum_{i\in\Delta} n_i \varpi_i^\vee$, then the length function is obtained from the {\it inverse} Cartan matrix.
One important property is that $\langle\overline\lambda,\overline\lambda\rangle=2$ if and only if $\overline\lambda$ is a coroot.

An equally important function is obtained from the {\it Weyl vector} $\rho=\frac{1}{2}\sum_{\alpha\in\Phi^+}\alpha=\sum_{i\in\Delta}\varpi_i\in X(T)$.
This satisfies $\alpha_i^\vee(\rho)=1$ for all simple coroots $\alpha_i^\vee$, hence we obtain the {\it height function}: $\overline\lambda(\rho)=\sum_{i\in \Delta} b_i$ where $\overline\lambda=\sum_{i\in\Delta} b_i \alpha_i^\vee$.

The $a$ anomaly of orbi-instantons can be extracted from the anomaly inflow calculation in \cite[Sec. 3.4]{Mekareeya:2017jgc}.
In our notation, i.e.  when expressed in terms of the data $\overline\lambda, k, n$, it reads:
\begin{equation}\label{eq:a6danalytic}
	\begin{split}
		a= & -\frac{3}{2} \overline\lambda(\rho) - \frac{16}{7} \langle \overline\lambda,\overline\lambda\rangle \Big[-15n+ k \left( 6n^2 - 6k n + 2 k^2 + 13 \right) \Big] +  \frac{48}{7} \langle \overline\lambda,\overline\lambda\rangle^2 \Big[ k -n\Big] \\
		   & + \frac{1}{420} \Big[ -502 + 14415 k + 251 k^2+3680 k^3 + 576 k^5                                                                                                                                                                                 \\
		   & - (18900 +12960 k^2 + 2880 k^4) n
		+ 2880k  \left( 5 + 2k^2 \right) n^2  - 3840k^2 n^3                                                                                                                                                                                                    \\
		   & + 400 \sum_{{\alpha}\in \Phi^+} \overline\lambda(\alpha)^3 - 96 \sum_{\alpha\in \Phi^+} \overline\lambda(\alpha)^5\Big].\
	\end{split}
\end{equation}
Since $N=-n$, we see that $a$ can only depend on $N$ and the choice of Kac diagram $\lambda_\text{Kac}=(k,\overline\lambda)$, and its large-$N$ leading term ($\tfrac{64}{7}N³k²$) scales like $N^3$ and is universal, as it should.\footnote{See \cite{Cremonesi:2015bld,Apruzzi:2017nck} for analogous calculations in T-brane theories. This leading term was also obtained via holography in \cite{Bah:2017wxp,Apruzzi:2017nck}.} For example, it is straightforward to verify that for $k=1$ (which implies $\overline\lambda= \overline\lambda(\rho) =\langle \overline\lambda,\overline\lambda\rangle =0$) the expression in \eqref{eq:a6danalytic} correctly reduces to \eqref{eq:aEstrplushyper}.

\subsection{The proof}
\label{sub:proof}

Because of Roy's Theorem \ref{Roytheorem}, the proof of the $a$-theorem for $k>1$ is reduced to the statement: $\Delta a = a(\lambda)-a(\mu)> 0$ for any pair $\lambda<\mu$ of dominant coweights of either of the following two forms:
\begin{itemize}
	\item[{}]
	      Type \emph{i}): $\mu=\lambda+\hat\alpha_I^\vee$ where $I$ is an irreducible component of the subdiagram of zeros of $\lambda$.
	\item[{}]
	      Type \emph{ii}): $\mu=\lambda+\alpha_i^\vee$ where $n_j>0$ for all $j$ which are connected to $i$ by an edge.
\end{itemize}
The proof comes down to a case-by-case check for each of the 53 proper connected subdiagrams of the affine Dynkin diagram in Type \emph{i}), and each of the nine simple coroots in Type \emph{ii}).  Although one could in principle deal with these cases by hand, many of the calculations become rather complicated; we therefore checked the details with a computer.
(We used the computational algebra platform GAP.)
We will however sketch an argument that covers most of the cases in Type \emph{i}).
Note that this type is often more straightforward to compute, because various coefficients $n_i$ must be zero.

Let $\lambda<\mu$ be a pair as in Type \emph{i}).
For simplicity, assume until further notice that $I$ does not contain the affine node $\alpha_1$.
Then $\lambda=(k,\overline\lambda,n)$ and $\mu=(k,\overline\lambda+\hat\alpha_I^\vee,n)$ (i.e. the value of $n$ does not change -- there are no instanton transitions in this RG flow).
Hence the second and third lines in the formula for the $a$ anomaly (\ref{eq:a6danalytic}) are unchanged in passing from $\lambda$ to $\mu$.
It follows from our earlier discussion that $\overline\mu(\rho) =\overline\lambda(\rho)+h_I-1$, where $h_I$ is the Coxeter number of the root subsystem spanned by $I$.
Moreover, by assumption we have $\langle \overline\lambda,\hat\alpha_I^\vee\rangle=0$ and hence
\begin{equation}
	\langle\overline\mu,\overline\mu\rangle=\langle\overline\lambda+\hat\alpha_I^\vee,\overline\lambda+\hat\alpha_I^\vee\rangle=\langle\overline\lambda,\overline\lambda\rangle+\langle\hat\alpha_I^\vee,\hat\alpha_I^\vee\rangle=\langle\overline\lambda,\overline\lambda\rangle+2\ .
\end{equation}
These observations lead to a simple formula for the change in the first line of (\ref{eq:a6danalytic}).
For the final line, we need to analyze the change in values of $\sum_{\alpha\in\Phi^+}\overline\lambda(\alpha)^3$ and $\sum_{\alpha\in\Phi^+}\overline\lambda(\alpha)^5$.

\subsubsection{Considerations on sums of powers of roots}

Fix an integer $m\geq 1$.
Let
\begin{equation}
	F_m(\overline\lambda)=\sum_{\alpha>0}\overline\lambda(\alpha)^m\ ,
\end{equation}
which we think of as a homogeneous form (i.e. a multivariate polynomial in the coefficients $n_i$) of degree $m$.
Introduce bihomogeneous forms of degree $(m-i,i)$, defined as follows:
\begin{equation}
	F_m^{(i)}(\overline\lambda,\theta)=\frac{m!}{i!(m-i)!}\sum_{\alpha\in\Phi^+}\overline\lambda(\alpha)^{m-i}\theta(\alpha)^i\ .
\end{equation}
Then $F_m(\overline\lambda+\theta)= \sum_{i=0}^m F_m^{(i)}(\overline\lambda,\theta)$ for all $\overline\lambda$, $\theta$.
In particular we have $F_m^{(0)}(\overline\lambda,\theta)=F_m(\overline\lambda)$ and $F_m^{(m)}(\overline\lambda,\theta)=F_m(\theta)$.
We will carry out a general analysis of the linear and quadratic functions $F_m^{(1)}(\overline\lambda,-)$ and $F_m^{(2)}(\overline\lambda,-)$.
We assume $\overline\lambda=\sum_i n_i\varpi_i^\vee$ and $\theta=\sum_i c_i \alpha_i^\vee$, where $n_i\geq 0$ and $c_i\geq 0$ for all $i$.
(In this section, all such sums are over $i\neq 1$, i.e. $i\in\Delta$.)
Note in particular that the assumption on $\theta$ holds for any positive coroot.
We start with the following observations.

\begin{lemma}
	If $\langle\overline\lambda,\theta\rangle=0$ then $F_m^{(1)}(\overline\lambda,\theta)=0$ for $m>1$.
\end{lemma}

\begin{proof}
	By definition, $F_m^{(1)}(\overline\lambda,\theta)=m\sum_{\alpha>0}\overline\lambda(\alpha)^{m-1}\theta(\alpha)$.
	Since $\theta$ is a non-negative linear combination of coroots $\alpha_i^\vee$ with $\overline\lambda(\alpha_i^\vee)=0$, it will suffice to prove the Lemma for $\theta=\alpha_i^\vee$.
	The positive roots $\beta\in\Phi^+\setminus \{\alpha_i\}$ are either orthogonal to $\alpha_i$ (in which case $\overline\lambda(\beta)\theta(\beta)=0$) or belong to a chain $\beta,\beta+\alpha_i$ of positive roots, with $\theta(\beta)=\alpha_i^\vee(\beta)=-1$.
	Then $\overline\lambda(\beta)^{m-1}\theta(\beta)+\overline\lambda(\beta+\alpha_i)^{m-1}\theta(\beta+\alpha)=-\overline\lambda(\beta)^{m-1}+\overline\lambda(\beta)^{m-1}=0$.
	Summing up over all such chains, we obtain $F_m^{(1)}(\overline\lambda,\theta)=0$.
\end{proof}

To understand the quadratic form $F_m^{(2)}(\overline\lambda,-)$, it will be useful to express $\overline\lambda$ simultaneously in terms of fundamental coweights and simple coroots:
\begin{equation}
	\overline\lambda=\sum_i n_i \varpi_i^\vee=\sum_i b_i \alpha_i^\vee\ .
\end{equation}
Writing $\theta$ similarly as $\sum_i c_i \alpha_i^\vee$, we have:
\begin{equation}
	F_m^{(1)}(\overline\lambda,\theta)=\sum_i \frac{\partial F_m}{\partial b_i} c_i\ .
\end{equation}
We therefore obtain:

\begin{corollary}
	For any $m\geq 2$, $\frac{\partial F_m}{\partial b_i}$ is divisible by $n_i$ (as a polynomial in $n_2,\ldots ,n_{2'}$).
\end{corollary}

\begin{proof}
We consider $F_m^{(1)}(\overline\lambda,\alpha_i^\vee)=\frac{\partial F_m}{\partial b_i}$ as a polynomial in $n_2,\ldots ,n_{2'}$.
By the Lemma, it is zero when the value of $n_i$ is zero; this just means that it is divisible by $n_i$ as a polynomial.
\end{proof}

We note the special case $\frac{\partial F_2}{\partial b_i}=60n_i$.
(In fact, we have $F_2(\overline\lambda)=30\langle\overline\lambda,\overline\lambda\rangle$, which can be deduced from inspection of (\ref{eq:lambdalambda}).)
For $m>2$, we write $\frac{\partial F_m}{\partial b_i}=p_{i} n_i$, where $p_{i}=p_i(\overline\lambda)$ has degree $m-2$.
Similarly to the above, we have:
\begin{equation}
	F_m^{(2)}(\overline\lambda,\theta)=\frac{1}{2}\left(\sum_i \frac{\partial^2 F_m}{\partial b_i^2} c_i^2 + \sum_{i\neq j} \frac{\partial^2 F_m}{\partial b_i\partial b_j} c_i c_j\right)\ ,
\end{equation}
where the second sum is over the unordered pairs $i\neq j$.
We are especially interested in the case $\langle \overline\lambda,\theta\rangle=0$, which is equivalent to $c_i n_i=0$ for all $i$.
Hence we only need to know the coefficients $\frac{\partial^2 F_m}{\partial b_i^2}$ modulo multiples of $n_i$ and $\frac{\partial^2 F_m}{\partial b_i\partial b_j}$ modulo linear combinations of $n_i$ and $n_j$ (written: $\bmod (n_i,n_j)$).

Recall that the {\it support} of $\theta=\sum_i c_i \alpha_i^\vee$ is the set $\{ i\in\Delta  : c_i\neq 0\}$.
The support of a coroot is a connected subset of the Dynkin diagram.

\begin{lemma}\label{quadlem}
	a) If $i-j$ is an edge then $p_i(\overline\lambda)\equiv p_j(\overline\lambda) \bmod (n_i,n_j)$.

	b) Assuming $\langle\overline\lambda,\theta\rangle=0$, we have $F_m^{(2)}(\overline\lambda,\theta)=\sum_i p_i(\overline\lambda) c_i^2-\sum_{i-j} p_i(\overline\lambda) c_i c_j$, where the second sum is over all edges $i-j$, taken in any order.
	(This is well-defined, by (a).)

	c) If in addition the support $J$ of $\theta$ is connected then $F_m^{(2)}(\overline\lambda,\theta)=\frac{1}{2} p_j(\overline\lambda) \langle\theta,\theta\rangle$ for any choice of $j\in J$.
	In particular, if $\theta$ is a coroot then $F_m^{(2)}(\overline\lambda,\theta)= p_j(\overline\lambda)$.
\end{lemma}

\begin{proof}
	Clearly, $\frac{\partial}{\partial b_i}(n_i)=2$ and, for $j\neq i$,
	\begin{equation}
		\frac{\partial}{\partial b_j}(n_i)=\begin{cases} -1 & \mbox{if $i-j$ an edge,} \\ 0 & \mbox{otherwise.}\end{cases}.
	\end{equation}
	From the above, we have $\frac{\partial^2 F_m}{\partial b_i^2}=2 p_i + n_i \frac{\partial p_i}{\partial b_i}$ and
	\begin{equation}
		\frac{\partial^2 F_m}{\partial b_i\partial b_j} = \begin{cases} n_i \frac{\partial p_i}{\partial b_j}-p_i & \mbox{if $i-j$ an edge,} \\
              n_i \frac{\partial p_i}{\partial b_j}     & \mbox{otherwise.}\end{cases}.
	\end{equation}
	We see immediately that if $i$ is not connected to $j$ by an edge then $\frac{\partial^2 F_m}{\partial b_i\partial b_j}$ is divisible by $n_i$, so the $c_i c_j$ term can be omitted from the sum when $\langle\overline\lambda,\theta\rangle=0$.
	Furthermore, if $i-j$ is an edge then we have $\frac{\partial^2 F_m}{\partial b_i\partial b_j}\equiv -p_i\equiv -p_j\bmod (n_i,n_j)$, giving (a) and (b).
	For (c), we observe that if $i,j\in J$ and $i-j$ is an edge then (since $\langle \overline\lambda,\theta\rangle=0$) we have $n_i=n_j=0$ and $\frac{1}{2}\frac{\partial^2 F_n}{\partial b_i^2}=p_i=p_j=\frac{1}{2}\frac{\partial^2 F_m}{\partial b_j^2}=-\frac{\partial^2 F_m}{\partial b_i\partial b_j}$.
	By assumption, the indices $i,j,k,$ etc. in $J$ can be connected by edges $i-j$, $j-k$, etc. and we therefore obtain $F_m^{(2)}(\overline\lambda,\theta)=p_i(\overline\lambda)(c_i^2-c_ic_j+c_j^2-c_jc_k+c_k^2-\ldots) = p_i(\overline\lambda)\langle \theta,\theta\rangle$.
	The final assertion follows from the fact that $\langle\alpha^\vee,\alpha^\vee\rangle=2$ for any coroot $\alpha^\vee$.
\end{proof}

The above observations apply to $F_m$ for any $m>1$.
Note that if $m=2$ then equality of the second order mixed derivatives and the connectedness of the Dynkin diagram imply that $p_i=p_j$ for all $i,j$.
(It is straightforward to calculate that $p_i=60$.)
The observations in the proof of Lemma \ref{quadlem} are also quite useful when $m=3$.

\begin{lemma}\label{n=3lem}
	Assume $m=3$.
	If $i$ is not connected to $j$ via an edge then $\frac{\partial p_i}{\partial b_j}=\frac{\partial p_j}{\partial b_i}=0$.
\end{lemma}

\begin{proof}
	By the remark in the proof of the above Lemma, we have $n_i \frac{\partial p_i}{\partial b_j}=n_j\frac{\partial p_j}{\partial b_i}$ as polynomials.
	Since $p_i$ is linear, it follows that $\frac{\partial p_i}{\partial b_j}$ and $\frac{\partial p_j}{\partial b_i}$ are scalars, hence they must both be zero.
\end{proof}

The previous two lemmas suffice to determine (up to a common scalar multiple) the polynomials $p_i$ modulo $n_i$.
Although computer verification of the following Lemma is very straightforward, we have retained the conceptual proof because the same argument can be applied to an arbitrary simply-laced Lie algebra.

\begin{lemma}\label{F3analysis}
	In the case $m=3$, we have $p_i\equiv 12\, b_6$ for $i=6,3',4',5$ and:
	\begin{align}
		 & p_{2'}\equiv 12 (b_6+n_{4'}) \bmod n_{2'},\;\; p_4\equiv 12(b_6+n_5)\bmod n_4\ ,      \\
		 & p_3\equiv 12 (b_6+n_5+2n_4)\bmod n_3,\;\; p_2\equiv 12(b_6+n_5+2n_4+3n_3)\bmod n_2\ .
	\end{align}
\end{lemma}

\begin{proof}
	By Lemma \ref{n=3lem}, $p_i$ depends only on $b_i$ and the $b_j$ such that $i-j$ is an edge.
	It follows in particular that $p_{3'}$ is a linear combination of $b_{3'}$ and $b_6$, and $p_6$ is a linear combination of $b_6,b_{3'},b_{4'},b_5$.
	Since $p_6\equiv p_{3'}\bmod (n_{3'},n_6)$, then $p_6$ must be of the form $\xi b_5+\xi b_{4'}+\eta b_6+\nu b_{3'}$.
	By exactly the same argument using $p_{4'}$, we must have $\nu=\xi$.
	Hence $p_6=\xi(b_5+b_{4'}+b_{3'})+\eta b_6=-\xi n_6+(\eta+2\xi)b_6$.
	Writing $p_{3'}=ya_{3'}+zb_6$ similarly, we obtain (by an equality above) $zn_{3'}-yb_{3'}-zb_6 = \xi n_6-\xi (b_5+b_{4'}+b_{3'})-\eta b_6$, so $\xi=0$.
	Thus $p_6=\eta b_6$.
	Subject to calculation of $\eta$, our previous observations now allow us to determine all the coefficients $p_i$ modulo $n_i$.
	In particular, $p_{3'}$ is congruent to $\eta b_6$ modulo $(n_{3'},n_6)$, but is also a linear combination of $b_{3'},b_6$, hence we must have $p_{3'}\equiv \eta b_6\bmod n_{3'}$.
	Similarly, $p_{4'}\equiv \eta b_6\bmod (n_{4'})$ and $p_5\equiv \eta b_5\bmod (n_5)$.
	Now, $p_{2'}$ is congruent to $\eta b_6$ modulo $(n_{2'},n_{4'})$; the only linear combination of $b_{2'},b_{4'}$ which has this property is $\eta(b_6+n_{4'})=\eta(2b_{4'}-b_{2'})$.
	Similarly, we must have $p_4\equiv \eta(b_6+n_5)\equiv \eta(2b_5-b_4)\bmod n_4$.
	With the same argument, we obtain $p_3\equiv \eta(b_6+n_5+2n_4)\equiv \eta(3b_4-2b_3)\bmod n_3$ and $p_2\equiv \eta(b_6+n_5+2n_4+3n_3)\equiv \eta(4b_3-3b_2)\bmod n_2$.

	It only remains to determine $\eta$.
	It follows from our earlier observations that $\frac{\partial^3 F_3}{\partial b_6^3}=4\frac{\partial p_6}{\partial b_6}=4\eta$.
	Hence the coefficient of $c_6^3$ in $F_3(\theta)$ is $\frac{2}{3}\eta$.
	But $F_3(\alpha_6^\vee)=8$, hence $\eta=12$.
\end{proof}

\begin{corollary}\label{quadcor}
	If $\overline\lambda(\beta)=0$, then $F_3^{(2)}(\overline\lambda,\beta^\vee)\geq 12b_6$, with equality if the support of $\beta$ contains at least one of $6,5,4',3'$.
\end{corollary}

Note that the condition in the Corollary fails precisely when $\beta=\alpha_{2'}$ or $\beta$ is in the root subsystem generated by $\alpha_2,\alpha_3,\alpha_4$.

\begin{proof}
	This follows immediately from Lemma \ref{quadlem} and Cor. \ref{quadcor}.
\end{proof}

Cor. 5.6 completes our analysis of $F_3$.
For $F_5$, we will require a certain lower bound on $F_5^{(2)}(\overline\lambda,\beta^\vee)$.
In theory, a similar analysis to Lemma \ref{F3analysis} could be applied to determine $F_5^{(2)}(\overline\lambda,\theta)$ as a quadratic polynomial in the $c_i$.
By Lemma \ref{quadlem}, we only need to find $p_i(\overline\lambda)$ modulo $n_i$ for each $i$.
The calculations involved are rather intricate, hence we turn to computer verification for the following result.
Recall that
\begin{equation}
	k=\sum_{i=1}^6 in_i+2n_{2'}+3n_{3'}+4n_{4'}=n_1+\langle\overline\lambda,\varpi_2^\vee\rangle=n_1+b_2\ .
\end{equation}
We define various integers $\tau_{i,j}$, as follows.
For fixed $i$, the number $r(i)$ of such integers is $i-2$ if $i$ is an unprimed integer, $5$ if $i=3',4'$ and $6$ if $i=2'$.
For $2\leq i\leq 6$ we have
\begin{equation}
	\tau_{i,1}=\sum_{2\leq j< i} n_j\ ,\quad \tau_{i,2}=\sum_{3\leq j <i} n_j\ ,\quad\ldots \quad, \tau_{i,i-2}=n_{i-1}\ .
\end{equation}
Similarly,
\begin{equation}
	\tau_{3',1}=\sum_{j=2}^6 n_j\ ,\quad \tau_{3',2}=\sum_{j=3}^6 n_j\ ,\quad\ldots ,\tau_{3',5}=n_6
\end{equation}
and $\tau_{4',j}=\tau_{3',j}$ for all $j$. Finally, $\tau_{2',j}=\tau_{3',j}+n_{3'}+n_{4'}$ for $1\leq j\leq 5$ and $\tau_{2',6}=n_{4'}$.
Then we have:

\begin{lemma}\label{F5lem}
	Assume $m=5$.
	Then we have:
	\begin{equation}
		\frac{1}{20}p_i(\overline\lambda)\equiv -2(b_2-\tau_{i,1})^3+6(b_2-\tau_{i,1})\langle\overline\lambda,\overline\lambda\rangle+2\sum_{j=1}^{r(i)} \tau_{i,j}^3 \mod (n_i)\ .
	\end{equation}
\end{lemma}

\begin{proof}
	This can be verified by a computational proof, differentiating $F_5$ and dividing by $n_i$.
\end{proof}

\begin{corollary}\label{F5cor}
	For any root $\beta$ with $\overline\lambda(\beta)=0$, we have $F_5^{(2)}(\overline\lambda,\beta^\vee)+20(2b_2^3-6b_2\langle\overline\lambda,\overline\lambda\rangle)\geq 0$.
\end{corollary}

\begin{proof}
	By Lemma \ref{quadlem}(c), we only need to show that $\frac{1}{20}p_i(\overline\lambda)+2b_2^3-6b_2\langle\overline\lambda,\overline\lambda\rangle\geq 0$ for each $i$.
	By Lemma \ref{F5lem}, it is enough to prove that $G_i+2b_2^3-6b_2\langle\overline\lambda,\overline\lambda\rangle\geq 0$, where
	\begin{equation}
		G_i=-2(b_2-\tau_{i,1})^3+6(b_2-\tau_{i,1})\langle\overline\lambda,\overline\lambda\rangle+2\sum_{j=1}^r\tau_{i,j}^3\ .
	\end{equation}
	This is a straightforward computer verification.
	In fact the following stronger property holds: $G_i$ is non-decreasing as $i$ moves along the Dynkin diagram from $i=2$ to $i=2'$.
\end{proof}

Note that $p_2(\overline\lambda)+2b_2^3-6b_2\langle\overline\lambda,\overline\lambda\rangle=0$, so the inequality in the Corollary is sharp.
We need one final observation about $F_5$.

\begin{lemma}\label{F54lem}
	For any root $\beta$ with $\overline\lambda(\beta)=0$, we have $F_5^{(3)}(\overline\lambda,\beta^\vee)=0$ and $F_5^{(4)}(\overline\lambda,\beta^\vee)=\frac{5}{3}F_3^{(2)}(\overline\lambda,\beta^\vee)$.
\end{lemma}

\begin{proof}
	We observe that the positive roots $\alpha$ come in three classes:
	\begin{itemize}
		\item
		      we have $\beta^\vee(\alpha)=2$ if and only if $\alpha=\beta$, in which case $\overline\lambda(\beta)=0$;
		\item
		      the remaining roots with $\beta^\vee(\alpha)\neq 0$ come in pairs $\{ \alpha,\alpha+\beta\}$, where $\beta^\vee(\alpha)=-1$ and $\beta^\vee(\alpha+\beta)=+1$ (note that $\overline\lambda(\alpha+\beta)=\overline\lambda(\alpha)$);
		\item
		      otherwise, $\beta^\vee(\alpha)=0$.
	\end{itemize}
	Let $S_\beta$ be the set of positive roots satisfying $\beta^\vee(\alpha)=-1$.
	Then
	\begin{equation}
		\sum_{\alpha\in\Phi^+}\overline\lambda(\alpha)^{m-i}\beta^\vee(\alpha)^i = 2\,\sum_{\alpha\in S_\beta} \overline\lambda(\alpha)^{m-i}(1+(-1)^i)\ .
	\end{equation}
	In particular, we have $F_5^{(3)}(\overline\lambda,\beta^\vee)=0$ and
	\begin{equation}
		F_3^{(2)}(\overline\lambda,\beta^\vee)=6\sum_{\alpha\in S_\beta} \overline\lambda(\alpha), \quad F_5^{(4)}(\overline\lambda,\beta^\vee)=10\sum_{\alpha\in S_\beta}\overline\lambda(\alpha)\ ,
	\end{equation}
	which completes our proof.
\end{proof}

\begin{corollary}\label{finalFmcor}
	Suppose further that the support of $\beta^\vee$ contains at least one of $6,5,4',3'$.
	Then we have $F_3(\overline\lambda+\beta^\vee)=F_3(\overline\lambda)+12\, b_6+8$ and
	\begin{equation}
		F_5(\overline\lambda+\beta^\vee)\geq F_5(\overline\lambda)+20(6b_2\langle\overline\lambda,\overline\lambda\rangle-2b_2^3)+20 b_6+32\ .
	\end{equation}
\end{corollary}

\begin{proof}
	It is straightforward to check $F_3^{(3)}(\overline\lambda,\beta^\vee)=F_3(\beta^\vee)=8$, since (using the argument in the proof of Lemma \ref{F54lem}) the only class of positive roots which has a non-zero contribution to the sum is $\alpha=\beta$.
	Similarly, $F_5(\beta^\vee)=32$.
\end{proof}

We use the notation $b_i$ for brevity in the above formulas.
In the sketch of the proof of the $a$-theorem which follows, it will be important to bear in mind the equalities:
\begin{align}
	 & b_2=2n_{2'}+3n_{3'}+4n_{4'}+\sum_{i=2}^6 in_i=k-n_1\ ,                                                                      \\
	 & b_6=\langle \overline\lambda,\alpha_6^\vee\rangle = 10n_{2'}+15n_{3'}+20n_{4'}+\sum_2^6 6in_i = 5k-\sum_{i=1}^5 (6-i)n_i\ .
\end{align}

\subsubsection{Finalizing the proof}

Let $(\lambda,\mu=\lambda+\hat\alpha_I^\vee)$ where $I$ is an irreducible component of the subdiagram of zeros of $\lambda$.
Assume the support of $I$ contains at least one of $6,5,4',3'$, so that Cor. \ref{finalFmcor} applies.
We recall the requirement that $n+k-N_\rho\leq 0$; in particular, we have $n\leq 0$.
We here outline the uniform argument which establishes $\Delta a = a(\lambda)-a(\mu) >  0$ in these cases; the remaining (and indeed all) cases have been checked by computer.

Inspecting the $a$ anomaly (\ref{eq:a6danalytic}) and using Cor. \ref{finalFmcor}, we can rewrite it as
\begin{align}
	\Delta a = & \  \frac{3}{2}(h_I-1)+  \frac{32}{7}(6kn^2-6k^2 n-15n+2k^3+13k)  -\frac{192}{7}\left(\langle\overline\lambda,\overline\lambda\rangle+1\right)(k-n)\ + \nonumber \\
	           & -\frac{80}{7}b_6-\frac{40}{21}(h_I+2)  +\frac{8}{35} F_5^{(2)}(\overline\lambda,\hat\alpha_I^\vee)+\frac{32}{7}b_6+\frac{16}{35}(h_I+14)\ .
\end{align}
Subtracting the strictly positive term $\frac{11}{210}h_I+\frac{229}{210}$, it will suffice to show that the sum of the remaining terms is non-negative.
Multiplying by $\frac{7}{32}$, we obtain a quadratic polynomial in $n$ of the following form:
\begin{itemize}
	\item
	      the coefficient of $n^2$ is $6k$, hence is {positive};
	\item
	      the coefficient of $n$ is $-6(k^2-\langle\overline\lambda,\overline\lambda\rangle)+9$.
\end{itemize}
Note that the linear coefficient here is {always negative (but $n$ is non-positive)}.
In fact, we observe that $\langle\overline\lambda,\overline\lambda\rangle\leq k^2$.
Indeed, we have $k\geq \sum_{i\neq 1}in_i$ and
\begin{equation}
	\langle\overline\lambda,\overline\lambda\rangle = \sum_i \langle \varpi_i^\vee,\varpi_i^\vee\rangle n_i^2 + 2 \sum_{i\neq j} \langle \varpi_i^\vee,\varpi_j^\vee\rangle n_i n_j\ ,
\end{equation}
and one sees easily that $\langle \varpi_i^\vee,\varpi_j^\vee\rangle\leq ij$.
Adding the non-negative term $6kn^2-9(n+k-N_\rho)$, it will suffice to show that the following expression is non-negative:
\begin{equation}
	2k^3-6k\langle\overline\lambda,\overline\lambda\rangle-\frac{3}{2}b_6+16k+\frac{1}{20}F_5^{(2)}(\overline\lambda,\hat\alpha_I^\vee)-9N_\rho \ .
\end{equation}
Recalling that $k=b_2+n_1$, we note that this expression is increasing with $n_1$, hence we may assume $k=b_2$.
The sum of the cubic terms is $2b_2^3-6b_2\langle\overline\lambda,\overline\lambda\rangle+\frac{1}{20}F_5^{(2)}(\overline\lambda,\hat\alpha_I^\vee)$, which is non-negative by Cor. \ref{F5cor}.

It only remains to check that $\frac{3}{2}b_6+9N_\rho\leq 16b_2$.
But it can be easily verified that:
\begin{equation}
	32b_2-3b_6 = 46 n_2+60 n_3+74 n_4+88 n_5+102 n_6+51 n_{3'}+68 n_{4'}+34 n_{2'}\ ,
\end{equation}
from which the required bound follows immediately.

This proves the desired inequality in the cases where $\overline{\mu}=\overline\lambda+\hat\alpha_I^\vee$ and $I$ contains at least one of the nodes $3',4',6,5$.
It is worth remarking again that the sharp equality in Cor. \ref{F5cor} is exactly what was needed to deal with the cubic terms in $\Delta a$.
We used the computational algebra platform GAP to check in all cases (including these, and the Type \emph{ii}) cases $\overline\mu=\overline\lambda+\alpha_i^\vee$) that the bound $\Delta a > 0$ holds.

\section{Conclusions}
\label{sec:conc}

In this paper we have realized the hierarchy of Higgs branch RG flows between the A-type orbi-instantons as the Hasse diagram of the stratification of the double affine Grassmannian of $E_8$.  The strata that we have considered are given by the orbi-instantons (potentially together with a number of decoupled E-strings) at fixed $k$ but different holonomy at infinity $\rho_\infty$ and (potentially) number $N$ of full instantons (M5-branes); the transverse slices among them are the flavor Higgsings reducing the $E_8$ flavor symmetry of the UV SCFT on its Higgs branch. Exploiting the natural partial order defined on this Hasse diagram, we have proven analytically the $a$-theorem for all such RG flows.

Several extensions of this work can be considered.
In this work we only considered the case where $\Gamma_\text{ADE}$ is cyclic, i.e. of type A.
For $\Gamma_\text{ADE}$ of type D or E, the problem of classifying homomorphisms $\Gamma_\text{ADE} \to G$ is (in general) open.
The homomorphisms from $\Gamma_{E_8}$ to exceptional type $G$ have been classified in a series of papers by Frey \cite{Frey1998-ff,Frey1998-e6,Frey2001-dh}.
Some partial results on homomorphisms from $\Gamma_{D_5}$ and $\Gamma_{D_7}$ to $E_8$ were also obtained in \cite{Frey1998-ff}.
In subsequent work with Rudelius \cite{Frey:2018vpw}, a conjectural classification of homomorphisms $\Gamma_{\textrm D}$ to $E_8$ was given in terms of orbi-instantons (and their F-theory electric quivers),\footnote{The 3d magnetic quivers can likewise be constructed \cite{Cabrera:2019dob}.} leading to a partial order on such homomorphisms via RG flows.
However, it is unclear what is the correct mathematical language to understand the flows.

For binary dihedral groups, i.e. for $\Gamma$ of type D, a ``Kac-style'' classification of homomorphisms $\Gamma\rightarrow {\rm Aut}\, G$ is possible, using the theory of involutions of simple Lie algebras.
The idea is as follows: let $\lambda_{\rm Kac}$ be a Kac diagram with an associated homomorphism $\Gamma:{\mathbb Z}_k\rightarrow T\subset G$ into the standard maximal torus of $E_8$.
There is an element $w_0$ of the Weyl group which acts as $-1$ on the Cartan subalgebra (i.e. as $t\mapsto t^{-1}$ on $T$).
Furthermore, $w_0$ has a representative $n_0$ in the normalizer of $T$, and $n_0^2=1$.
It follows that $\Gamma$ extends to a homomorphism from the dihedral group
\begin{equation}
	\langle a,b : a^k=b^2=1\,, \  bab^{-1}=a^{-1}\rangle
\end{equation}
to $E_8$, sending $a$ to $\Gamma(1)$ and $b$ to $n_0$.
One can modify this construction by replacing $n_0$ by $n_0 g$ for any $g$ which commutes with $\Gamma(a)$, i.e. which lies in the zero subalgebra of $\lambda_{\rm Kac}$.
The challenge now is to determine which $g$ lead to homomorphisms from the binary dihedral group $\Gamma_{D_{k-2}}$ to $E_8$, and when two choices of $g$ lead to conjugate homomorphisms.
We will return to this (and its physical implications) in future work.

We remark here an obvious connection with unipotent orbits in $G$. (This connection was hinted at in various places in the physics literature, perhaps starting with \cite[Sec. 5]{Heckman:2018pqx}.) Recall that the conjugacy classes of such homomorphisms are in one-to-one correspondence with the homomorphisms $\SU(2)\to G$.
Specifically, for each unipotent element $x\in G$ there exists a homomorphism $\phi:\SU(2)\to G$ such that $\phi\left(\begin{smallmatrix} 1 & 1 \\ 0 & 1 \end{smallmatrix}\right)=u$, and any two such homomorphisms are conjugate by an element of the centralizer of $u$.
If $u\neq 1$ then $\phi$ is either injective or has kernel $\{ -I\}$.
Thus, any finite subgroup $\Gamma_\text{ADE}\subset\SU(2)$ has image $\phi(\Gamma_\text{ADE})\subset G$ which is isomorphic to $\Gamma_\text{ADE}$ or $\Gamma_\text{ADE}/\{ \pm I\}$.
For a fixed subgroup $\Gamma_\text{ADE}$, this construction (ranging over various unipotent orbits) produces various homomorphisms $\Gamma_\text{ADE}\to G$.
It is clear from Kac's classification there exist homomorphisms from finite cyclic groups $\Gamma$ to $G$ which cannot be extended to homomorphisms from $\SU(2)$ to $G$; by considering an involution which acts as $-1$ on a fixed Cartan subalgebra, this is also true for binary dihedral groups.
On the other hand, in Frey's classification it appears that all homomorphisms from the binary icosahedral group $I$ to $E_8$ arise as the restriction of a homomorphism $\SU(2)\to G$.
However, it is quite possible for two distinct homomorphisms $\phi_1,\phi_2: \SU(2)\to G$ to have the same restriction to a finite subgroup $\Gamma_\text{ADE}\subset\SU(2)$.
We mention for example that in Frey's classification \cite{Frey1998-ff,Frey2001-dh} there is exactly one homomorphism $I\to G$ whose image has trivial connected centralizer.
By a straightforward calculation involving standard results on centralizers \cite{lawther-testerman}, one observes that this holds for the restriction of the $\SU(2)$ for each of the following seven (distinguished) unipotent orbits: those with Bala--Carter labels $E_8$ (regular), $E_8(a_1)$ (subregular), $E_8(a_2)$, $E_8(a_3)$, $E_8(a_4)$, $E_8(b_5)$ and $E_8(a_7)$.
This complicates the relationship between unipotent orbits and homomorphisms $I\rightarrow G$: while one might expect such homomorphisms to be encapsulated in the set of unipotent elements (the unipotent cone), this is \emph{not} consistent with the repetition observed above.

Another direction worth exploring is the study  of the Coulomb branch of the magnetic quiver of 6d orbi-instantons as a quiver variety. The problem, as we mentioned in the introduction, is that the former is neither of finite nor affine type, so it is beyond reach using  results already available in the mathematics literature. However, given its shape (see \eqref{eq:genericmagquivinf}), a possible trick would be to ``embed'' in it a long quiver of type D to study (at least partially) the root space, which is needed to apply the results of \cite{Nakajima:2015txa}.

One final comment regards the study of the full Higgs branch of the orbi-instantons, by which we mean also including flavor Higgsings of the \emph{right} $[\SU(k)]$ factor in \eqref{eq:genelquivpreblow}, i.e. $k$-changing transitions. (Some early examples include \cite{Mekareeya:2017jgc,Fazzi:2022yca,Lawrie:2023uiu}. An explicit example of the full Higgs branch for $N=k=2$ i.e. $P=N+N_\rho=2+2$ is given in appendix \ref{app:fullhasse}.) As is by now well known, these are classified by nilpotent orbits of its $\su{k}$ algebra, so the study of the full Higgs branch would involve mixing left homomorphisms $\text{Hom}(\zz_k,E_8)$ and right orbits $\mathcal{O}_\su{k}$. However we also remark that the $a$-theorem for such a combined flow would still boil down to $a$-theorems for its constituents (at least for long-enough quivers, i.e. generic choices of $N$), which are either already proven \cite{Heckman:2016ssk, Mekareeya:2016yal} (right orbit) or have been proven here (left homomorphism).

\section*{Acknowledgments}

We would like to thank A. Bourget,  S. Giacomelli, J. Grimminger, D. Juteau,  M. Sperling, Y. Tachikawa,  A. Tomasiello, Z. Zhong and especially H. Nakajima for helpful discussions and correspondence.  MF would like to thank the University of Milano--Bicocca,  the Pollica Physics Center,  the University of Lancaster, Kavli IPMU at the University of Tokyo for the warm hospitality at various stages of this work. MF and SG acknowledge support from the Simons Center for Geometry and Physics,  Stony Brook University (2022 workshops ``5d $\mathcal{N}=1$ SCFTs and Gauge Theories on Brane Webs'' and ``Supersymmetric Black Holes, Holography and Microstate Counting'', MF; 2023 Simons Physics Summer Workshop, MF and SG) at which some of the research for this paper was performed.  MF would like to thank the organizers of the 2023 ``Symplectic Singularities and Supersymmetric QFT'' conference held at UPJV in Amiens for providing a stimulating environment. SG would like to thank the University of Milano--Bicocca for hospitality during early stages of this work. The work of MF is supported in part by the Knut and Alice Wallenberg Foundation under grant KAW 2021.0170, the Swedish Research Council grant VR 2018-04438, the Olle Engkvists Stiftelse grant No. 2180108, the European Union's Horizon 2020 research and innovation programme under the European Research Council (ERC) grant agreement No. 851931~-~MEMO. The work of SG was conducted with funding awarded by the Swedish Research Council grant VR 2022-06157.

\appendix

\section{Magnetic quivers and hyperkähler quotients}
\label{app:magquivs}

In this appendix we summarize the results of \cite{Cabrera:2019izd}, which constructed 3d magnetic quivers for 6d orbi-instantons.  The former are 3d $\mathcal{N}=4$ unitary quiver gauge theories which, in the case of orbi-instantons, flow in the IR to SCFTs.  (They are ``good'' quivers in the sense of \cite{Gaiotto:2008ak}.)

On the tensor branch of the 6d orbi-instanton (i.e. at finite coupling for all gauge algebras in \eqref{eq:genelequiv} or \eqref{eq:genelequivNmu0}), the associated magnetic quiver is star-shaped with many arms:
\begin{equation}\label{eq:genericmagquiv}
	{\displaystyle
	1 - 2 - \cdots - (k-1) - \hspace{-.4cm}\overset{\overset{\displaystyle \overbrace{\displaystyle 1 \cdots 1}^{\widetilde{P}}}{ \rotatebox[origin=c]{190}{$\setminus$}\ \rotatebox[origin=c]{-190}{$/$} }}{k} \hspace{-.4cm} - r_1 -r_2 -r_3 - r_4 -r_5-\overset{\overset{\displaystyle r_{3'}}{\vert}}{r_6}-r_{4'}-r_{2'}
	}
\end{equation}
with $\widetilde{P}\coloneqq N+\sum_{i=1}^6 n_i$, or, equivalently,
\begin{equation}\label{eq:genericmagquivmod}
	{\scriptstyle
		1 - 2 - \cdots - (k-1) - \hspace{-0.275cm}\overset{\overset{\scriptstyle \overbrace{\scriptstyle 1 \ \cdots \ 1}^{P}}{ \rotatebox[origin=c]{190}{$\setminus$}\ \rotatebox[origin=c]{-190}{$/$} }}{k} \hspace{-0.275cm} - (r_1-p) -(r_2-2p) -(r_3-3p) - (r_4-4p) -(r_5-5p)-\overset{\overset{\scriptstyle (r_{3'}-3p)}{\vert}}{(r_6-6p)}-(r_{4'}-4p)-(r_{2'}-2p)
	}
\end{equation}
with $P \coloneqq \widetilde{P}+p = N+N_\rho$, remembering the definition in \eqref{eq:Nmu}.  

The ranks $r_i,r_{i'}$  are given by:
\begin{align}
	 & r_j \coloneqq (1-\delta_{j6})\sum_{i=1}^{6-j} i n_{i+j} + 2n_{2'}+3n_{3'}+4n_{4'} =k-\sum_{i=1}^6 i n_i + (1-\delta_{j6})\sum_{i=1}^{6-j} i n_{i+j}
\end{align}
for $j=1,\ldots,6$ and
\begin{align}
	 & r_{2'}\coloneqq n_{3'}+n_{4'}\ ,  \quad r_{3'} \coloneqq n_{2'}+n_{3'}+2n_{4'}\ ,  \quad r_{4'} \coloneqq n_{2'}+2n_{3'}+2n_{4'}\ .
\end{align}
Notice that $r_{2'}$ also equals $2(p+r)$ in the notation of \cite{Cabrera:2019izd}; that is, $r$ is defined as
\begin{equation}
	r \coloneqq \frac{r_{2'}}{2} - p = \frac{r_{2'}}{2}- \min \left( \lfloor \left(n_{3'}+n_{4'}\right)/2 \rfloor, \lfloor   \left(n_{2'}+n_{3'}+2n_{4'}\right)/3 \rfloor \right)\ .
\end{equation}
The number $r_{2'}=2(p+r)$ is the total number of half-NS5's stuck on the O8$^-$ in a Type IIA engineering of the orbi-instanton. This number is important, as starting at level $k=7$ the following happens: there are some Kac diagrams for which the resulting ranks $r_i,r_{i'}$ are \emph{all} greater or equal to the Coxeter label of affine $E_8$ in position $i,i'$.  That is, $p\neq 0$. Let us call these the ``saturating diagrams''.\footnote{Since for this to be possible we must have
\begin{equation}
	r_{2'}=n_{3'}+n_{4'}\geq 2\ , \ r_{3'} = n_{2'}+n_{3'}+2n_{4'}\geq 3\ , \ r_{4'} = n_{2'}+2n_{3'}+2n_{4'} \geq 4
\end{equation}
(which automatically imply that $r_6=2n_{2'}+3n_{3'}+4n_{4'} \geq 6$),  we see that the first case is the $[3',4']$ diagram at $k=7$ (which saturates the above bounds).  A simple case-by-case analysis shows that the saturating diagrams must have $(n_{3'},n_{4'})\geq (1,1),(0,2),(2,0)$ where in the last case we should also impose that $n_{2'}\geq 1$. } For these diagrams the bouquet of $\U(1)$'s in the finite-coupling magnetic quiver \eqref{eq:genericmagquiv} has $p$ extra 1's.  In other words, the total number of full instantons (NS5's in the orientifold of Type IIA) is increased by $p$. Notice that this is already taken care of by our definition of $P\coloneqq N+N_\rho$, since $N_\rho \coloneqq \sum_{i=1}^6 n_i +p$:  for the saturating diagrams, this $p\neq 0$ precisely accounts for the extra full instantons that we get in the Type IIA picture. Then, we simply remove the $p$ ``extra'' affine $E_8$ Dynkins from the right tail, obtaining \eqref{eq:genericmagquivmod}.

The 6d SCFT at the origin of the tensor branch has a magnetic quiver obtained by performing $\widetilde{P}$ small $E_8$-instanton transitions.  Performing the transitions in \eqref{eq:genericmagquiv},  the core proposal of \cite{Cabrera:2019izd} is that we obtain the following star-shaped quiver in 3d:
\begin{equation}\label{eq:genericmagquivinf}
	{\scriptstyle
		1 - 2 - \cdots - k - (r_1+\widetilde P) -(r_2+2\widetilde P) -(r_3+3\widetilde P) - (r_4+4\widetilde P) -(r_5+5\widetilde P)-\overset{\overset{\scriptstyle r_{3'}+3\widetilde P}{\vert}}{(r_6+6\widetilde P)}-(r_{4'}+4\widetilde P)-(r_{2'}+2\widetilde P)}\ .
\end{equation}
The Coulomb branch dimension of the above quiver was computed in \cite{Mekareeya:2017jgc}, and reads:
\begin{equation}\label{eq:dimmodspE8}
	\dim_\mathbb{H} \hat{\mathcal{M}}_\text{C} =  30 (N+k) +\frac{k}{2}(k+1) -\overline{\lambda}(\rho)-1\ ,
\end{equation}
for all $k>1$,  where $\overline{\lambda}(\rho)$ is the height function associated with the diagram $\lambda_\text{Kac}=(k,\overline{\lambda})$ defined in section \ref{sub:afun}:
\begin{equation}\label{eq:height}
	\overline{\lambda}(\rho)  =  29 n_2+57 n_3+84 n_4+110 n_5+135 n_6+46 n_{2'}+68 n_{3'}+91 n_{4'}\ ,
\end{equation}
which obviously vanishes only for the diagram $\mu_\text{Kac}=[1^k]$.

\subsection[The $\SU(k)$ hyperkähler quotient of the orbi-instanton magnetic quiver]{The \texorpdfstring{$\bm{\text{SU}(k)}$}{SU(k)} hyperkähler quotient of the orbi-instanton magnetic quiver}
\label{app:moduli}

We can now prove that the hyperkähler quotient $(\hat{\mathcal{M}}_\text{C}\times \mathcal{O}_{\bm{\xi}})/\!\!/\!\!/_{\bm{\xi}}\SU(k)$, with $\hat{\mathcal{M}}_\text{C}$ the Coulomb branch of \eqref{eq:genericmagquivinf}, equals the moduli space of $E_8$-instantons on the deformation/resolution of $\cc^2/\zz_k$ with desingularization parameter $\bm{\xi}\neq 0$, i.e. the Coulomb branch $\mathcal{M}_\text{C}^{\text{inst},\bm{\xi}}$ of
\begin{equation}
	\label{eq:quotmagquivinf}
	{\scriptstyle \boxed{\scriptstyle k} - (r_1+\widetilde P) -(r_2+2\widetilde P) -(r_3+3\widetilde P) - (r_4+4\widetilde P) -(r_5+5\widetilde P)-\overset{\overset{\scriptstyle r_{3'}+3\widetilde P}{\vert}}{(r_6+6\widetilde P)}-(r_{4'}+4\widetilde P)-(r_{2'}+2\widetilde P)}\ ,
\end{equation}
with a $\PU(k):=\U(k)/\U(1)$ action from the flavor $\boxed{k}$ node and a triplet of mass parameters determined by ${\bm\xi}$. (As is well known, that the Coulomb branch of the above quiver gives the instanton moduli space was originally proposed by \cite{Intriligator:1996ex} for $k=1$ and \cite{deBoer:1996ck} for any $k$.) Because of the hyperkähler quotient, we also have
\begin{equation}\label{eq:dimmodinst}
	\dim_\mathbb{H} \mathcal{M}_\text{C}^{\text{inst},{\bm\xi}} = 30 (N+k) -\overline{\lambda}(\rho)\ 
\end{equation}
when $\xi_\mathbb{C}$ is generic but $\xi_\mathbb{R}=0$.  (An ``overall'' $\U(1)$ gauge group does not decouple in this case due to the presence of the $\boxed{k}$ flavor node on the left. This explains the absence of the $-1$ summand, which, by contrast, appears in \eqref{eq:dimmodspE8}.)

We would like to thank Hiraku Nakajima for suggesting the following argument.

\begin{proof}
We can simply apply the gluing construction from \cite{Braverman:2017ofm} to $T(\SU(k))$ and \eqref{eq:genericmagquivinf}, which are both equipped with a ${\rm PU}(k)$ flavor symmetry factor coming from $\boxed{k}$.  (The gluing construction is motivated by \cite{Cremonesi:2014kwa,Cremonesi:2014vla} and, as explained there, in 3d QFT language it simply corresponds to gauging the diagonal of two identical flavor symmetries.) More precisely, let $G=\PU(k)$ and let ${\mathcal A}_1$, resp. ${\mathcal A}_R$ be the ring-like objects in the category $D_{G[\![z]\!]}(\Gr_G)$ representing the (Coulomb branch of the) quiver gauge theory \eqref{eq:quotmagquivinf}, resp.  $T(\SU(k))$, cf. \cite[Sec. 2(v)]{Braverman:2017ofm}.
Letting ${\mathcal A}_2$ be the skyscraper sheaf at the base-point of $\Gr_G$, then \cite[Rem. 5.21]{Braverman:2017ofm} becomes:
\begin{equation}\label{eq:remark521}
        H^*_{G[\![z]\!]}(\Gr_G,\mathcal{A}_1) \cong
        H^*_{G[\![z]\!]}(\Gr_G,\mathcal{A}_R\otimes^!\mathcal{A}_1)\otimes
        \mathfrak C_{G^\vee} H^*_{G[\![z]\!]}(\Gr_G,\mathcal{A}_R) /\!\!/\!\!/ \Delta_{G^\vee} \,
\end{equation}
where $G^\vee=\SU(k)$ is the Langlands dual group.  As explained in \cite{Braverman:2017ofm}, these cohomology rings correspond to the coordinate rings of the relevant Coulomb branches. The symbol $\otimes^!$ denotes the gluing operation, while $/\!\!/\!\!/ \Delta_{G^\vee}$ represents the hyperkähler quotient by the diagonal subgroup of two copies of $G^\vee$, which corresponds to the flavor symmetry being gauged in 3d. This construction yields the quiver \eqref{eq:genericmagquivinf}.  Thus we obtain the desired statement:
\begin{equation}
\mathcal{M}_\text{C}^{\text{inst},\bm\xi} \cong	(\hat{\mathcal{M}}_\text{C} \times \mathcal{M}_\text{C}(T(\SU(k))  ) /\!\!/\!\!/_{\bm{\xi}} \SU(k) \ .
\end{equation}
\end{proof}

\noindent As an interesting observation, we note that $\mathcal{A}_R$ is also known to give rise to the Springer resolution of $\mathcal{N}_{G^\vee}$ \cite[Thm. 8.5.2]{Arkhipov2004Quantum}.  

Furthermore, the Coulomb branch $\mathcal{M}_\text{C}^{\text{inst},\bm\xi}$, with $\xi_\mathbb{C}$ generic but $\xi_\mathbb{R} = 0$, can be understood as a slice between the coweights $([1^k],N)$ and $([1^k],-k)$ in the double affine Grassmannian.   On the other hand, the space $\mathcal{M}_\text{C}^{\text{inst}}\coloneqq\mathcal{M}_\text{C}^{\text{inst},{\bf 0}}$ corresponds to the union of the strata in $\mathcal{M}_\text{C}^{\text{inst},\bm\xi}$ associated with coweights $(\lambda_{\rm Kac}, j)$ for $j>0$. It is also equivalent to a closed union of strata in $\hat{\mathcal{M}}_\text{C}$.

\section{Proving the \texorpdfstring{$\bm{a}$}{a}-theorem for decoupled systems}
\label{app:anomaly_decoupled}

In section \ref{subsub:newdeg}, we described four types of minimal degeneration, i.e. transition.
The $a$-theorem was established for the first type, i.e. for transitions changing the holonomy $\rho_\infty$.
We will now show that the $a$ central charge decreases along RG flows corresponding to the other three types of transition.
All three types involve a decoupled system given by a decoupled product of rank $m_i$ E-strings and a triple $(N,k,\rho_\infty)$, i.e. a coweight for affine $E_8$.
Since the $a$ anomaly for such a decoupled system is the sum of the $a$ anomalies of the individual factors, we can ignore any E-strings found on both sides of the transition.
This reduces us to proving that $\Delta a>0$ for the following RG flows (the numbering corresponds to the classification in \ref{subsub:newdeg}):
\begin{itemize}
\item[\emph{i)}] from a rank 1 E-string to the trivial system (dissolving an M5 into flux);
\item[\emph{ii)}] from the system given by the triple $(N+M,k,\rho_\infty)$ to a decoupled system given by $(N,k,\rho_\infty)$ and a rank $M$ E-string (moving a substack of $M$ M5's away from the singularity);
\item[\emph{iii)}] from a rank $M=p+q$ E-string to a decoupled product of a rank $p$ and a rank $q$ E-string (splitting a stack of M5's into two substacks).
\end{itemize}

Case \emph{ii)} involves the trivial verification that the $a$ anomaly for a rank 1 E-string is strictly positive.
The other two transitions are established in the following two lemmas.

\begin{lemma}\label{lemma:B1}
	The $a$ anomaly decreases along the RG flow from the orbi-instanton $(N+M,k,\rho_\infty)$  to $(N,k,\rho_\infty)$ and a single rank-$M$ E-string:
	\begin{equation}
		Δa^{(iii)} \coloneqq a\left(N+M,k,\rho_\infty\right) -a\left(N,k,\rho_\infty\right) -a\left(M,1,\left[1^1\right]\right) > 0\ .
	\end{equation}
\end{lemma}
\begin{proof}
	\begin{equation}
		\begin{aligned}
			Δa^{(iii)} & \coloneqq a\left(N+M,k,\rho_\infty\right) -a\left(N,k,\rho_\infty\right) -a\left(M,1,\left[1^1\right]\right)                                                                                                              \\
			         & = \frac{48M}{7} \left[ \frac{k²}{6} \left(2 M^2-3\right)+\left(\frac{5}{2}+k (M+2 N)+\underbrace{k²-\langle\overline\lambda,\overline\lambda \rangle}_{\geq 0}\right)^2-\left(\frac{7}{4}+3M+\frac{4M²}{3}\right) \right] \\
			         & \geq  \frac{48M}{7} \left[ \frac{k²}{6} \left(2 M^2-3\right)+\left(\frac{5}{2}+kM\right)^2-\left(\frac{7}{4}+3M+\frac{4M²}{3}\right) \right]                                                                              \\
			         & \geq 0 \quad \forall\, k \geq 1 \textrm{ and } M \geq 1\ .
		\end{aligned}
	\end{equation}
\end{proof}

\begin{lemma}\label{lemma:B2}
	The $a$ anomaly decreases when a single rank-$M$ E-string breaks into two E-strings of rank $p$ and $q$ respectively such that $p+q=M$.
	\begin{equation}
		Δa^{(iv)} \coloneqq a\left(M,1,\left[1^1\right]\right) - a\left(p,1,\left[1^1\right]\right) - a\left(q,1,\left[1^1\right]\right) > 0\ .
	\end{equation}
\end{lemma}

\begin{proof}
It is straightforward to check that:
	\begin{equation}
		Δa^{(iv)} = \frac{96}{7} p q \left( 3 + 2p + 2q \right) > 0\ .
	\end{equation}
\end{proof}

\section{Hasse diagram of dominant coweights and symmetric products}
\label{app:fullhasse}

In this section we showcase for illustrative purposes the full Hasse diagram of the Higgs branch of the UV orbi-instanton $(N=2,k=2,\rho_\infty^\text{UV}:2=[1^2])$, i.e. $P=N+N_\rho=2+2$ --  see figure \ref{fig:outputk=2}. (In other words, we have 2 M5's, or full instantons,  probing the M9 with trivial boundary condition, meaning the wall will fractionate into 2 fractions, or fractional instantons, because of the $\cc^2/\zz_2$ orbifold probed by the M5's.) 

As remarked in sections \ref{sec:intro} and \ref{sec:conc},  the stratification of this Higgs branch includes as possible RG flows (i.e. slices): left flavor Higgsings,  splittings of separated M5's on the M9, and right flavor Higgsings.  The first two types of slices have been interpreted here as, respectively,  those between dominant coweight strata of the double affine Grassmannian of $E_8$ (see figure \ref{fig:k=2hasse} for $k=2$), i.e. the minimal degenerations \ref{bullet:i325} introduced in section \ref{subsub:newdeg}, and those between strata of the symmetric product of $N$ objects (see figure \ref{fig:degs}),  i.e.  the minimal degenerations \ref{bullet:ii325}-\ref{bullet:iv325}, potentially with multiplicity $l$.  (Note that by combining the proof in section \ref{sec:a} and that in appendix \ref{app:anomaly_decoupled} one can prove the $a$-theorem along any slice/RG flow that includes left flavor Higgsings and splittings of M5's.)

This stratification should coincide with that of the Nakajima quiver variety $\widehat{\mathcal{M}}_\text{C}$, i.e. the Coulomb branch of \eqref{eq:genericmagquivinf}.
However, the latter quiver is of indefinite type (i.e. it is neither finite nor affine), for which almost nothing has been mathematically proven about the set of strata.
We have argued in section \ref{sec:intro} that, staying in the vicinity of $\xi_{\mathbb{R}}=0$, one can safely disregard the contribution from the $T(\SU(k))$ tail on the left of \eqref{eq:genericmagquivinf}, and study the resulting $\mathcal{M}_\text{C}$.
While the stratification in \eqref{eq:stratMC} has not been mathematically proven either, this affine $E_8$ case is modeled on the known affine A-type case \cite{Nakajima:2016guo} and hence may be within reach of current mathematical technology (see \cite{Nakajima:1994} for preliminary results). This stratification can be seen in figure \ref{fig:subposet3d}, \ref{fig:subposet6d}, and \ref{fig:subposetMC}, being expressed in terms of resp. 3d magnetic quivers, 6d F-theory electric quivers on the tensor branch of the SCFT, and symplectic leaves in the Coulomb branch $\mathcal{M}_\text{C}$ (understood as a Nakajima quiver variety) of the UV theory.

The authors of \cite{Bourget:2023dkj,Bourget:2024mgn} have recently taken a more pragmatic approach, and have proposed a ``decay and fission'' algorithm which (for an arbitrary quiver gauge theory, in particular a 6d one on the tensor branch) outputs the stratification of the full 6d Higgs branch, taking as input the 3d magnetic quiver of the UV theory.\footnote{We thank the authors of \cite{Bourget:2023dkj,Bourget:2024mgn} for clarifying that in affine type A the decay and fission algorithm correctly reproduces the stratification result of \cite{Nakajima:2016guo}.}
Applying their algorithm to the $k=2,\widetilde{P}=P=4$ case of \eqref{eq:genericmagquivinf} (i.e.  the magnetic quiver of the $(N=2,k=2,\rho_\infty^\text{UV}:2=[1^2])$ UV orbi-instanton) produces figure \ref{fig:outputk=2} below.  

A few remarks will help clarify the figure:
\begin{itemize}
\item all oriented paths represent allowed RG flows: the $a$ anomaly decreases from \emph{bottom} to \emph{top} of the figure (contrarily to what happens in all previous figures), and thus the larger symplectic leaves (strata) lie at the top (following a convention mostly used in the mathematics literature and in some physics literature).  Accordingly,  the IR lies at the \emph{top} in all figures below;

\item $\otimes$ indicates a decoupled product of theories, e.g. an E-string and an orbi-instanton (i.e.  some of the M5's are at the singularity while others are away from it while still being on the M9). Accordingly, the combined $a$ anomaly is given by the sum of those of the constituents.   $\emptyset$ denotes the empty theory with zero interacting degrees of freedom in 6d, as the Higgs branch flow from the E-string leaves behind 29 free hypermultiplets;\footnote{Applying (\ref{eq:aEstr}) and (\ref{eq:DeltaaHiggs}) to the present case, we obtain $a_\text{E-str}(1)=\frac{307}{7} \xrightarrow{\text{RG}} a_{\emptyset + 29\,\text{free hypers}}=0+ \frac{319}{210}$, so that $\Delta a = \frac{296}{7} >0$.}

\item when present, the superscripts on the quivers identify the outputs of the algorithm of \cite{Bourget:2023dkj,Bourget:2024mgn} following the same numbering used there;

\item the \azure{$10-5-3$} and \azure{$14-11-7$ diagonals} (and the minimal degenerations/RG flows among theories on a diagonal, and across the two diagonals) embed into figure \ref{fig:k=2hasse} in the obvious way,  the latter being the Hasse diagram for the left flavor Higgsings among $k=2$ orbi-instantons (excluding splittings of M5's or right flavor Higgsings). 

The diagonals can be extended  to the right as $k$ increases (i.e. as we include more Kac diagrams for larger values of $k$), and their number can be increased indefinitely,  linearly with $N$ (i.e. we simply add more diagonals below the one at the bottom): this is the semi-infinite periodic Hasse diagram of the double affine Grassmannian of $E_8$, which was already shown in figures \ref{fig:k34notvec},  \ref{fig:k56notvec}, and \ref{fig:k7notvec}  for a few values of $k$. For higher $k$ and/or $N$ the output of the algorithm \cite{Bourget:2023dkj,Bourget:2024mgn} becomes increasingly involved; however,  the semi-infinite periodic structure of left flavor Higgsings is fully captured by the double affine Grassmannian (cf. the observation at the beginning of section \ref{sec:check});

\item the \magenta{magenta nodes} are all the (nontrivial) symmetric product strata (cf. \eqref{eq:stratMC}).
Note that in this case we only see copies of ${\rm Sym}^1({\mathbb C}^2/{\mathbb Z}_k \!\setminus\!\{0\})$, because $\lambda+2c>\mu$, or,  equivalently, only one M5 can be removed from the singularity;

\item the \red{$\mathfrak{e}_8$ degenerations} are of type \ref{bullet:ii325} in the language of section \ref{subsub:newdeg}, i.e. we are performing a small $E_8$-instanton transition dissolving one M5 into flux on the M9. Then \red{$l\mathfrak{e}_8$} means the union of $l$ copies (branches) of the $\mathfrak{e}_8$ singularity (as explained there),  i.e.  1 out of $l$ M5's transitions into flux (and we have $l$ ways of selecting which one dissolves);

\item the \nnyellow{$A_1$} and \nnyellow{$m$ degenerations} are of type \ref{bullet:iv325}, i.e. we are splitting a stack of separated M5's into substacks (and \nnyellow{$lA_1$} denotes the union as above); 

\item the \orange{$A_1$ degenerations} are of type \ref{bullet:iii325}, i.e. we are stripping a number of M5's off of the singularity while still keeping them on the M9;

\item the \azure{$\mathfrak{e}_8,\mathfrak{e}_7,\mathfrak{d}_8,\mathfrak{a}_1=A_1$ degenerations} are of type \ref{bullet:i325}, i.e. are left flavor Higgsings;

\item the \ao{$\mathfrak{a}_1, \mathfrak{b}_3,\mathfrak{g}_2$ degenerations} are \emph{right} flavor Higgsings (and as such are excluded from the analysis of section \ref{subsub:newdeg}). \ao{$\mathfrak{a}_1$} is easily understood as the Higgsing of $[\SU(2)]$, which is visible on the tensor branch of the orbi-instanton the RG flow originates from; on the other hand,  the $[\SO(7)]$ and $[G_2]$ flavors which are Higgsed via \ao{$\mathfrak{b}_3,\mathfrak{g}_2$} respectively are \emph{not} visible on the tensor branch, but are known to arise at the CFT point (this is an accident due to the nongeneric value of $N$ -- see e.g.  \cite[Sec. 5.1]{Mekareeya:2017jgc}).  These degenerations (slices) connect strata whose associated magnetic quivers have different ``left tail structures'', going from $1-2-$ to $0-2-$ (and, more generally, from $1-\cdots-k-$ to $0-k-$).  In other words,  these ``$k$-changing transitions'' (as we called them in footnote \ref{foot:kchange}) connect strata coming from the root space of the $T(\SU(2))$ tail of $\hat{\mathcal{M}}_\text{C}$ understood as a quiver variety. It is possible to interpret these ``new'' left tail structures as $T_\rho(\SU(2))$ theories, in the language of \cite{Gaiotto:2008ak}, where $T(\SU(2))=T_{\rho=[1^2]}(\SU(2))=1-\boxed{2}$ corresponds to the trivial partition of $k=2=[1^2]=[1=1-0,1=2-1]$ and $T_{\rho=[2]}(\SU(2))=0-\boxed{2}$ corresponds to the other partition, $[2]=[2-0]$ ($\boxed{2}$ being gauged in the full 3d magnetic quiver of indefinite type). This observation generalizes in the obvious way for higher $k$;

\item accordingly, the \ao{green nodes} are examples of massive E-string theories, i.e. 6d SCFTs which do \emph{not} admit an engineering in M-theory but do so in massive Type IIA string theory (as they can be defined via a brane construction with more than 8 D8's),  or F-theory (as the orbi-instantons do) -- see e.g. \cite{Fazzi:2022yca} for RG flows among massive E-strings defined by less than 8 D8's;\footnote{Note that for fewer than 8 D8's the magnetic quivers of the massive E-strings always start with the usual $1-\cdots-k-$ tail \cite[App. A]{Fazzi:2022yca}), and because of this fact the former can be identified with orbi-instantons at a lower value of $k$, cf.  \cite[Eqs. (2.14) \& (4.1)]{Fazzi:2022yca}.}

\item the \(\mathfrak{e}_8\), \(\mathfrak{e}_7\), \(\mathfrak{d}_8\), and \(\mathfrak{b}_9\) degenerations represent left flavor Higgsings that originate from either massive E-string theories or decoupled products of orbi-instanton theories and rank-1 E-strings. These degenerations land on another E-string,  massive or massless.
\end{itemize}
A mathematical proof of the stratification produced by the decay and fission algorithm \cite{Bourget:2023dkj,Bourget:2024mgn} would certainly be of interest (if only for certain subclasses of theory such as those given by star-shaped 3d quivers).
In our case, proving the affine $E_8$ case as a 3d quiver (i.e., establishing the stratification \eqref{eq:stratMC}) would be sufficient to complete our proof of the $a$-theorem for fixed $k$.
This covers the vast majority of transitions produced by the decay and fission algorithm.
We intend to address the remaining transitions in upcoming work.

\begin{landscape}
\begin{figure}
\centering
\begin{tikzpicture}[node distance=45pt, every node/.style={scale=0.9}]
			\tikzstyle{arrow} = [thick,->,>=stealth]
		
			\node (n3) {$\left(  \includegraphics{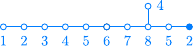}  \right)^{\azure{3}}$};
			\node (n16) [below=40pt of n3] {\magenta{$\left( \includegraphics{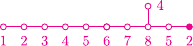} \otimes \includegraphics{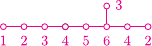}\right)^{16}$}};
			\node (n5) [left=20pt of n16] {$\left(  \includegraphics{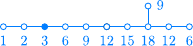} \right)^{\azure{5}}$};
	
			\node (n7) [below=40pt of n16] {$\left(  \includegraphics{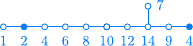} \right)^{\azure{7}}$};	
			\node (n18) [left=20pt of n7] {\magenta{$\left( \includegraphics{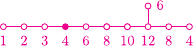}\otimes  \includegraphics{16b-magenta.pdf}	\right)^{18}$}};
			\node (n10) [left=20pt of n18] {$\left(  \includegraphics{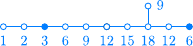} \right)^{\azure{10}}$};	

			\node (nphant8) [below=60pt of n7] {};	
			\node (nphant9) [below=40pt of nphant8] {};	
			\node (nphant10) [below=20pt of nphant9] {};				
			
			\node (n11) [below=40pt of n18] {$\left( \includegraphics{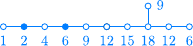}\right)^{\azure{11}}$};
			\node (n20) [below=40pt of n10] {\magenta{$\left( \includegraphics{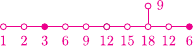}\otimes  \includegraphics{16b-magenta.pdf}\right)^{20}$}};	
		
			\node (n14) [below=40pt of n20] {$\left( \includegraphics{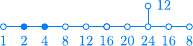}\right)^{\azure{14}}$};

			\draw [arrow] (n14)--(n11) node [midway, above left] {} node[midway, below right] {\azure{$\mathfrak{e}_8$}};
			\draw [arrow] (n14)--(n20) node [midway, above left] {} node[midway, left] {\orange{$A_1$}};	
			\draw [arrow] (n11)--(n7) node [midway, above left] {} node[midway, below right] {\azure{$\mathfrak{e}_7$}};	
			\draw [arrow] (n11)--(n18) node [midway, above left] {} node[midway, left] {\orange{$A_1$}};	
			\draw [arrow] (n11)--(n10) node [midway, above left] {} node[midway, above left=5pt and 0pt] {\azure{$\mathfrak{a}_1$}};
			\draw [arrow] (n20)--(n18) node [midway, above left] {} node[midway, above right=10pt and 10pt] {\azure{$\mathfrak{e}_8$}};		
			\draw [arrow] (n20)--(n10) node [midway, above left] {} node[midway, left] {\red{$\mathfrak{e}_8$}};
\draw [arrow] (n7)--(n16) node [midway, above left] {} node[midway, right] {\orange{$A_1$}};	
\draw [arrow] (n7)--(n5) node [midway, above left] {} node[midway, above left=10pt and 5pt] {\azure{$\mathfrak{d}_8$}};
\draw [arrow] (n18)--(n16) node [midway, above left] {} node[midway, above right=10pt and 10pt] {\azure{$\mathfrak{e}_7$}};
\draw [arrow] (n18)--(n5) node [midway, above left] {} node[midway, left] {\red{$\mathfrak{e}_8$}};	

\draw [arrow] (n10)--(n5) node [midway, above left] {} node[midway, above left] {\azure{$\mathfrak{e}_8$}};	
\draw [arrow] (n16)--(n3) node [midway, above left] {} node[midway, right] {\red{$\mathfrak{e}_8$}};	

\draw [arrow] (n5)--(n3) node [midway, above left] {} node[midway, above left] {\azure{$\mathfrak{e}_7$}};											
		\end{tikzpicture}
	\caption{3d magnetic quivers in the notation of \cite{Bourget:2023dkj,Bourget:2024mgn}. The superscript labels the quivers following the same numbering as in figure \ref{fig:outputk=2}.}
	\label{fig:subposet3d}
\end{figure}
\end{landscape}

\begin{landscape}
\begin{figure}[t]
\centering
\begin{tikzpicture}[node distance=45pt, every node/.style={scale=.8}]

			\tikzstyle{arrow} = [thick,->,>=stealth]
		
			\node (n3) {$\left(\magenta{ [\SO(16)] \,\overset{\mathfrak{usp}(2)}{1} \,\overset{\mathfrak{su}(2)}{2} [\SU(2)]}\right)^{\azure{3}}_{\azure{[2'],N=1,N_\rho=0}}$};
		
			\node (n16) [below=50pt of n3] {\magenta{$\left( \begin{array}{c} [\SO(16)]\, \overset{\mathfrak{usp}(2)}{1} \,  \overset{\mathfrak{su}(2)}{2} [\SU(2)] \\ \otimes \\ \left[E_8\right] \, 1	 \end{array}  \right)^{16}_{\scriptsize \begin{array}{c} [2'],N=1,N_\rho=0 \\  \otimes \\  \text{E-string}	\end{array}}$}};

			\node (n5) [left = 20pt of n16] {$\left( \magenta{ [E_7] \,  1 \, \overset{\mathfrak{su}(2)}{\underset{[N_\text{f}=2]}{2}} [\SU(2)]} \right)^{\azure{5}}_{\azure{[2],N=1,N_\rho=1}}$};
	
			\node (n7) [below=40pt of n16] {$\left(\magenta{ [\SO(16)] \,\overset{\mathfrak{usp}(2)}{1} \,\overset{\mathfrak{su}(2)}{2}  \overset{\mathfrak{su}(2)}{2}  [\SU(2)]} \right)^{\azure{7}}_{\azure{[2'],N=2,N_\rho=0}}$};	
			\node (n18) [left=10pt of n7] {\magenta{$\left( \begin{array}{c} \left[E_7\right] \,  1 \, \overset{\mathfrak{su}(2)}{\underset{[N_\text{f}=2]}{2}} [\SU(2)] \\ \otimes \\ \left[E_8\right] \, 1 \end{array}\right)^{18}_{\scriptsize \begin{array}{c} [2],N=1,N_\rho=1 \\  \otimes \\  \text{E-string}	\end{array}}$}};
			\node (n10) [left=10pt of n18] {$\left( \magenta{[E_8] \, 1 \, \overset{\mathfrak{su}(1)}{2} \overset{\mathfrak{su}(2)}{\underset{[N_\text{f}=1]}{2}} [\SU(2)]}   \right)^{\azure{10}}_{\azure{[1^2],N=1,N_\rho=2}}$};	

			\node (nphant8) [below=60pt of n7] {};	
			\node (nphant9) [below=40pt of nphant8] {};	
			\node (nphant10) [below=20pt of nphant9] {};				
			
			\node (n11) [below=40pt of n18] {$\left(\magenta{ [E_7] \,  1 \, \overset{\mathfrak{su}(2)}{\underset{[N_\text{f}=2]}{2}} \overset{\mathfrak{su}(2)}{2}  [\SU(2)] }\right)^{\azure{11}}_{\azure{[2],N=2,N_\rho=1}}$};
			\node (n20) [below=40pt of n10] {\magenta{$\left(\begin{array}{c} [E_8] \, 1 \, \overset{\mathfrak{su}(1)}{2} \overset{\mathfrak{su}(2)}{\underset{[N_\text{f}=1]}{2}} [\SU(2)] \\ \otimes \\ \left[E_8\right] \, 1 \end{array}\right)^{20}_{\scriptsize \begin{array}{c} [1^2],N=1,N_\rho=2 \\  \otimes \\  \text{E-string}	\end{array}}$}};	
		
			\node (n14) [below=40pt of n20] {$\left(\magenta{[E_8] \,  1 \, \overset{\mathfrak{su}(1)}{2} \overset{\mathfrak{su}(2)}{\underset{[N_\text{f}=1]}{2}} \overset{\mathfrak{su}(2)}{2}  [\SU(2)]}\right)^{\azure{14}}_{\azure{[1^2],N=2,N_\rho=2}}$};

			\draw [arrow] (n14)--(n11) node [midway, above left] {} node[midway, below right] {\azure{$\mathfrak{e}_8$}};
			\draw [arrow] (n14)--(n20) node [midway, above left] {} node[midway, left] {\orange{$A_1$}};	
			\draw [arrow] (n11)--(n7) node [midway, above left] {} node[midway, below right] {\azure{$\mathfrak{e}_7$}};	
			\draw [arrow] (n11)--(n18) node [midway, above left] {} node[midway, left] {\orange{$A_1$}};	
			\draw [arrow] (n11)--(n10) node [midway, above left] {} node[midway, above left=5pt and 0pt] {\azure{$\mathfrak{a}_1$}};
			\draw [arrow] (n20)--(n18) node [midway, above left] {} node[midway, above right=10pt and 5pt] {\azure{$\mathfrak{e}_8$}};		
			\draw [arrow] (n20)--(n10) node [midway, above left] {} node[midway, left] {\red{$\mathfrak{e}_8$}};
\draw [arrow] (n7)--(n16) node [midway, above left] {} node[midway, right] {\orange{$A_1$}};	
\draw [arrow] (n7)--(n5) node [midway, above left] {} node[midway, above left=10pt and 5pt] {\azure{$\mathfrak{d}_8$}};
\draw [arrow] (n18)--(n16) node [midway, above left] {} node[midway, above right=10pt and 5pt] {\azure{$\mathfrak{e}_7$}};
\draw [arrow] (n18)--(n5) node [midway, above left] {} node[midway, left] {\red{$\mathfrak{e}_8$}};	

\draw [arrow] (n10)--(n5) node [midway, above left] {} node[midway, above left] {\azure{$\mathfrak{e}_8$}};	
\draw [arrow] (n16)--(n3) node [midway, above left] {} node[midway, right] {\red{$\mathfrak{e}_8$}};	

\draw [arrow] (n5)--(n3) node [midway, above left] {} node[midway, above left] {\azure{$\mathfrak{e}_7$}};									
		\end{tikzpicture}
		\caption{6d F-theory electric quivers on the tensor branch with 3d magnetic quivers given in figure \ref{fig:subposet3d}.  The superscript labels the quivers following the same numbering as in figure \ref{fig:outputk=2}; the subscript indicates the type of orbi-instanton,  E-string, or decoupled product thereof.}
		\label{fig:subposet6d}
\end{figure}
\end{landscape}

\begin{landscape}
	\begin{figure}[t]
		\centering
		\subfigure[In the notation of \cite{Nakajima:2015txa,Nakajima:2015gxa}, where $S_{[1]}$ denotes the symmetric product of 1 object.] {
			\resizebox{0.45\columnwidth}{!}{
				\begin{tikzpicture}[node distance=45pt, every node/.style={scale=1}]
					\tikzstyle{arrow} = [thick,->,>=stealth]
					\node (n1) {\azure{$\mathcal{M}_\text{C}^\text{s}(\varpi_{2'}+\delta,2\varpi_1)$}};
					\node (n3) [below = 60pt of n1] {\magenta{$\mathcal{M}_\text{C}^\text{s}(\varpi_{2'},2\varpi_1)\times S_{(1)}$}};
					\node (n4R) [below = 60pt of n3] {\azure{$\mathcal{M}_\text{C}^\text{s}(\varpi_{2'},2\varpi_1)$}};
					
					\node (n2) [left=70pt of n3]  {\azure{$\mathcal{M}_\text{C}^\text{s}(\varpi_{2}+\delta,2\varpi_1)$}};
					\node (n4) [below = 60pt of n2] {\magenta{$\mathcal{M}_\text{C}^\text{s}(\varpi_{2},2\varpi_1)\times S_{(1)}$}};
					\node (n6) [below = 60pt of n4] {\azure{$\mathcal{M}_\text{C}^\text{s}(\varpi_{2},2\varpi_1)$}};
					
					\node (n4L) [left=70pt of n4] {\azure{$\mathcal{M}_\text{C}^\text{s}(2\varpi_{1}+\delta,2\varpi_1)$}};
					\node (n5) [below = 60pt of n4L] {\magenta{$\mathcal{M}_\text{C}^\text{s}(2\varpi_{1},2\varpi_1)\times S_{(1)}$}};
					\node (n7) [below = 60pt of n5] {\azure{$\mathcal{M}_\text{C}^\text{s}(2\varpi_{1},2\varpi_1)$}};
					
					\draw [arrow] (n2)--(n1) node [midway, above left] {} node[midway, above left] {\azure{$\mathfrak{e}_7$}};
					\draw [arrow] (n3)--(n1) node [midway, above right] {} node[midway, above right] {\red{$\mathfrak{e}_8$}};
					\draw [arrow] (n4)--(n2) node [midway, left] {} node[midway, below left] {\red{$\mathfrak{e}_8$}};
					\draw [arrow] (n4L)--(n2) node [midway, right] {} node[midway, above left] {\azure{$\mathfrak{e}_8$}};
					\draw [arrow] (n4R)--(n2) node [midway, right] {} node[midway, above left = 13pt and 20pt] {\azure{$\mathfrak{d}_8$}};
					\draw [arrow] (n4)--(n3) node [midway, right] {} node[midway, above right=10pt and 10pt] {\azure{$\mathfrak{e}_7$}};
					\draw [arrow] (n4R)--(n3) node [midway, right] {} node[midway, above right=-2pt] {\orange{$A_1$}};
					\draw [arrow] (n5)--(n4) node [midway, above left] {} node[midway, above right=10pt and 10pt] {\azure{$\mathfrak{e}_8$}};
					\draw [arrow] (n5)--(n4L) node [midway, above left] {} node[midway, below left] {\red{$\mathfrak{e}_8$}};
					\draw [arrow] (n6)--(n4L) node [midway, above left] {} node[midway, above left = 13pt and 20pt] {\azure{$\mathfrak{a}_1$}};
					\draw [arrow] (n6)--(n4R) node [midway, above left] {} node[midway, below right] {\azure{$\mathfrak{e}_7$}};
					\draw [arrow] (n6)--(n4) node [midway, above right=0pt and 1pt] {} node[midway, above right] {\orange{$A_1$}};
					\draw [arrow] (n7)--(n5) node [midway, left] {} node[midway, below left] {\orange{$A_1$}};
					\draw [arrow] (n7)--(n6) node [midway, right] {} node[midway, below right] {\azure{$\mathfrak{e}_8$}};
				\end{tikzpicture}		
			}
		}
		\hfill
		\subfigure[In the notation of section \ref{sec:grass} ($\varpi_i^\vee$ are the simple coroots).] {
			\resizebox{0.45\columnwidth}{!}{
				\begin{tikzpicture}[node distance=45pt, every node/.style={scale=1}]
					\tikzstyle{arrow} = [thick,->,>=stealth]
					\node (n1) {\azure{$\mathcal{M}_\text{C}^\text{smooth}(\varpi^\vee_{2'}+c,2\varpi^\vee_1)$}};
					\node (n3) [below = 60pt of n1]{\magenta{$\mathcal{M}_\text{C}^{[1]}(\varpi^\vee_{2'}+c,2\varpi^\vee_1)$}};
					\node (n4R) [below = 60pt of n3] {\azure{$\mathcal{M}_\text{C}^\text{smooth}(\varpi^\vee_{2'},2\varpi^\vee_1)$}};
					
					\node (n2) [left=70pt of n3] {\azure{$\mathcal{M}_\text{C}^\text{smooth}(\varpi^\vee_{2}+c,2\varpi^\vee_1)$}};
					\node (n4) [below = 60pt of n2] {\magenta{$\mathcal{M}_\text{C}^{[1]}(\varpi^\vee_{2}+c,2\varpi^\vee_1)$}};
					\node (n6) [below = 60pt of n4] {\azure{$\mathcal{M}_\text{C}^\text{smooth}(\varpi^\vee_{2},2\varpi^\vee_1)$}};
					
					\node (n4L) [left=70pt of n4] {\azure{$\mathcal{M}_\text{C}^\text{smooth}(2\varpi^\vee_{1}+c,2\varpi^\vee_1)$}};
					\node (n5) [below = 60pt of n4L] {\magenta{$\mathcal{M}_\text{C}^{[1]}(2\varpi^\vee_{1}+c,2\varpi^\vee_1)$}};
					\node (n7) [below = 60pt of n5] {\azure{$\mathcal{M}_\text{C}^\text{smooth}(2\varpi^\vee_{1},2\varpi^\vee_1)$}};
					
					\draw [arrow] (n2)--(n1) node [midway, above left] {} node[midway, above left] {\azure{$\mathfrak{e}_7$}};
					\draw [arrow] (n3)--(n1) node [midway, above right] {} node[midway, above right] {\red{$\mathfrak{e}_8$}};
					\draw [arrow] (n4)--(n2) node [midway, left] {} node[midway, below left] {\red{$\mathfrak{e}_8$}};
					\draw [arrow] (n4L)--(n2) node [midway, right] {} node[midway, above left] {\azure{$\mathfrak{e}_8$}};
					\draw [arrow] (n4R)--(n2) node [midway, right] {} node[midway, above left = 13pt and 20pt] {\azure{$\mathfrak{d}_8$}};
					\draw [arrow] (n4)--(n3) node [midway, right] {} node[midway, above right=10pt and 10pt] {\azure{$\mathfrak{e}_7$}};
					\draw [arrow] (n4R)--(n3) node [midway, right] {} node[midway, above right=-2pt] {\orange{$A_1$}};
					\draw [arrow] (n5)--(n4) node [midway, above left] {} node[midway, above right=10pt and 10pt] {\azure{$\mathfrak{e}_8$}};
					\draw [arrow] (n5)--(n4L) node [midway, above left] {} node[midway, below left] {\red{$\mathfrak{e}_8$}};
					\draw [arrow] (n6)--(n4L) node [midway, above left] {} node[midway, above left = 13pt and 20pt] {\azure{$\mathfrak{a}_1$}};
					\draw [arrow] (n6)--(n4R) node [midway, above left] {} node[midway, below right] {\azure{$\mathfrak{e}_7$}};
					\draw [arrow] (n6)--(n4) node [midway, above right=0pt and 1pt] {} node[midway, above right] {\orange{$A_1$}};
					\draw [arrow] (n7)--(n5) node [midway, left] {} node[midway, below left] {\orange{$A_1$}};
					\draw [arrow] (n7)--(n6) node [midway, right] {} node[midway, below right] {\azure{$\mathfrak{e}_8$}};
				\end{tikzpicture}
			}	
		}
	\caption{Hasse diagram \ref{fig:subposet3d} (3d Coulomb branch) or \ref{fig:subposet6d} (6d Higgs branch) understood as the stratification of ${\mathcal{M}}_\text{C}$, i.e.  a subdiagram of the full stratification of $\hat{\mathcal{M}}_\text{C}$ obtained by neglecting strata associated with the $T(\SU(2))$ tail.  Larger (smaller) symplectic leaves, smaller (larger) $a$ anomalies at the \emph{top} (\emph{bottom}).}
	\label{fig:subposetMC}
\end{figure}
\end{landscape}
\begin{landscape}
\begin{figure}
\centering
\begin{tikzpicture}[node distance=45pt, every node/.style={scale=0.6}]
	\tikzstyle{arrow} = [thick,->,>=stealth]
	\node (n1) {$\left(\emptyset \right)^1$};
	\node (n2) [below=15pt of n1] {$\left([E_8] \, 1\right)^2$};
	\node (nphant2) [below=15pt of n2] {};
	\node (nphant3) [below=15pt of nphant2] {};
	\node (n15) [left =110pt of nphant2] {$\left([E_8] \, 1 \otimes [E_8] \, 1\right)^{15}$};
	\node (n4) [left =250pt of nphant3] {$\left([E_8] \, 12\right)^4$};
	\node (n22) [left =408pt of nphant3] {$\left([E_8] \, 1 \otimes [E_8] \, 1 \otimes [E_8] \, 1\right)^{22}$};
	\node (n3) [below=20pt of nphant3] {$\left(\magenta{ [\SO(16)] \,\overset{\mathfrak{usp}(2)}{1} \,\overset{\mathfrak{su}(2)}{2} [\SU(2)]}\right)^{\azure{3}}_{\azure{[2'],N=1,N_\rho=0}}$};
	\node (n17) [below=30pt of n22] {$\left([E_8] \, 1 \otimes [E_8] \, 12\right)^{17}$};
	\node (n24) [left=425pt of n3] {$\left(\begin{array}{c} \left[E_8\right] \, 1 \otimes \left[E_8\right] \, 1  \ \otimes \\ \left[E_8\right] \, 1 \otimes \left[E_8\right] \, 1\end{array} \right)^{24}$};
	\node (n16) [below=40pt of n3] {\magenta{$\left( \begin{array}{c} [\SO(16)]\, \overset{\mathfrak{usp}(2)}{1} \,  \overset{\mathfrak{su}(2)}{2} [\SU(2)] \\ \otimes \\ \left[E_8\right] \, 1	 \end{array}  \right)^{16}_{\scriptsize \begin{array}{c} [2'],N=1,N_\rho=0 \\  \otimes \\  \text{E-string}	\end{array}}$}};
	\node (n5) [above left= -22pt and 10pt of n16] {$\left( \magenta{ [E_7] \,  1 \, \overset{\mathfrak{su}(2)}{\underset{[N_\text{f}=2]}{2}} [\SU(2)]} \right)^{\azure{5}}_{\azure{[2],N=1,N_\rho=1}}$};
	\node (n6) [below=60pt of n4] {\ao{$\left( [\SO(18)] \overset{\mathfrak{usp}(2)}{1} \overset{\mathfrak{su}(1)}{2}			\right)^{6}$}};	
	\node (n8) [below=60pt of n17] {$\left([E_8] \, 122\right)^{8}$};								
	\node (n23) [below=65pt of n24] {$\left([E_8] \, 1 \otimes [E_8] \, 1 \otimes [E_8] \, 12\right)^{23}$};	
	\node (n7) [below=40pt of n16] {$\left(\magenta{ [\SO(16)] \,\overset{\mathfrak{usp}(2)}{1} \,\overset{\mathfrak{su}(2)}{2}  \overset{\mathfrak{su}(2)}{2}  [\SU(2)]} \right)^{\azure{7}}_{\azure{[2'],N=2,N_\rho=0}}$};	
	\node (n18) [below=40pt of n5] {\magenta{$\left( \begin{array}{c} \left[E_7\right] \,  1 \, \overset{\mathfrak{su}(2)}{\underset{[N_\text{f}=2]}{2}} [\SU(2)] \\ \otimes \\ \left[E_8\right] \, 1 \end{array}\right)^{18}_{\scriptsize \begin{array}{c} [2],N=1,N_\rho=1 \\  \otimes \\  \text{E-string}	\end{array}}$}};
	\node (n10) [below=80pt of n6] {$\left( \magenta{[E_8] \, 1 \, \overset{\mathfrak{su}(1)}{2} \overset{\mathfrak{su}(2)}{\underset{[N_\text{f}=1]}{2}} [\SU(2)]}   \right)^{\azure{10}}_{\azure{[1^2],N=1,N_\rho=2}}$};	
	\node (n9) [left=280pt of n7] {\ao{$\left( [E_7] \,  1 \,  \overset{\mathfrak{su}(2)}{\underset{[N_\text{f}=3]}{2}} \overset{\mathfrak{su}(1)}{2}\right)^{9}$}};	
	\node (n19) [below=40pt of n8] {$\left([E_8] \, 1 \otimes [E_8] \, 122\right)^{19}$};								
	\node (n21) [below=75pt of n23] {$\left([E_8] \, 12 \otimes [E_8] \, 12\right)^{21}$};	
	\node (nphant8) [below=60pt of n7] {};	
	\node (n11) [below=60pt of n18] {$\left(\magenta{ [E_7] \,  1 \, \overset{\mathfrak{su}(2)}{\underset{[N_\text{f}=2]}{2}} \overset{\mathfrak{su}(2)}{2}  [\SU(2)] }\right)^{\azure{11}}_{\azure{[2],N=2,N_\rho=1}}$};
	\node (n20) [below=60pt of n10] {\magenta{$\left(\begin{array}{c} [E_8] \, 1 \, \overset{\mathfrak{su}(1)}{2} \overset{\mathfrak{su}(2)}{\underset{[N_\text{f}=1]}{2}} [\SU(2)] \\ \otimes \\ \left[E_8\right] \, 1 \end{array}\right)^{20}_{\scriptsize \begin{array}{c} [1^2],N=1,N_\rho=2 \\  \otimes \\  \text{E-string}	\end{array}}$}};	
	\node (n12) [below=75pt of n19] {$\left([E_8] \, 1222\right)^{12}$};	
	\node (nphant9) [below=40pt of nphant8] {};	
	\node (n13) [below=90pt of n9] {\ao{$\left([E_8] \,  1 \, \overset{\mathfrak{su}(1)}{2} \overset{\mathfrak{su}(2)}{\underset{[N_\text{f}=2]}{2}} \overset{\mathfrak{su}(1)}{2}\right)^{13}$}};	
	\node (nphant10) [below=20pt of nphant9] {};	
	\node (n14) [below=45pt of n20] {$\left(\magenta{[E_8] \,  1 \, \overset{\mathfrak{su}(1)}{2} \overset{\mathfrak{su}(2)}{\underset{[N_\text{f}=1]}{2}} \overset{\mathfrak{su}(2)}{2}  [\SU(2)]}\right)^{\azure{14}}_{\azure{[1^2],N=2,N_\rho=2}}$};			
	
	\draw [arrow] (n2)--(n1) node [midway, above left] {} node[midway, left] {\red{$\mathfrak{e}_8$}};
	\draw [arrow] (n15)--(n2) node [midway, above left] {} node[midway, above left] {\red{$2\mathfrak{e}_8$}};
	\draw [arrow] (n22)--(n15) node [midway, above left] {} node[midway, above left] {\red{$3\mathfrak{e}_8$}};			
	\draw [arrow] (n24)--(n22) node [midway, above left] {} node[midway, above left] {\red{$4\mathfrak{e}_8$}};	
	\draw [arrow] (n14)--(n11) node [midway, above left] {} node[midway, above left] {\azure{$\mathfrak{e}_8$}};
	\draw [arrow] (n14)--(n20) node [midway, above left] {} node[midway, above left] {\orange{$A_1$}};	
	\draw [arrow] (n11)--(n7) node [midway, above left] {} node[midway, above left] {\azure{$\mathfrak{e}_7$}};	
	\draw [arrow] (n11)--(n18) node [midway, above left] {} node[midway, above left] {\orange{$A_1$}};	
	\draw [arrow] (n11)--(n10) node [midway, above left] {} node[midway, above left=5pt and 0pt] {\azure{$A_1$}};
	\draw [arrow] (n11)--(n9) node [midway, above left] {} node[midway, below right=5pt] {\ao{$\mathfrak{a}_1$}};			
	\draw [arrow] (n20)--(n18) node [midway, above left] {} node[midway, above right=13pt and 10pt] {$\mathfrak{e}_8$};		
	\draw [arrow] (n20)--(n10) node [midway, above left] {} node[midway, above left=13pt and 0pt] {\red{$\mathfrak{e}_8$}};
	\draw [arrow] (n7)--(n16) node [midway, above left] {} node[midway, above left] {\orange{$A_1$}};	
	\draw [arrow] (n7)--(n5) node [midway, above left] {} node[midway, above left=10pt and 5pt] {\azure{$\mathfrak{d}_8$}};
	\draw [arrow] (n18)--(n16) node [midway, above left] {} node[midway, above right=5pt and 5pt] {$\mathfrak{e}_7$};		
	\draw [arrow] (n18)--(n5) node [midway, above left] {} node[midway, above right=7pt and 2pt] {\red{$\mathfrak{e}_8$}};	
	\draw [arrow] (n10)--(n5) node [midway, above left] {} node[midway, above left] {\azure{$\mathfrak{e}_8$}};	
	\draw [arrow] (n16)--(n3) node [midway, above left] {} node[midway, above left] {\red{$\mathfrak{e}_8$}};	
	\draw [arrow] (n5)--(n3) node [midway, above left] {} node[midway, above left] {\azure{$\mathfrak{e}_7$}};											
	\draw [arrow] (n3)--(n2) node [midway, above left] {} node[midway, above left] {$\mathfrak{e}_7$};
	\draw [arrow] (n4)--(n15) node [midway, above left] {} node[midway, below right] {\nnyellow{$A_1$}};
	\draw [arrow] (n8)--(n17) node [midway, above left] {} node[midway, above left] {\nnyellow{$m$}};
	\draw [arrow] (n17)--(n4) node [midway, above left] {} node[midway, above left] {\red{$\mathfrak{e}_8$}};
	\draw [arrow] (n17)--(n22) node [midway, above left] {} node[midway, left] {\nnyellow{$A_1$}};
	\draw [arrow] (n14)--(n13) node [midway, above left] {} node[midway, above right] {\ao{$\mathfrak{a}_1$}};
	\draw [arrow] (n13)--(n9) node [midway, above left] {} node[midway, above left] {$\mathfrak{e}_8$};
	\draw [arrow] (n9)--(n6) node [midway, above left] {} node[midway, above left] {$\mathfrak{e}_7$};
	\draw [arrow] (n13)--(n12) node [midway, above left] {} node[midway, below left] {\ao{$\mathfrak{a}_1$}};
	\draw [arrow] (n12)--(n19) node [midway, above left] {} node[midway, above left] {\nnyellow{$m$}};
	\draw [arrow] (n19)--(n23) node [midway, above left] {} node[midway, below left] {\nnyellow{$m$}};
	\draw [arrow] (n20)--(n19) node [midway, above left] {} node[midway, below left=-20pt and 30pt] {\ao{$\mathfrak{g}_2$}};
	\draw [arrow] (n12)--(n21) node [midway, above left] {} node[midway, below left] {\nnyellow{$A_1$}};
	\draw [arrow] (n21)--(n23) node [midway, above left] {} node[midway, above left] {\nnyellow{$2A_1$}};
	\draw [arrow] (n7)--(n6) node [midway, above left] {} node[midway, above left=30pt and 70pt] {\ao{$\mathfrak{a}_1$}};	
	\draw [arrow] (n18)--(n17) node [midway, above left] {} node[midway,above right=-15pt] {\ao{$\mathfrak{b}_3$}};
	\draw [arrow] (n10)--(n8) node [midway, above left] {} node[midway, above left=5pt and 5pt] {\ao{$\mathfrak{g}_2$}};
	\draw [arrow] (n9)--(n8) node [midway, above left] {} node[midway, below left] {$\mathfrak{g}_2$};
	\draw [arrow] (n19)--(n8) node [midway, above left] {} node[midway, above left] {\red{$\mathfrak{e}_8$}};			
	\draw [arrow] (n6)--(n4) node [midway, above left] {} node[midway, above left] {$\mathfrak{b}_9$};	\draw [arrow] (n23)--(n17) node [midway, above left] {} node[midway, above left] {\red{2$\mathfrak{e}_8$}};			
	\draw [arrow] (n23)--(n24) node [midway, above left] {} node[midway, above left] {\nnyellow{$2A_1$}};						
	\draw [arrow] (n5)--(n4) node [midway, above left] {} node[midway, above right] {\ao{$\mathfrak{b}_3$}};	
	\draw [arrow] (n16)--(n15) node [midway, above left] {} node[midway, above right] {$\mathfrak{d}_8$};	

\end{tikzpicture}
		\caption{Output of the algorithm of \cite{Bourget:2023dkj,Bourget:2024mgn} for $k=2$ and $P=N+N_\rho=2+2=4$.  The $a$ anomaly decreases from \emph{bottom} to \emph{top} along any allowed RG flow (oriented path), i.e. the IR lies at the \emph{top} here.  Figure \ref{fig:subposet6d} embeds into this one in the obvious way.}
\label{fig:outputk=2}
\end{figure}
\end{landscape}

\clearpage
\small
\bibliography{main}
\bibliographystyle{at}

\end{document}

%% file: k2_grass.pdf_tex
\begingroup%
  \makeatletter%
  \providecommand\color[2][]{%
    \errmessage{(Inkscape) Color is used for the text in Inkscape, but the package 'color.sty' is not loaded}%
    \renewcommand\color[2][]{}%
  }%
  \providecommand\transparent[1]{%
    \errmessage{(Inkscape) Transparency is used (non-zero) for the text in Inkscape, but the package 'transparent.sty' is not loaded}%
    \renewcommand\transparent[1]{}%
  }%
  \providecommand\rotatebox[2]{#2}%
  \newcommand*\fsize{\dimexpr\f@size pt\relax}%
  \newcommand*\lineheight[1]{\fontsize{\fsize}{#1\fsize}\selectfont}%
  \ifx\svgwidth\undefined%
    \setlength{\unitlength}{166.03812175bp}%
    \ifx\svgscale\undefined%
      \relax%
    \else%
      \setlength{\unitlength}{\unitlength * \real{\svgscale}}%
    \fi%
  \else%
    \setlength{\unitlength}{\svgwidth}%
  \fi%
  \global\let\svgwidth\undefined%
  \global\let\svgscale\undefined%
  \makeatother%
  \begin{picture}(1,0.96571914)%
    \lineheight{1}%
    \setlength\tabcolsep{0pt}%
    \put(0,0){\includegraphics[width=\unitlength,page=1]{k2_grass.pdf}}%
    \put(0.61224845,0.94718049){\color[rgb]{0,0,0}\makebox(0,0)[lt]{\lineheight{1.25}\smash{\begin{tabular}[t]{l}${\scriptstyle([1^2],N)}$\end{tabular}}}}%
    \put(0.61593456,0.7137514){\color[rgb]{0,0,0}\makebox(0,0)[lt]{\lineheight{1.25}\smash{\begin{tabular}[t]{l}${\scriptstyle([2],N)}$\end{tabular}}}}%
    \put(0.61967536,0.4968269){\color[rgb]{0,0,0}\makebox(0,0)[lt]{\lineheight{1.25}\smash{\begin{tabular}[t]{l}${\scriptstyle([2^\prime],N)}$\end{tabular}}}}%
    \put(0.40469806,0.49598783){\color[rgb]{0,0,0}\makebox(0,0)[lt]{\lineheight{1.25}\smash{\begin{tabular}[t]{l}${\scriptstyle([2],N-1)}$\end{tabular}}}}%
    \put(0.39633412,0.72692082){\color[rgb]{0,0,0}\makebox(0,0)[lt]{\lineheight{1.25}\smash{\begin{tabular}[t]{l}${\scriptstyle([1^2],N-1)}$\end{tabular}}}}%
    \put(0.39812754,0.28244757){\color[rgb]{0,0,0}\makebox(0,0)[lt]{\lineheight{1.25}\smash{\begin{tabular}[t]{l}${\scriptstyle([2^\prime],N-1)}$\end{tabular}}}}%
    \put(0.13881045,0.50635418){\color[rgb]{0,0,0}\makebox(0,0)[lt]{\lineheight{1.25}\smash{\begin{tabular}[t]{l}${\scriptstyle([1^2],N-2)}$\end{tabular}}}}%
    \put(0.02546565,0.06604428){\color[rgb]{0,0,0}\makebox(0,0)[lt]{\lineheight{1.25}\smash{\begin{tabular}[t]{l}${\scriptstyle([2^\prime],N-2)}$\end{tabular}}}}%
    \put(0.03351177,0.27875066){\color[rgb]{0,0,0}\makebox(0,0)[lt]{\lineheight{1.25}\smash{\begin{tabular}[t]{l}${\scriptstyle([2],N-2)}$\end{tabular}}}}%
    \put(0.60861758,0.8133275){\color[rgb]{0,0,0}\makebox(0,0)[lt]{\lineheight{1.25}\smash{\begin{tabular}[t]{l}$\mathfrak{e}_8$\end{tabular}}}}%
    \put(0.61412993,0.59276093){\color[rgb]{0,0,0}\makebox(0,0)[lt]{\lineheight{1.25}\smash{\begin{tabular}[t]{l}$\mathfrak{e}_7$\end{tabular}}}}%
    \put(0.47089252,0.44949708){\color[rgb]{0,0,0}\makebox(0,0)[lt]{\lineheight{1.25}\smash{\begin{tabular}[t]{l}$\mathfrak{d}_8$\end{tabular}}}}%
    \put(0.25420202,0.44351686){\color[rgb]{0,0,0}\makebox(0,0)[lt]{\lineheight{1.25}\smash{\begin{tabular}[t]{l}$A_1$\end{tabular}}}}%
    \put(0,0){\includegraphics[width=\unitlength,page=2]{k2_grass.pdf}}%
    \put(0.25777181,0.22549377){\color[rgb]{0,0,0}\makebox(0,0)[lt]{\lineheight{1.25}\smash{\begin{tabular}[t]{l}$\mathfrak{d}_8$\end{tabular}}}}%
    \put(0,0){\includegraphics[width=\unitlength,page=3]{k2_grass.pdf}}%
    \put(0.46492898,0.66718341){\color[rgb]{0,0,0}\makebox(0,0)[lt]{\lineheight{1.25}\smash{\begin{tabular}[t]{l}$A_1$\end{tabular}}}}%
    \put(0.3904964,0.59150947){\color[rgb]{0,0,0}\makebox(0,0)[lt]{\lineheight{1.25}\smash{\begin{tabular}[t]{l}$\mathfrak{e}_8$\end{tabular}}}}%
    \put(0.39046345,0.36909438){\color[rgb]{0,0,0}\makebox(0,0)[lt]{\lineheight{1.25}\smash{\begin{tabular}[t]{l}$\mathfrak{e}_7$\end{tabular}}}}%
    \put(0.12553827,0.36414601){\color[rgb]{0,0,0}\makebox(0,0)[lt]{\lineheight{1.25}\smash{\begin{tabular}[t]{l}$\mathfrak{e}_8$\end{tabular}}}}%
    \put(0.1273537,0.15467031){\color[rgb]{0,0,0}\makebox(0,0)[lt]{\lineheight{1.25}\smash{\begin{tabular}[t]{l}$\mathfrak{e}_7$\end{tabular}}}}%
    \put(0,0){\includegraphics[width=\unitlength,page=4]{k2_grass.pdf}}%
    \put(0.08054658,0.22263424){\color[rgb]{0,0,0}\makebox(0,0)[lt]{\lineheight{1.25}\smash{\begin{tabular}[t]{l}$A_1$\end{tabular}}}}%
    \put(0.08411637,0.00461114){\color[rgb]{0,0,0}\makebox(0,0)[lt]{\lineheight{1.25}\smash{\begin{tabular}[t]{l}$\mathfrak{d}_8$\end{tabular}}}}%
    \put(0,0){\includegraphics[width=\unitlength,page=5]{k2_grass.pdf}}%
    \put(0.66555875,0.88491801){\color[rgb]{0,0,0}\makebox(0,0)[lt]{\lineheight{1.25}\smash{\begin{tabular}[t]{l}$A_1$\end{tabular}}}}%
    \put(0.66912815,0.66689491){\color[rgb]{0,0,0}\makebox(0,0)[lt]{\lineheight{1.25}\smash{\begin{tabular}[t]{l}$\mathfrak{d}_8$\end{tabular}}}}%
  \end{picture}%
\endgroup%